\documentclass[conference,10pt]{IEEEtran}
\IEEEoverridecommandlockouts

\usepackage[OT1]{fontenc} 

\usepackage{amsmath}
\usepackage{amsfonts}
\usepackage{amssymb}
\usepackage{amsthm}
\usepackage{isomath}
\usepackage{graphicx}
\usepackage{color}
\usepackage{stmaryrd}
\usepackage{tikz}
\usepackage{ulem}
\usepackage{ebproof}
\usepackage{lineno}
\usepackage{interval}
\intervalconfig{soft open fences}

\usetikzlibrary{positioning,automata}

\renewcommand{\>}{\rangle}
\newcommand{\<}{\langle}
\newcommand{\bra}[1]{\langle #1 \vert}
\newcommand{\ket}[1]{\vert #1 \rangle}

\newcommand{\HH}{\mathcal{H}}
\newcommand{\DD}{\mathcal{D}}
\newcommand{\cS}{\mathcal{S}}

\newcommand{\LL}{\mathcal{L}}
\newcommand{\QO}{\mathcal{QO}}
\newcommand{\QC}{\mathcal{QC}}
\newcommand{\EE}{\mathcal{E}}

\newcommand{\sem}[1]{\llbracket #1 \rrbracket}

\newcommand{\NN}{\mathbb{N}}
\newcommand{\RR}{\mathbb{R}}

\newcommand{\E}{\mathbb{E}}

\newcommand{\guard}{\Box}
\newcommand{\qqRHL}{\mathsf{qOTL}}
\newcommand{\pqRHL}{\mathsf{pqRHL}}
\newcommand{\rqPD}{\mathsf{rqPD}}
\newcommand{\eRHL}{\mathsf{eRHL}}
\newcommand{\eqdef}{\triangleq}
\newcommand{\AST}{\mathsf{AST}}
\newcommand{\Pos}{\mathsf{Pos}}
\newcommand{\PosI}{\mathsf{Pos}^{\infty}}
\newcommand{\Herm}{\mathsf{Herm}}

\newcommand{\TD}{\mathsf{TD}}
\newcommand{\var}{\mathit{var}}
\newcommand{\symeq}{=_{\text{sym}}}

\newcommand{\abs}[1]{| #1 |}

\newcommand{\qprogs}{ \textbf{qProgs} }
\newcommand{\qvar}{\textbf{qVar}}

\DeclareMathOperator{\tr}{tr}
\DeclareMathOperator{\spanv}{span}

\DeclareMathOperator{\supp}{supp}

\newcommand{\skp}{\mathbf{skip}}
\newcommand{\ifb}{\mathbf{if}}
\newcommand{\ife}{\mathbf{fi}}
\newcommand{\while}{\mathbf{while}}
\newcommand{\wdo}{\mathbf{do}}
\newcommand{\wod}{\mathbf{od}}

\newcommand{\oq}{\overline{q}}

\newcommand{\triple}[3]{\{#1\} \ #2 \ \{#3\}}
\newcommand{\rtriple}[4]{\triple{#1}{#2 \sim #3}{#4}}
\newcommand{\varsinone}{\langle1\rangle}
\newcommand{\varsintwo}{\langle2\rangle}
\newcommand{\geqlow}{\sqsupseteq}
\newcommand{\leqlow}{\sqsubseteq}
\newcommand{\swap}{\mathsf{SWAP}}
\newcommand{\halfI}{\frac{I}{2}}

\newcommand{\qWhile}{\mathsf{qWhile}}

\newcommand{\whileloop}{\while\ M[\oq]=1\ \wdo\ S\ \wod}
\newcommand{\ifmeasure}{\ifb\ (\guard m\cdot M[\oq] = m \to S_m)\ \ife}

\usepackage{mathtools}
\DeclarePairedDelimiter\rbra{\lparen}{\rparen}
\DeclarePairedDelimiter\sbra{\lbrack}{\rbrack}

\ifCLASSOPTIONcompsoc
\usepackage[nocompress]{cite}
\else
\usepackage{cite}
\fi

\makeatletter
\let\NAT@parse\undefined
\makeatother
\usepackage[hidelinks]{hyperref} % Must be at the end
\usepackage{orcidlink}

\usepackage{cleveref}

\theoremstyle{plain} % default
\newtheorem{theorem}{Theorem}[section]

\newtheorem{definition}[theorem]{Definition}
\newtheorem*{definition*}{Definition}
\newtheorem{proposition}[theorem]{Proposition}

\newtheorem{lemma}[theorem]{Lemma}
\newtheorem*{lemma*}{Lemma}

\allowdisplaybreaks[1]

\begin{document}
% \linenumbers

\title{Complete Quantum Relational Hoare Logics from Optimal Transport Duality}

% \author{%
  % \IEEEauthorblockN{Anonymous Author(s)}}
  
% \author{Li Zhou}
% \affiliation{%
%   \department{Key Laboratory of System Software (Chinese Academy of Sciences) and State Key Laboratory of Computer Science}
%   \institution{Institute of Software, Chinese Academy of Sciences}
%   \country{China}
% }
% \email{zhouli@ios.ac.cn}

  %% \institution is required
  %\country{Germany}

%\subtitle{}  
\author{\IEEEauthorblockN{
    Gilles Barthe\IEEEauthorrefmark{1}\thanks{${}^\ast$Corresponding authors: Gilles Barthe, Li Zhou}\IEEEauthorrefmark{2},
    Minbo Gao\IEEEauthorrefmark{3}, 
    Theo Wang\IEEEauthorrefmark{4}, 
    Li Zhou\IEEEauthorrefmark{1}\IEEEauthorrefmark{5}}
    \IEEEauthorblockA{
    % \IEEEauthorrefmark{1}Corresponding author: Gilles Barthe, Li Zhou
    \IEEEauthorrefmark{2}MPI for Security and Privacy, Germany and IMDEA Software Institute, Spain, 
    \href{mailto:gilles.barthe@mpi-sp.org}{gilles.barthe@mpi-sp.org}\\
    \IEEEauthorrefmark{3}
%    Key Laboratory of System Software (Chinese Academy of Sciences) and State Key Laboratory of Computer Science, \\
    Institute of Software, CAS, China${}^1$\thanks{${}^1$Key Laboratory of System Software (Chinese Academy of Sciences) and State Key Laboratory of Computer Science, Institute of Software, Chinese Academy of Sciences, China} and University of Chinese Academy of Sciences, China, 
    \href{mailto:gaomb@ios.ac.cn}{gaomb@ios.ac.cn} \\
    \IEEEauthorrefmark{4}University of Cambridge, United Kingdom, \href{mailto:tcw57@cam.ac.uk}{tcw57@cam.ac.uk} \\
    \IEEEauthorrefmark{5}
%    Key Laboratory of System Software (Chinese Academy of Sciences) and State Key Laboratory of Computer Science, \\
    Institute of Software, CAS, China${}^1$,
    \href{mailto:zhouli@ios.ac.cn}{zhouli@ios.ac.cn}}}

\maketitle

% \affiliation{%
%   %\position{Position2a}
%   \institution{Institute of Software, Chinese Academy of Sciences, China}             %% \department is recommended
%   %\institution{Tsinghua University}
%   %\city{Beijing}
%   %\country{China}                   %% \country is recommended
% }
% \email{liujy@ios.ac.cn}         %% \email is recommended
% \author{Li Zhou}
% \affiliation{%
%   \institution{Max Planck Institute for Security and Privacy, Germany}            %% \institution is required
%   %\country{Germany}
%  }

\begin{abstract}
We introduce a quantitative relational Hoare logic for quantum
programs. Assertions of the logic range over a new infinitary
extension of positive semidefinite operators. We prove that our logic
is sound, and complete for bounded postconditions and almost surely
terminating programs. Our completeness result is based on a quantum
version of the duality theorem from optimal transport. We also define
a complete embedding into our logic of a relational Hoare logic with
projective assertions.

\end{abstract}

 %\begin{CCSXML}
%<ccs2012>
%<concept>
%<concept_id>10011007.10011006.10011008</concept_id>
%<concept_desc>Software and its engineering~General programming languages</concept_desc>
%<concept_significance>500</concept_significance>
%</concept>
%<concept>
%<concept_id>10003456.10003457.10003521.10003525</concept_id>
%<concept_desc>Social and professional topics~History of programming languages</concept_desc>
%<concept_significance>300</concept_significance>
%</concept>
%</ccs2012>
%\end{CCSXML}

%\ccsdesc[500]{Software and its engineering~General programming languages}
%\ccsdesc[300]{Social and professional topics~History of programming languages}

% \keywords{Quantum programming, quantum weakest precondition,  expected runtime, physical observable, termination, quantum random walk}  %% \keywords are mandatory in final camera-ready submission

\section{Introduction}
Relational Hoare logics are program logics used to reason about
relationships between programs. Typically, their judgments are of the
form $\rtriple{P}{S_1}{S_2}{Q}$, where $S_1$ and $S_2$ are programs,
and $P$ and $Q$ are relational assertions, traditionally known as pre-
and postcondition. In this paper, we consider the setting where $S_1$
and $S_2$ are quantum programs in the pure $\qWhile$ language.
In this setting, it is natural to define validity based on quantum
couplings. Indeed, there exist several proof systems that support a
rich set of proof rules and are sound w.r.t.\, coupling-based notions
of validity~\cite{qRHL_Unruh_2019,barthe_rqpd,qRHLWithExpectations}.
These proof systems have been used to reason about quantum processes
and quantum security. However, the proof-theoretic foundations of
these proof systems remain unexplored. In particular, there is no
prior account of the completeness of these systems. The challenge with
completeness arises from the existential nature of coupling-based
reasoning: validity of a Hoare judgment $\rtriple{P}{S_1}{S_2}{Q}$
asserts the existence of a suitable coupling, called witness coupling,
between (output states of) $S_1$ and $S_2$. Therefore, the completeness of
the proof system is intuitively equivalent to proving that the rules of the proof system suffice to build all valid couplings between two programs. Unfortunately, it seems difficult to establish a direct argument of this kind. One reason is that proof rules are compositional and allow to build couplings that respect the structure of programs, so it seems plausible that the proof rules are incomplete. In this paper, we do not attempt a direct proof of completeness. Rather, we observe that one can achieve completeness by leveraging a duality theorem for quantum couplings.

\subsubsection*{Contributions}
The main contribution of this paper is a complete proof system for
almost surely terminating programs and positive semi-definite (PSD)
assertions. The proof system contains three parts. The first part is a
minimalistic, standard, set of rules---concretely, one left and right
rule for each construct, two-sided rules for skip and sequential
composition, and a rule of consequence w.r.t.\, the usual L\"owner
order $\sqsubseteq$ on assertions. We prove that this set of rules is
complete for \emph{split} postconditions, i.e.\, postconditions of the
form $Q_1\otimes I + I\otimes Q_2$, where $Q_1$ and $Q_2$ are unary
assertions. The proof follows by classic structural induction on
programs---for technical considerations that will be explained later,
the proof also requires that validity be defined using a new variant
of quantum coupling, called partial coupling, of independent interest.
The second part is a new structural rule, called the duality rule.
The validity of the (duality) rule is based on a quantum duality
theorem, akin to the celebrated Kantorovich-Rubinstein duality theorem
for the probabilistic setting. The main benefit of the rule is that it allows to reduce a judgment of the form
$$\rtriple{P}{S_1}{S_2}{Q}$$
to a judgment of the form
$$\rtriple{P}{S_1}{S_2}{Q_1\otimes I - I\otimes Q_2}
     $$ 
where informally $Q_1$ and $Q_2$ are quantified universally over all unary
assertions such that $Q_1\otimes I - I\otimes Q_2 \sqsubseteq
Q$. Therefore, the duality rules allow us to reduce the proof
of a judgment with an arbitrary postcondition to the validity of a
judgment with a split postcondition, for which the standard rules
suffice.
The third part of the logic are two-sided proof rules. These proof
rules are important for the usability of the logic and are present in
prior works, but are not needed for completeness, and will only
be discussed briefly in the paper.

The second contribution of the paper is an alternative interpretation
of our proof system where assertions are drawn from an infinite-valued
generalization of PSD operators. The logic remains sound for all
postconditions and complete for all bounded postconditions---provided
one restricts the (dual) rule to bounded postconditions. However, the
main benefit of this generalization is that it provides a means to
unify projective assertions, used e.g.\, in~\cite{qRHL_Unruh_2019},
and positive semi-definite operators. As an application, we provide a
complete embedding into our logic of a relational Hoare logic with
projective predicates.

Finally, we leverage our completeness theorems to characterize some
properties of interest. We give two characterizations of program
equivalence. The first characterization is based on (finite-valued) positive semi-definite assertions and uses tools from stable quantum optimal transport. The second characterization is based on projective
assertions (and infinite-valued predicates). We also present characterizations of quantum distance measures (trace distance and Wasserstein semi-distance), diamond norm for programs, non-interference and quantum differential privacy. Finally, as a contribution of independent interest, we prove that the recently proposed relational Hoare logic $\eRHL$ for probabilistic programs~\cite{erhl} is complete for all bounded postconditions and AST programs.

\subsection*{Summary of contributions}
In summary, the main contributions of the paper are:
\begin{itemize}
\item a sound and complete relational program logic for quantum
  programs (\Cref{thm:soundness}, \Cref{thm:completeness post cond});
 \item a complete semantic embedding of quantum relational Hoare logics using projective predicates using infinite-valued predicates (\Cref{prop:semanticembedding projector logic});
 \item characterizations of observational equivalence (\Cref{thm:equal rule}), trace distance and diamond norm (\Cref{prop:encodingoftracedistance} and \Cref{thm:completeness diamond norm}), Wasserstein distance (\Cref{thm:completeness-quantum-wasserstein-semi-distance}),
   non-interference (\Cref{thm:completeness-quantum-non-interference}), and quantum differential privacy (\Cref{thm:completeness-quantum-differential-privacy});
 \item a proof of completeness for the $\eRHL$ relational
   program logic for probabilistic programs (\Cref{prop:completeness-erhl}).
 \end{itemize}

%  \tw{Should we add trace distance and diff priv to the contributions?}

%% Amongst all the previous proposals of quantum relational Hoare logics \cite{qRHL_Unruh_2019,qRHLWithExpectations,barthe_rqpd}, there has been no completeness result. This is exactly what we are trying to achieve here, based on the recent approach in probabilistic relational Hoare logics taken by \cite{erhl}.

\section{Notation and Preliminaries}

We assume basic familiarity to quantum computing (see standard textbook \cite{Nielsen_Chuang_2010}) and set the scene with some notations.

\paragraph{Quantum states and maps} Let $\HH$ be a Hilbert space.
We define $\DD(\HH)$ and $\DD^1(\HH)$ to be the set of partial density operators
(i.e.
positive semi-definite (PSD) operators with trace $\leq 1$) and density
operators (i.e.
partial density operators with trace 1) over $\HH$, respectively.
Intuitively, $\DD(\HH)$ represents the subdistributions over pure states in
$\HH$ and $\DD^1(\HH)$ contains only the full distributions.
Furthermore, we write $\QC(\HH)$ and $\QO(\HH)$ for the set of quantum channels
(CPTP maps) and quantum operations (trace-nonincreasing CP maps) over $\HH$.
We use the former to interpret all almost surely terminating quantum programs
and the latter to represent general quantum programs.
Obviously, $\QC(\HH)\subsetneq \QO(\HH)$.
% The state space of a quantum system is modelled by a Hilbert space $\HH$, and elements of $\HH$ are referred to as pure quantum states. Mixed quantum states are probability distributions over pure states. If state $\ket{\psi_i}$ has probability $p_i$, the mixed state can be represented as the positive semidefinite operator $p_i \ket{\psi_i}\bra{\psi_i}$ with trace $1$. These operators are called density operators $\DD^1(\HH)$. If instead we consider subdistributions of pure states, we get the partial density operators $\DD(\HH)$, i.e. a positive operator with trace $\leq 1$.
% Following standard definitions, we define mixed states (which can be seen as a probability distribution over `pure' quantum states in $\HH$)
% we define $\DD(\HH)$, and $\DD^1(\HH)$ to be the set of partial density operators (i.e. positive semi-definite operators with the trace $\leq 1$) and density operators (i.e. partial density operators with the trace 1) over $\HH$, respectively. 

\paragraph{Quantum predicates} We define $\cS(\HH)$ and $\Pos(\HH)$ to be respectively the closed subspaces (equivalently the orthogonal projectors) and the PSD operators on $\HH$. Subspaces can be used as a `discrete' predicate: a state $\rho \in \DD(\HH)$ satisfies $X \in \cS(\HH)$ if $\supp(\rho) \subseteq X$. General PSD operators are used as bounded quantitative predicates: the `extent' to which $\rho$ satisfies $P \in \Pos(\HH)$ is defined to be $\tr(P\rho)$. 
Commonly used predicates in the work include: the `symmetric' predicate $P_{sym}[\HH] = \frac{1}{2}(I + \swap[\HH])$ (we sometimes denote it as $\symeq$) where $\swap[\HH] = \sum_{ij}|ij\>\<ji|$, and parameter $\HH$ is omitted if it is clear from the context; and the `anti-symmetric' predicate $P_{sym}^\bot$, i.e., the complement of the projector $P_{sym}$, $P_{sym}^\bot[\HH] = \frac{1}{2}(I - \swap[\HH])$. Note that both $P_{sym}[\HH]$ and $P_{sym}^\bot[\HH]$ are in $\cS(\HH\otimes\HH)$.

\paragraph{Infinite-valued predicates} In this work, we introduce a novel notion of
possibly infinite-valued quantitative predicates, denoted $\PosI(\HH)$, by allowing
positive operators to have an eigenspace corresponding to eigenvalue $+\infty$.
In other words, any $A\in\PosI(\HH)$ has an eigenvalue decomposition
$\{(\lambda_i,X_i)\}_i$ where $\lambda_i\in \mathbb{R}^{+\infty} \eqdef [0, +\infty]$, the non-zero
eigenspaces $X_i$ are pairwise orthogonal, and $\sum_iX_i = I$.
As a convention, we define $(+\infty) \cdot 0 = 0 \cdot (+\infty) = 0$, $(+\infty) + a = a + (+\infty) = +\infty$
for $a\in \mathbb{R}^{+\infty}$, and $+\infty\le +\infty$.
We now extend the definitions of various operations on PSD operators to
$\PosI(\HH)$.
Firstly, for any $|\psi\>$, the inner product $\<\psi|A|\psi\>$ is defined as
$$\<\psi|A|\psi\> \triangleq \sum_i\lambda_i \<\psi|X_i|\psi\>.$$
This definition allows us to extend all the operations and constructions on PSD operators that this work relies on to the infinite-valued case. For example, the extended L\"owner order is defined by $A_1\sqsubseteq A_2$ if for all $|\psi\>$, $\<\psi|A_1|\psi\> \le \<\psi|A_2|\psi\>$. We refer the reader to the appendix for more details on the supported operations (see \Cref{def:operations-infinite-valued-pre}). 
%[INSERT REF]. 
Finally, for $X\in \cS(\HH)$ and $A\in\PosI(\HH)$, we define $X\mid A \triangleq A + (+\infty\cdot X^\bot) \in\PosI(\HH)$. This will be useful for enforcing assertion-based, projective preconditions in the quantitative setting.

% \tw{Insert ref for appendix}

For compactness reasons, in this paper, we present the technical development of our results in terms of the more general infinite-valued predicates.

\section{Quantum Couplings}
We review basic definitions and theorems of quantum couplings and quantum optimal transport.

\subsection{Basic Definitions and Duality Theorems}
In probability theory, probabilistic couplings are a powerful tool for
reasoning about different ways of correlating two distributions. A
coupling of two distributions $d_1$, $d_2$ is a joint distribution
with $d_1$, $d_2$ as its respective marginals. Quantum couplings are
the quantum analogue of probabilistic couplings; instead of (sub)distributions, they relate (partial) density operators.
\begin{definition}[Quantum Coupling]
\label{def:coupling}
  Let $\rho_1 \in \DD(\HH_1)$ and $\rho_2 \in \DD(\HH_2)$ be two partial density operators. A coupling between $\rho_1$ and $\rho_2$ is a partial density operator $\rho \in \DD(\HH_1\otimes \HH_2)$ such that $\tr_2(\rho) = \rho_1$ and $\tr_1(\rho) = \rho_2$. We write $\rho: \langle \rho_1, \rho_2\rangle$. 
\end{definition}
Strassen's theorem~\cite{strassen1965existence} provides a necessary and sufficient condition for the existence of a coupling with respect to a given relation. Zhou \textit{et al}.~\cite{zhou2018quantumcouplingstrassentheorem} lift Strassen's theorem to the quantum setting. Their theorem relates a quantum lifting (where for any subspace $X$, a lifting $\rho_1 X^{\#} \rho_2$ is witnessed by couplings of the form $\rho: \langle \rho_1, \rho_2\rangle$ such that $\supp (\rho) \subseteq X$) to a universally quantified property that reasons about $\rho_1$ and $\rho_2$ separately. 
Their proof is based on semi-definite programming (SDP), a common technique in quantum computing and information theory. It turns out that the same technique can be generalized to accommodate for a more general, `quantitative' version of liftings, as stated below.
\begin{definition}[Quantum Lifting with Defects]
\label{def:quantum lifting alter}
  Let $\rho_1 \in \DD(\HH_1)$ and $\rho_2 \in \DD(\HH_2)$, and $\epsilon \in \RR^{+\infty}$ be a defect. Let $X\in \Pos(\HH_1 \otimes \HH_2)$. Then $\rho \in \DD(\HH_1\otimes \HH_2)$ is called a witness of the lifting $\rho_1 X^\#_\epsilon \rho_2$ iff 
  \begin{enumerate}
    \item $\rho: \langle \rho_1, \rho_2\rangle$,
    \item $\tr(X\rho)\leq \epsilon$,
  \end{enumerate}
  % \[
  % \rho: \langle \rho_1, \rho_2\rangle, \text{ and } \tr(X\rho)\leq \epsilon
  % \]
\end{definition}
Note that for any subspace $X$, $\rho_1 X^\# \rho_2$ iff $\rho_1 (X^\bot)^\#_0 \rho_2$.
\begin{theorem}[Quantum Strassen's Theorem with Defects]
  \label{thm:quantum strassen defect alter}
  For any $\rho_1 \in \DD(\HH_1)$ and $\rho_2 \in \DD(\HH_2)$ with $\tr(\rho_1) = \tr(\rho_2)$, for any defect $\epsilon \in \RR^{+\infty}$ and for any $X \in \Pos(\HH_1 \otimes \HH_2)$, the following are equivalent:
  \begin{enumerate}
    \item $\rho_1 X^\#_\epsilon \rho_2$;
    \item For any $Y_1 \in \Pos(\HH_1)$ and $Y_2 \in \Pos(\HH_2)$ such that $X \geqlow Y_1 \otimes I_2 - I_1 \otimes Y_2$, it holds that 
    \[
      \tr(Y_1\rho_1) \leq \tr(Y_2\rho_2) + \epsilon
    \]
  \end{enumerate}
\end{theorem}
% \begin{proof}
%   Follows the same technique as \cite{zhou2018quantumcouplingstrassentheorem}.
% \end{proof}
The setting of the primal and dual problems in the proof is essentially the same as in \cite{zhou2018quantumcouplingstrassentheorem, QuantumOptimalTransport_Cole_2023}. 

\subsection{Partial Couplings}
The following fact is a basic consequence of the definition of quantum couplings.
\begin{lemma}[Trace Equivalence]
  \label{lemma:traceequiv}
  Let $\rho: \langle \rho_1, \rho_2 \rangle$. Then, $\tr(\rho) = \tr(\rho_1) = \tr(\rho_2)$.
  % \lz{write as context instead of lemma.}
\end{lemma}

% \begin{proof}
%   Direct from the definition. 
% \end{proof}
It follows that partial density operators can be coupled only if they have the same trace. This basic fact is a limiting factor for coupling-based relational Hoare logics. In particular, it limits our ability to reason about pairs of non-trace-preserving quantum operations (e.g.\ quantum programs with while loops). 
To address this limitation, we draw ideas from \cite{erhl} ($\star$-couplings) and \cite{qRHLWithExpectations} (quantum $\bot$-memories) and introduce the concept of partial couplings (see \Cref{lemma: relation star coupling} for precise relationship between $\star$-couplings and partial couplings).
\begin{definition}[Partial Coupling]
  For $\rho_1\in\DD(\HH_1)$ and $\rho_2\in\DD(\HH_2)$, we say $\rho\in\DD(\HH_1\otimes\HH_2)$ is a partial coupling of $\rho_1$ and $\rho_2$, written $\rho : \<\rho_1,\rho_2\>_p$, if:
  $$\tr_2(\rho) \sqsubseteq \rho_1, \quad\tr_1(\rho)\sqsubseteq \rho_2,\quad \tr(\rho_1) + \tr(\rho_2) \le 1 + \tr(\rho).$$
\end{definition}
The first two inequalities say that the coupling $\rho$ is partial: it represents a correlation between parts of the marginal state $\rho_1$, $\rho_2$, and leaves another part of the states uncorrelated. The last inequality is a requirement on the uncorrelated parts of the marginal states. It can be decomposed into two inequalities:
% only `couples' part of the states. 
% The last inequality draw from the following constraints:
\begin{align*}
  \tr(\rho_1 - \tr_2(\rho)) &\le 1 - \tr(\rho_2) \\
  \tr(\rho_2 - \tr_1(\rho)) &\le 1 - \tr(\rho_1). 
\end{align*}
% They put forward requirements on the uncoupled part: for example, 
Explained using programming language terms, the first inequality says that the probability of the uncorrelated part of the first system, $\tr(\rho_1 - \tr_2(\rho))$, should not exceed the probability of non-termination in the second system, $1 - \tr(\rho_2)$. The meaning of the second inequality can be obtained by symmetry.

% The first inequality states that the probability of the uncoupled part of first system, i.e., $\tr(\rho_1 - \tr_2(\rho))$, should not be greater than the probability that the second system, i.e., $1 - \tr(\rho_2)$

% The first inequality requires that the probability of the uncoupled part of first system, i.e., $\tr(\rho_1 - \tr_2(\rho))$, should not be greater than the probability that second system does not terminate, i.e., $1 - \tr(\rho_2)$. In quantum program semantics, if $\rho_2$ is the current program state, $1-\tr(\rho_2)$ can be interpreted as the probability that the program does not terminate. In $\star$-coupling or quantum $\bot$-memories, $1 - \tr(\rho_2)$ is indeed the the probability of the component on $\star$-state or $\bot$-state. 

% As a remark, this definition of partial couplings bears close resemblance to $\star$-couplings. The precise relationship between partial couplings and $\star$-couplings is studied in Appendix. \lz{App ref}

Obviously, any coupling is a partial coupling, i.e., $\rho :\<\rho_1,\rho_2\>$ implies $\rho :\<\rho_1,\rho_2\>_p$.
In the case where $\rho_1, \rho_2$ are density operators, any partial coupling is also a coupling, i.e., $\rho :\<\rho_1,\rho_2\>_p$ implies $\rho :\<\rho_1,\rho_2\>$ if $\tr(\rho_1) = \tr(\rho_2) = 1$. Partial coupling is preserved under (sub-)convex combination and 
scalar multiplication (see \Cref{lem: scale of partial coupling}).
A variant of duality theorem for partial coupling is established via SDP (see \Cref{thm: strassen partial coupling}).

\section{Quantum Optimal Transport}
\label{sec: QOT}
% Quantum optimal transport generalizes optimal transport to the quantum
% setting. This section reviews basic definitions and states the duality
% theorem, which is key to our completeness result.

%\lz{I call judgment for asserting properties of quantum operations, and correctness formula for asserting properties of quantum programs. change them if needed.}

One of the applications of quantum coupling is to reason about relational properties of quantum states and thus quantum channels and operations. We first review the basic concept of quantum optimal transport and then show how it can be used to characterize the equivalence of quantum channels.

% \lz{introduction to quantum optimal transport}

\subsection{Basic Definitions}
The optimal transport problem~\cite{monge1781memoire} is a classical
optimization problem. Its goal is to minimize the transportation cost
of goods from sources to sinks. The optimal transport problem has a
natural formulation based on probabilistic couplings. In this section,
we review a quantum version of optimal transport.
We mainly follow~\cite{QuantumOptimalTransport_Cole_2023}.
\begin{definition}[Partial Quantum Optimal Transport (c.f. \cite{QuantumOptimalTransport_Cole_2023})]
    For a given cost function $C\in\PosI(\HH_1\otimes\HH_2)$ and two states $\rho_1\in\DD(\HH_1), \rho_2\in\DD(\HH_2)$, the quantum optimal transport 
    $$T_C(\rho_1,\rho_2)\triangleq \min_{\rho:\<\rho_1,\rho_2\>_p}\tr(C\rho),$$
    where $\rho$ is ranging over all partial couplings of $\rho_1$ and $\rho_2$. 
\end{definition}
The minimum can be attained because the set of partial couplings is an non-empty, closed and convex set (see \Cref{prop:partial-coupling}). 
%\lz{Minbo: help check this.} 
Whenever $\rho_1$ and $\rho_2$ are (total) density operators, every partial coupling is a coupling, and therefore, $T_C(\rho_1,\rho_2) = \min_{\rho:\<\rho_1,\rho_2\>}\tr(C\rho)$.

The basic properties of QOT have been systematically studied, see \cite{QuantumOptimalTransport_Cole_2023} for a comprehensive review.
For example, QOT is jointly convex on its input (see \Cref{lem:QOT convexity}).

\subsection{QOT under Data Processing}
The original definition of QOT studies the relationship between quantum \textit{states}. In this work, we go one step further and ask: can QOT be used to represent and evaluate the relationship between quantum \textit{state transformers} (i.e. quantum channels or operations)? To answer this question, we study how QOT evolves `under data processing'.
\begin{definition}
Let $C_i, C_o\in\PosI(\HH_1\otimes\HH_2)$ be input and output cost
functions respectively. We say that a pair of quantum operations
$(\EE_1,\EE_2)$ is monotone w.r.t. $C_i$ and $C_o$ iff
$T_{C_o}(\EE_1(\rho_1),\EE_2(\rho_2)) \le T_{C_i}(\rho_1,\rho_2)$ hold
for all possible inputs $\rho_1\in\DD(\HH_1)$ and
$\rho_2\in\DD(\HH_2)$. 
\end{definition}
It can be shown that it is sufficient to check monotonicity on inputs
$\rho_1\in\DD^1(\HH_1)$ and $\rho_2\in\DD^1(\HH_2)$, see
\Cref{lem:validity quantum operation alter}.  The next proposition
establishes key properties of monotonicity.

\begin{proposition}
\label{prop:monotone basic}
Monotonicity satisfies several desired properties for data processing:
\begin{enumerate}
    \item 
  \label{lem:judgment two side}
  \emph{Backward}.
  $(\EE_1,\EE_2)$ is monotone w.r.t. $(\EE_1^\dag\otimes \EE_2^\dag)(C)$ and $C$. Here, $\EE^\dag$ is the dual of $\EE$, which satisfies $\tr(A\EE(B)) = \tr(\EE^\dag(A)B)$ for all linear operator $A,B$\footnote{For any super-operator $\EE$, its dual $\EE^\dag$ is another super-operator. Whenever $\EE$ is a quantum operation with Kraus operator $\{E_i\}$, then $\EE^\dag$ has Kraus representation $\{E_i^\dag\}$.}.
  \item
  \label{lem:judgment csq}
  \emph{Consequence}.
  Suppose $(\EE_1,\EE_2)$ is monotone w.r.t. $C_i'$ and $C_o'$, and $C_i'\sqsubseteq C_i$, $C_o\sqsubseteq C_o'$, then $(\EE_1,\EE_2)$ is monotone w.r.t. $C_i$ and $C_o$.
  \item 
  \label{lem:judgment seq}
  \emph{Sequential composition}.
  Suppose $(\EE_1,\EE_1')$ is monotone w.r.t. $C_i$ and $C_m$, and $(\EE_2,\EE_2')$ is monotone w.r.t. $C_m$ and $C_o$, then $(\EE_2\circ \EE_1,\EE_2'\circ \EE_1')$ is monotone w.r.t. $C_i$ and $C_o$.
  \end{enumerate}
  Here, $\circ$ is the composition of two quantum operations, i.e., for all $\rho$, $(\EE_1\circ \EE_2)(\rho) \triangleq \EE_1(\EE_2(\rho))$.
\end{proposition}

The (Backward) property asserts that every pair of quantum channels is
monotonic w.r.t.\, an output cost and its pre-image under some form of
relational pre-image. The (Consequence) property states that one can
strengthen the input cost or weaken the output cost in the style of
the rule of consequence. The (Sequential composition) property states
that monotonocity is compositional.
%\gb{move the last clause to the appendix.}

% \begin{lemma}[Sequential composition]
%   \label{lem:judgment seq}
%   Suppose $(\EE_1,\EE_1')$ is monotone w.r.t. $C_i$ and $C_m$, and $(\EE_2,\EE_2')$ is monotone w.r.t. $C_m$ and $C_o$, then $(\EE_2\circ \EE_1,\EE_2'\circ \EE_1')$ is monotone w.r.t. $C_i$ and $C_o$.
%   Here, $\circ$ is the composition of two quantum operations, i.e., for all $\rho$, $(\EE_1\circ \EE_2)(\rho) \triangleq \EE_1(\EE_2(\rho))$.
% \end{lemma}

% \begin{lemma}[Consequence]
%   \label{lem:judgment csq}
%   Suppose $(\EE_1,\EE_2)$ is monotone w.r.t. $C_i'$ and $C_o'$, and $C_i'\sqsubseteq C_i$, $C_o\sqsubseteq C_o'$, then $(\EE_1,\EE_2)$ is monotone w.r.t. $C_i$ and $C_o$.
% \end{lemma}

% \begin{lemma}[Backward]
%   \label{lem:judgment two side}
%   $(\EE_1,\EE_2)$ is monotone w.r.t. $(\EE_1^\dag\otimes \EE_2^\dag)(C)$ and $C$. Here, $\EE^\dag$ is the dual of $\EE$\footnote{For any super-operator $\EE$, its dual $\EE^\dag$ is another super-operator such that $\tr(A\EE(B)) = \tr(\EE^\dag(A)B)$ holds for all linear operator $A,B$. Whenever $\EE$ is a quantum operation with Kraus operator $\{E_i\}$, then $\EE^\dag$ has Kraus representation $\{E_i^\dag\}$.}.
% \end{lemma}

Additionally, our formulation of QOT under data processing allows us to translate the previous duality result about quantum states (\Cref{thm:quantum strassen defect alter}) to the following duality theorem about quantum operations. Specifically, we show that monotonicity w.r.t.\, $C_i$ and $C_o$ is equivalent to monotonicity w.r.t.\, a split
output cost function. The duality theorem is carefully stated to match our assumptions, in particular that $C_o$ is positive. It is also restricted to the case that $C_o$ is finite.
\begin{theorem}[Duality 
%Theorem 
under Data Processing]
  \label{lem:judgment strassen}
  Suppose $\EE_1$ and $\EE_2$ are quantum channels and costs $C_i\in\PosI$ and $C_o\in \Pos$ (i.e., $C_o$ is finite). Then the following statements are equivalent:
  \begin{enumerate}
    \item $(\EE_1,\EE_2)$ is monotone w.r.t. $C_i$ and $C_o$;
    \item for all $(Y_1,Y_2,n)\in\mathcal{Y}$, $(\EE_1,\EE_2)$ is monotone w.r.t. $C_i + nI$ and $Y_1\otimes I + I\otimes (nI - Y_2)$,
    where $\mathcal{Y} \eqdef \{(Y_1, Y_2, n) \ | \ n \in \NN; 0\leqlow Y_1; 0\leqlow Y_2 \leqlow nI; C_o \geqlow Y_1 \otimes I - I \otimes Y_2\}$.
  \end{enumerate}
\end{theorem}
This formalization will be instrumental in reducing arbitrary judgments
to judgments with split postconditions.  %% At first glance, it seems
%% that (2) is not worthwhile since we must check monotonicity for more
%% output costs. The key point is that computing QOT is closely related
%% to the complexity of the cost, e.g., we need to solve it via SDP when
%% the cost is entangled. In contrast, for specific costs such as the
%% output cost in (2), the computation will be much easier as it can be
%% directly solved without SDP. As a consequence, (1) and (2) are a
%% trade-off: directly check the monotonicity for the more complicated
%% costs, or check a series for simpler output costs.

%% \tw{I think we should give some intuition about this big $\mathcal{Y}$ blob.}

\subsection{Characterizing Equivalence}
\label{sec: QOT equal}

The symmetric and anti-symmetric predicates are standard tools used to characterize equivalence of quantum \textit{states}~\cite{qRHL_Unruh_2019,barthe_rqpd}: indeed, two states $\rho_1$, $\rho_2$ are equal iff there a (non-quantitative) lifting of the form $\rho_1 (\symeq)^{\#} \rho_2$. In this section, by lifting this tool to the setting of QOT under data processing, we give a complete characterization of equivalence between quantum \textit{channels}. This is a significant result: as we shall see in \Cref{thm:equal rule}, it directly leads to the first complete characterization of program equivalence in quantum relational Hoare logics only using finite-valued PSD predicates.

% Under data processing, this tool is often used to characterise equivalence between two quantum operations: 

% A very interesting special case of our framework for QOT under data processing is $C_o = C_i = P_{sym}^\perp$. 

% We now study a special case of QOT under data processing, 

% Our framework for QOT under data processing 

% Symmetric and anti-symmetric spaces are one of the standard ways to characterise equivalence of states~\cite{qRHL_Unruh_2019,barthe_rqpd}. While these works use a satisfaction-based approach, i.e., program logics using subspaces as predicates, to establish a completeness theorem for reasoning about equivalent programs, it remains unknown if similar results holds for expectation-based approach. We give an affirmative answer stated as \Cref{lem: channel equivalence} and \Cref{thm:equal rule}.
%, i.e., program logics using PSD or observables as predicates. 

%\subsubsection*{Expectation-based approach}

Our starting point is the instantiation of QOT under data processing with $C_o = C_i = P_{sym}^\perp$, studied in \cite{Bistron_2023,QuantumOptimalTransport_Cole_2023,mullerhermes2022monotonicity,quantumearthmover}.
% The special case of QOT with $C_o=C_i=P_{sym}^\perp$ has attracted considerable attention \cite{Bistron_2023,QuantumOptimalTransport_Cole_2023,mullerhermes2022monotonicity,quantumearthmover}.
For simplicity, we write $T$ instead of $T_{P_{sym}^\bot}$.
It is clear that $T$ encapsulates some notion of equivalence: $T(\rho,\sigma) = 0$ if and only if $\rho = \sigma$, given $\rho,\sigma$ density operators. However, $T$ cannot fully capture equivalence under data processing, because it is not monotone under general quantum channels \cite{mullerhermes2022monotonicity}, i.e. the following does not always hold for every $\mathcal{E}$.
\[
  T(\mathcal{E}(\rho), \mathcal{E}(\sigma)) \leq T(\rho, \sigma)
\]
Indeed, it does hold for tensoring with an arbitrary quantum state~\cite{QuantumOptimalTransport_Cole_2023}, i.e., $\mathcal{E}: \rho \mapsto \rho\otimes\gamma$, but not for the partial trace.
% there are quantum channels $\mathcal{E}$ (e.g. the partial trace map) such that $T(\mathcal{E}(\rho), \mathcal{E}(\sigma)) > T(\rho, \sigma)$ for some $\rho, \sigma$.
% However, this definition is not quite right, because 
% $T$ is monotone under tensoring with an arbitrary quantum state~\cite{QuantumOptimalTransport_Cole_2023}, i.e., 
% $T(\rho\otimes\gamma,\sigma\otimes\gamma)\le T(\rho,\sigma)$.
% However, it is not monotone in general under quantum channels \cite{mullerhermes2022monotonicity}. For instance, one can show that the quantum channel defined by the partial trace is \textit{not} monotone with respect to $T$. 
This makes it difficult to \emph{completely} reason about the equivalence of data processing operations using the current definition of $T$. 

% it is 
%First asked by \cite{Bistron_2023}, later conjectured by \cite{QuantumOptimalTransport_Cole_2023}, and finally negated by \cite{mullerhermes2022monotonicity}, that $T$ is 
% not monotone under general quantum channels~\cite{mullerhermes2022monotonicity}, i.e., there exists counterexample (partial trace is one of such channel) that violate $T(\EE(\rho),\EE(\sigma)) \le T(\rho,\sigma)$.

Fortunately, \cite{mullerhermes2022monotonicity} proposed a stabilized version of $T$, defined by $T_s\triangleq \inf_{\gamma}T(\rho\otimes\gamma,\sigma\otimes\gamma)$ by extending (tensoring) with an arbitrary auxiliary state $\gamma$, which satisfies several desired properties such as joint convexity, and
\begin{itemize}
  % \item (Joint convexity) For 
  % %$\rho_1,\rho_2,\sigma_1,\sigma_2\in\DD^1(\HH)$ and 
  % $\lambda\in[0,1]$,
  %   \begin{align*}
  %     &T_s(\lambda\rho_1 + (1-\lambda)\rho_2,\lambda\sigma_1 + (1-\lambda)\sigma_2)\\
  %     \le\ &\lambda T_s(\rho_1, \sigma_1) + (1-\lambda)T_s(\rho_2, \sigma_2)
  %   \end{align*}
  \item (Invariance under tensor product) $$T_s(\rho\otimes\gamma,\sigma\otimes\gamma) = T_s(\rho,\sigma).$$
  \item  (Monotonicity under data processing) For $\EE\in\QC$,
  $$T_s(\EE(\rho),\EE(\sigma))\le T_s(\rho,\sigma).$$ 
\end{itemize}
Surprisingly, it turns out that $T_s(\rho,\sigma) = T\big(\rho\otimes\halfI,\sigma\otimes\halfI\big)$. The proof is technical and employs techniques like the Haar measure; we leave the details to Appendix \ref{sec:stablized-quantum-ot-cost}, and provide some intuition. Intuitively, this fact can be understood from two perspectives:
1) the quantum marginal problem, such as the monogamy of entanglement~\cite{koashi2004monogamy}, implies that extending the state can yield more couplings and therefore $T_s(\rho,\sigma)\le T(\rho,\sigma)$ and 2)
extending it by a maximally mixed qubit is sufficient to produce all couplings that minimize optimal transport on the cost function $P_{sym}^\bot$, instead of ranging over all $\gamma$. 
These properties give a complete criterion for checking the equivalence of two quantum channels:
\begin{proposition}
  \label{lem: channel equivalence}
    Two quantum channels $\EE_1$ and $\EE_2$ are equivalent if and only if for all density operators $\rho_1,\rho_2$, $T_s(\EE_1(\rho_1),\EE_2(\rho_2))\le T_s(\rho_1,\rho_2).$
\end{proposition}

While \cite{mullerhermes2022monotonicity} already gives a precise characterization of $T_s$ in terms of QOT, as a semi-definite program, we rephrase it as the following duality theorem:
\begin{proposition}[Duality 
%Theorem 
for Stabilized QOT]
  \label{thm:strassen QOT}
  Given $\rho_1,\rho_2\in\DD^1(\HH)$ and $\epsilon\in \mathbb{R}^+$, the following are equivalent:
  \begin{enumerate}
    \item $T_s(\rho_1,\rho_2)\le \epsilon$;
    \item For all $Y_1, Y_2\in \Pos(\HH\otimes\HH_2)$ such that $P_{sym}^\bot[\HH\otimes\HH_2] \ge 2(Y_1\otimes I - I\otimes Y_2)$, it holds that:
    $$\tr(\tr_2(Y_1)\rho_1) \le \tr(\tr_2(Y_2))\rho_2) + \epsilon.$$
  \end{enumerate}
\end{proposition}
This property is crucial for establishing a judgment characterizing program equivalence (see \Cref{thm:equal rule}), as $T_s$ itself cannot be directly encoded within our program logic.
It additionally allows us to use a split postcondition and thus make the judgment completely derivable (\Cref{thm:weakcompleteness}) without first applying the duality rule.

% \tw{Say sth about why we want this theorem: because we want split postconditions.}
% \subsubsection*{Satisfaction-based approach}
% \begin{lemma}
%   \label{lem:channel equal proj}
%   Two quantum channels $\EE_1$ and $\EE_2$ are equivalent if and only if $(\EE_1,\EE_2)$ are monotone w.r.t. $P_{sym}\mid 0$ and $P_{sym}^\bot$.
% \end{lemma}

\section{Quantum Programs}
\label{sec:quantum_programs}

We now present the syntax and semantics of the quantum programs considered in this paper. % We assume that the reader is familiar with basic ideas of quantum computing. If not, we advise the reader to consult the standard textbook \cite{Nielsen_Chuang_2010} or the preliminary sections of quantum programming literature, e.g.  \cite{qGCL_SZ00,Selinger04}.
% (\textit{For convenience of the reader, a brief introduction of quantum theory is also provided in Appendix \ref{sec:preliminaries}}.)

\subsection{Syntax}\label{sec:quantum_while_program}

% In this subsection, we formally define the syntax of quantum programs studied in this paper.

We choose to use the quantum \textbf{while}-language defined in \cite{Ying11,Ying16}. We assume a finite set $\qvar$ of quantum variables and use $q, q_0, q_1, q_2, \dots$ to denote them.
The finite-dimensional state Hilbert space of a quantum variable $q$ is denoted
$\HH_q$. 

A quantum register is a finite sequence of distinct quantum variables. The state space of a quantum register $\oq = q_0 \dots q_n$ is then the tensor product $ \HH_{\oq} = \bigotimes_{i = 0}^{n} \HH_{q_i}.$

% \tw{Maybe change countably infinite to finite -- doesn't matter anyway, but greatly simplifies everything.}\lz{Agree, and necessary to assume finite many variables with finite dimensions since quantum Strassen's theorem only works for the finite-dimensional case. Although there's an extension to the infinite-dimensional case, the statement of quantum Strassen's theorem is somewhat different.}

\begin{definition}[Syntax \cite{Ying11}] The set $\qprogs$ of quantum \textbf{while}-programs is defined by the following syntax:
\begin{align}
  S :: = \ & \skp \mid  S_1; S_2 \mid q := \ket{0} \mid \oq := U [\oq]  \label{progc2}\\
              &\mid \ifb\ (\guard m\cdot M[\oq] = m \to S_m)\ \ife \label{progc3}\\
              &\mid \while\ M[\oq]=1\ \wdo\ S\ \wod \label{progc4}
\end{align}\end{definition}

% A brief explanation of the above program constructs is given as follows.
The constructs $\skp$ and sequential composition $S_1; S_2$ are similar to their counterparts in the classical or probabilistic \textbf{while}-programs. 
The initialization $q := \ket{0}$ sets the quantum register $q$ to the basis state $\ket{0}$. 
The statement $\oq := U [\oq]$ means that unitary transformation $U$ is performed on the quantum register $\oq$. 
The construct in (\ref{progc3}) is a quantum generalization of classical case 
statement. In the execution, measurement $M = \{ M_m \}$ is performed on $\oq$, 
and then a subprogram $S_m$ will be selected according to the 
outcome of the measurement. 
The statement in (\ref{progc4}) is a quantum generalization of \textbf{while}-loop, where 
the measurement $M$ has only two possible outcomes: if the outcome is
$0$, the program terminates, and if the outcome $1$ occurs, the program
executes the loop body $S$ and then continues the loop. % Most of these constructs were already used in our working Examples \ref{ex-coin}, \ref{ex-rw} and \ref{ex_2_3}.  

\subsection{Semantics}
%We now define the denotational semantics of quantum \textbf{while}-programs.
For each quantum program $S$, we write $\mathit{var}(S) \subseteq V$ for the set of all variables $q\in\qvar$ appearing in $S$.
The Hilbert space of program $S$ is the tensor product $$\HH_S = \bigotimes_{q \in \mathit{var}(S)} \HH_{q}.$$ 

% Finally, for our convenience, we further define $$\HHall = \bigotimes_{q \in \qvar} \HH_q.$$

% Let $\DD(\HH)$ be the set of all partial density operators (i.e. positive operators with the trace $\leq 1$) on $\HH$. 
We interpret each program $S$ denotationally as a complete positive trace non-increasing map $\sem{S} \in \QO(\HH_S)$ as follows:

% A state of  program $S$ is then represented by a partial density operator $\rho \in \DD(\HH_S)$. Furthermore, a configuration is defined as a pair $(S,\rho)$ of a program $S$ and a state $\rho$. The operational semantics of quantum programs can be defined as a transition relation between configurations (\textit{see Appendix \ref{sub:operational_semantics_of_quantum_while_programs}}). define the denotational seamntics as follows.
% \begin{definition}[Denotational Semantics]
%   \tw{cite Ying 2009}
%   Let $S$ be a quantum program
  
% \end{definition}
%By the term partial density operator, we mean a positive operator $\rho$ with trace $\mathit{tr}(\rho)\leq 1$. If $\mathit{tr}(\rho)=1$, then $\rho$ is called a density operator, and is (the mathematical representation of) a (mixed) quantum state. 
%  the following denotational semantics can be  derived, and will be extensively used in this paper:
%\gb{I am sure that I have asked the question before, but why not simply give the representation of the denotational semantics, and move the relation to operational semantics in the appendix?}
\begin{definition}[Denotational Semantics  \cite{Ying11}]\label{lem-structural} For any input state $\rho \in \HH_S$, we have: 
\begin{enumerate}
  \item \label{dsem_skp} $\sem{\skp}(\rho) = \rho$;
  \item \label{dsem_init} $\sem{q := \ket{0}}(\rho) = \sum_n|0\rangle_q\langle n|\rho|n\rangle_q\langle 0|$;
  \item \label{dsem_uni} $\sem{\oq := U[\oq]}(\rho) = U_{\oq} \rho U_{\oq}^\dagger$;
  \item \label{dsem_comp} $\sem{S_1; S_2}(\rho) = \sem{S_2}(\sem{S_1}(\rho))$; 
  \item \label{dsem_if} $\sem{\ifb(\guard m\cdot M[\oq] = m \to S_m)\ife}(\rho) \\ = \sum_m \sem{S_m}(M_m \rho M_m^\dagger)$; 
    \item \label{dsem_while} for loop $\while[M, S]\equiv \while\ M[\oq]=1\ \wdo\ S\ \wod$: 
    $$\sem{\while[M, S]}(\rho) = \bigsqcup_{k = 0}^{\infty} \sem{\while^{(k)}[M,S]}(\rho),$$ 
where $\while^{(k)}[M,S]$ is the $k$-fold iteration of the loop $\while$: 
\begin{equation*}\label{iteration}\begin{cases}
  &\while^{(0)}[M,S]  \equiv \textbf{abort}, \\
  &\while^{(k+1)}[M,S] \\
  &\begin{aligned} 
  \equiv\ &\ifb\ M[\oq]  =\  0 \to \skp\\  
  &\guard\qquad \qquad 1\to S; \while^{(k)}[M,S]\ \ife
\end{aligned}\end{cases}\end{equation*}
for $k\geq 0$, $\bigsqcup$ stands for the least upper bound in the CPO of partial density operators with the L\"owner order $\sqsubseteq$ (see \cite{Ying16}, Lemma 3.3.2), 
%{\color{red} 
and $\textbf{abort}$ is a program that never terminates so that $\sem{\textbf{abort}}(\rho) = \bf{0}$ for all $\rho$.
%}
\end{enumerate}\end{definition}
In the special case where $\sem{S} \in \QC(\HH_S)$, we say that $S$ is almost-surely terminating (AST), or simply write $S\in\AST$. 

% Intuitively, a partial density operator $\rho \in \DD(\HH_S)$ can be thought of as a representation of a sub-distribution of pure states in $\HH_S$. %Thus, the sum of partial density operators %in the above lemma is analogous to the sum of sub-distributions of states in probabilistic programming. 

% It immediately follows from the above lemma that the denotational semantics $\sem{S}$ of a quantum program is a (completely positive) super-operator, which is the mathematical formalism of quantum operations or the (discrete-time) dynamics of open quantum systems. 

% \tw{Clean Up this section}

\section{A Quantum Relational Hoare Logic}
We now present $\qqRHL$, a quantum relational Hoare logic similar to \cite{barthe_rqpd} extended with logical variables, and prove its soundness. As we shall see in \Cref{section:completeness}, this extension is crucial to enabling our completeness results.

\subsection{Definition}

In $\qqRHL$, judgments are of the form
\[
  \vdash Z : \rtriple{P}{S_1}{S_2}{Q}
\]
where predicates $P, Q \in\PosI(\HH_{S_1} \otimes \HH_{S_2})$, i.e.,
are infinite-valued positive semi-definite operators over $\HH_{S_1}
\otimes \HH_{S_2}$, parameterized over $Z$, and $S_1, S_2$ are
programs.  Validity of the judgment is defined using partial
couplings.
\begin{definition}[$\qqRHL$ Validity]
  The judgment $\vdash Z: \rtriple{P}{S_1}{S_2}{Q}$ is valid, written
  \[
    \vDash Z : \rtriple{P}{S_1}{S_2}{Q}
  \]
  if for every $z \in Z$, and $\rho \in \DD^1(\HH_{S_1} \otimes \HH_{S_2})$, there exists a partial coupling $\sigma$ for $\langle \sem{S_1}(\tr_2(\rho)), \sem{S_2}(\tr_1(\rho)) \rangle_p$ such that
  \[
  \tr(P_z \rho) \geq \tr(Q_z \sigma).
  \]
  We usually write $P$ (resp $Q$) instead of $P_z$ (resp $Q_z$) when there is no ambiguity.
\end{definition}

Whenever $S_1,S_2$ are AST programs, the partial coupling $\sigma$ is a coupling (see \Cref{lem:validity AST program}), which is consistent or similar to previous works~\cite{qRHL_Unruh_2019,barthe_rqpd}. 
% \lz{help rephrase or decide to remove}
% \lz{Besides, ... see \Cref{lem:validity AST program}}
% By employing \Cref{lem:validity and monotone,lem:validity quantum operation alter}, we immediately have the following variant:
% \begin{lemma}[$\qqRHL$ Validity for AST Programs]
% \label{lem:validity AST program}
%   Suppose $S_1$ and $S_2$ are AST programs. Then $\vDash Z: \rtriple{P}{S_1}{S_2}{Q}$ if and only if for every $z \in Z$, and $\rho \in \DD(\HH_{S_1} \otimes \HH_{S_2})$, there exists a coupling $\sigma : \langle \sem{S_1}(\tr_2(\rho)), \sem{S_2}(\tr_1(\rho)) \rangle$ such that
%   \[
%   \tr(P_z \rho) \geq \tr(Q_z \sigma).
%   \]
% \end{lemma}

Validity can be recast in terms of monotonicity, which allows us to investigate it from a QOT view.
% \gb{We have too many intermediate lemmas.}
\begin{lemma}
  \label{lem:validity and monotone}
  $\vDash Z : \rtriple{P}{S_1}{S_2}{Q}$ if and only if for all $z\in Z$, $(\sem{S_1},\sem{S_2})$ is monotone w.r.t. $P$ and $Q$.
\end{lemma}

Fig. \ref{fig:qqrhlrules} introduces a minimal set of proof rules of our
logic. Our set of proof rules contains so-called one-sided rules for
initialization, unitaries, conditionals and loops. We only show left
rules; there exists a similar right rule for each construct. We note
that the one-sided rules are the obvious counterparts of the usual
rules for quantum Hoare logic~\cite{Ying11}; for instance, the rule
for while loops requires users to provide a loop invariant. Besides,
our proof system features the usual two-sided rules for skip and
sequential compositions. Lastly, our proof system features two
structural rules. The (csq) rule is the rule of consequence; it is
based on L\"owner order. The (duality) rule is an application of the
duality theorem, and is used to reduce postconditions to universally
quantified split postconditions. Note that the rule requires that the
postcondition $Q$ is bounded, i.e.\, $Q\in \Pos$ rather than
$Q\in\Pos^\infty$. In particular, our core set of rules does not 
feature additional two sided-rules. We discuss two-sided rules in
\Cref{sec:twosided}.
\begin{figure*}
  \textbf{Two-sided rules:}\quad
  (skip) \ \ 
  $
    \begin{prooftree}
      \infer0{\vdash Z: \rtriple{P}{\skp}{\skp}{P}}
      \end{prooftree}
  $\qquad
  (seq) \ \ 
  $
    \begin{prooftree}
      \hypo{\vdash Z: \rtriple{P}{S_1}{S'_1}{Q}}
      \hypo{\vdash Z: \rtriple{Q}{S_2}{S'_2}{R}}
      \infer2{\vdash Z: \rtriple{P}{S_1;S_2}{S_1';S_2'}{R}}
    \end{prooftree}
  $
  \\[0.5cm]
  \textbf{One-sided rules:}\quad
  (assign-L)\ \ 
  $
    \begin{prooftree}
      \infer0{\vdash Z: \rtriple{\sum_{ij} (\ket{i}_{q_1 \varsinone} \bra{0})P(\ket{0}_{q \varsinone} \bra{i})}{q:=\ket{0}}{q:=\ket{0}}{P}}
    \end{prooftree}
  $\\[0.5cm]
  \hspace*{2.74cm} (apply-L) \ \
  $
    \begin{prooftree}
      \infer0{\vdash Z: \rtriple{(U \otimes I_2)^\dagger P (U\otimes I_2)}{\bar{q} := U[\bar{q}]}{\skp}{P}}
    \end{prooftree}
  $\\[0.5cm]
  \hspace*{2.74cm} (if-L)\ \ 
  $
    \begin{prooftree}
      \hypo{\forall m. \vdash Z: \rtriple{P_m}{S_m}{\skp}{Q}}
      \infer1{\vdash Z: \triple{\sum_m (M_m \otimes I)^\dagger_m P_m (M_m \otimes I)}{
        \ifmeasure \sim \skp
      }{Q}}
    \end{prooftree}
  $\\[0.5cm]
  \hspace*{2.74cm} (while-L)\ \ 
  $
    \begin{prooftree}
      \hypo{\vdash Z: \rtriple{Q}{S}{\skp}{(M_0 \otimes I)^\dagger P (M_0 \otimes I) + (M_1 \otimes I)^\dagger Q (M_1 \otimes I)}}
      %\hypo{\whileloop \text{ AST}}
      \infer1{\vdash Z: \rtriple{(M_0 \otimes I)^\dagger P (M_0 \otimes I) + (M_1 \otimes I)^\dagger Q (M_1 \otimes I)}{\whileloop}{\skp}{P}}
    \end{prooftree}
  $\\[0.5cm]
  % (seq):
  % \[
  %   \begin{prooftree}
  %     \hypo{\vdash Z: \rtriple{P}{S_1}{\skp}{Q}}
  %     \hypo{\vdash Z: \rtriple{Q}{\skp}{S_2}{R}}
  %     \infer2{\vdash Z: \rtriple{P}{S_1}{S_2}{R}}
  %   \end{prooftree}
  % \]
  \textbf{Structural rule:}\quad
  (csq)\ \ 
  $
    \begin{prooftree}
      \hypo{P \geqlow P'}
      \hypo{\vdash Z: \rtriple{P'}{S_1}{S_2}{Q'}}
      \hypo{Q' \geqlow Q}
      \infer3{\vdash Z: \rtriple{P}{S_1}{S_2}{Q}}
    \end{prooftree}
  $\\[0.5cm]
  \textbf{Logical rule:}\quad
  (duality)\ \ 
  $
    \begin{prooftree}
      \hypo{
        \begin{array}{c}
          \vdash Z, (Y_1,Y_2,n)\in\mathcal{Y} : \rtriple{P + nI}{S_1}{S_2}{Y_1\otimes I + I\otimes (nI - Y_2)} 
          \qquad\qquad\quad S_1, S_2 \in \AST
          \\[0.1em]
        \mbox{ where } 
          \mathcal{Y} \eqdef \{(Y_1, Y_2, n) \ | \ n \in \NN; 0\leqlow Y_1; 0\leqlow Y_2 \leqlow nI; Q \geqlow Y_1 \otimes I - I \otimes Y_2\}
          \qquad\quad Q\in \Pos
          \end{array}
      }
      \infer1{\vdash Z: \rtriple{P}{S_1}{S_2}{Q}}
    \end{prooftree}
  $
  \\
  \caption{Rules for $\qqRHL$}
  \label{fig:qqrhlrules}
\end{figure*}

\subsection{Soundness and Completeness}
\label{section:completeness}
Every derivable judgment is valid. 
\begin{theorem}[Soundness]
  \label{thm:soundness}
If $\vdash Z: \rtriple{P}{S_1}{S_2}{Q}$ then 
  $\vDash Z: \rtriple{P}{S_1}{S_2}{Q}$.
\end{theorem}
%% \begin{proof}
%%   By induction, using ...
%%   \tw{Cite lemmata used.}
%% \end{proof}
Conversely, one can prove completeness for bounded postconditions and
AST programs. The proof is divided into two main steps. First, we
establish completeness result for \textit{split postconditions}, i.e., postconditions of the form $Q_1
\otimes I_2 + I_1 \otimes Q_2$. With this result in place, we can then leverage
duality to derive completeness for all AST programs, and bounded
postconditions.

% We now present the completeness results for $\qqRHL$, and their applications. Following \cite{erhl}, we first prove a weak completeness result on post conditions which are \textit{separated}, i.e.\ of the form $Q_1 \otimes I_2 + I_1 \otimes Q_2$. This is the last piece of the puzzle. Then, using the previous results on QOT, we can reduce stronger completeness results to it. In particular, we show completeness of $\qqRHL$ for all AST programs, for program equivalence, and a weak completeness result for trace distance. 

% we show that this result is sufficiently strong to prove completeness for all AST programs and for proving program 

% we present two somewhat surprising completeness results which follow directly from this. 

% this result is sufficient for us to derive completeness for all AST programs, using Strassen's theorem. Finally, we present two more applications of our weak completeness result: completeness for proving program equivalence, and for upper-bounding the trace distance between the final states. 

The first step towards completeness is to show some form of one-sided
weakest precondition for AST programs.
\begin{lemma}[One-Sided Weakest Preconditions]
  \label{lemma:onesidedweakestpre}
  For every AST program $S$, we have
  \[
  \vdash Z: \rtriple{(\sem{S}^\dagger \otimes I)(Q)}{S}{\skp}{Q}.
  \]
  % where for any CP map $\EE$, $\EE^\dagger$ is its adjoint with respect to the Hilbert-Schmidt inner product.
\end{lemma}
The lemma is proved by structural induction on the program $S$. One
can then lift the results to the case of two programs. 
\begin{lemma}[Two-Sided Weakest Preconditions]
  \label{lemma:twosidedweakestpre}
  For every AST programs $S_1, S_2$, we have
  \[
    \vdash Z: \rtriple{(\sem{S_1}^\dagger \otimes \sem{S_2}^\dagger) (Q)}{S_1}{S_2}{Q}
  \]
\end{lemma}

Now we are ready to give our completeness result.
\begin{theorem}[Completeness for Split Postconditions]
  \label{thm:weakcompleteness}
  For every AST programs $S_1, S_2$, we have:
  % any $Z$, any $Z$-parameterised predicate $P$ and $Z$-parameterised $Q_1 \in \PosI(\HH_{S_1})$ and $Q_2\in \PosI(\HH_{S_2})$, if 
  \[
    \vDash Z: \rtriple{P}{S_1}{S_2}{Q_1 \otimes I + I \otimes Q_2}
  \]
  implies
  \[
   \vdash Z: \rtriple{P}{S_1}{S_2}{Q_1 \otimes I + I \otimes Q_2}
  \]
\end{theorem}
Using the duality theorem, we can then derive that $\qqRHL$ is
complete for all terminating programs with finite postconditions.
\begin{theorem}[Completeness for Terminating Programs]
  \label{thm:completeness post cond}
  For every AST $S_1, S_2$ programs and bounded predicate $Q\in
  \Pos(\HH_{S_1} \otimes \HH_{S_2})$, we have
  \[\vDash Z:
  \rtriple{P}{S_1}{S_2}{Q}\] implies
  \[\vdash Z:
  \rtriple{P}{S_1}{S_2}{Q}\]
\end{theorem}
\begin{proof}
The desired judgment follows from an application of the duality rule
and the provability of:
  \begin{multline*}
    \vdash Z: \rtriple{P}{S_1}{S_2}{Q}
    \iff \\ \vdash Z, (Y_1, Y_2, n) \in \mathcal{Y}: \rtriple{P + nI}{\\ S_1}{S_2}{Y_1 \otimes I + I \otimes (nI - Y_2)}
  \end{multline*}
  where $\mathcal{Y}$ is defined as in \Cref{lem:judgment strassen}
  with $C_i = P$ and $C_o = Q$. Provability of the latter follows from
  completeness for split postconditions.
\end{proof}

\subsection{Two-Sided Rules}\label{sec:twosided}
This part considers two-sided rules. Such rules are not needed for completeness. However, they allow to carry lock-step reasoning about structurally similar programs, and typically lead to simpler and more intuitive derivations. For example, it may be easier to establish the equivalence of two loops using a two-sided loop rule rather than using twice a one-sided loop rule, simply because a two-sided loop rule may use the loop invariant that the two loop bodies preserve state equivalence.
However, it can be challenging to define sound and expressive two-sided proof rules for control-flow constructs. For instance, \cite{barthe_rqpd}
uses two-sided rules that involve measurement conditions and entailment between measurement conditions---where these entailments are proved by semantic means. In this section, we show that these rules remain sound for infinite-valued predicates, and we further show how our formalism yields some proof rules to reason about measurement conditions.
\begin{definition}[Measurement Condition and Entailment, c.f. \cite{barthe_rqpd}]
  Suppose $M = \{M_1,\cdots,M_k\}$ and $N = \{N_1,\cdots,N_k\}$ are two measurements with the same output set.
  We say two states $\rho,\sigma\in\DD$ satisfy the measurement condition $M\approx N$, written $(\rho,\sigma)\vDash M\approx N$, if for all $i$, $\tr(M_i\rho M_i^\dagger) = \tr(N_i\sigma N_i^\dagger)$. 
  
  We further define the entailment relation of two programs $S_1,S_2$, written $\Gamma\stackrel{(S_1,S_2)}{\vDash}\Gamma'$, if for all $\rho,\sigma\in\DD^1$ such that $(\rho,\sigma)\vDash\Gamma$, it holds $(\sem{S_1}(\rho), \sem{S_2}(\sigma)) \vDash \Gamma'$.
\end{definition}

Checking the entailment relation involves the program constructions is highly nontrivial~\cite{barthe_rqpd}. In fact, 
the proposed method in \cite{barthe_rqpd} is based on the semantics of the programs. Here, we give a complete characterization so that checking entailment itself can be done using program logic.
\begin{theorem}
  \label{thm:side-condition equal}
  For AST programs $S_1,S_2$, and measurements $M = \{M_1,\cdots,M_k\}$ and $N = \{N_1,\cdots,N_k\}$, the following are equivalent:
  \begin{enumerate}
    \item $\emptyset\stackrel{(S_1,S_2)}{\vDash}M\approx N$;
    \item $\vDash (Y_1,\cdots,Y_k,Z_1,\cdots,Z_k,n)\in\mathcal{Y}_k: \{nI\}S_1\sim S_2$\\[-0.5cm]
    \begin{multline*}
      \qquad\ \big\{ (\mbox{$\sum_i$}M_i^\dagger Y_i M_i)\otimes I +
      I \otimes \big[nI - (\mbox{$\sum_i$}N_i^\dagger Z_i N_i)\big]\big\}
    \end{multline*}
      where $\mathcal{Y}_k = \{  (Y_1,\cdots,Y_k,Z_1,\cdots,Z_k,n) \mid$\\[-0.6cm]
    \begin{multline*}
      \forall\,i,\ 0\sqsubseteq Y_i,\ 0\sqsubseteq Z_i\sqsubseteq nI, Y_i\otimes I - I\otimes Z_i\sqsubseteq 0,\\
      \forall\,j\neq i,\ Y_i\otimes I - I\otimes Z_j\sqsubseteq I\}.
    \end{multline*}
  \end{enumerate}
\end{theorem}

We further define measurement properties as side conditions to set up two-sided rules for $\ifb$ and $\while$. Our definition unifies Def. 5.4 and 7.2 in \cite{barthe_rqpd} (see \Cref{prop:measurement properties embed}).
\begin{definition}[Measurement Property, c.f. \cite{barthe_rqpd}]
  Define $\Gamma \vDash Z : \{P \} M\approx N \{Q_k\}$
  if for all $\rho,\sigma\in\DD^1$ such that $(\rho,\sigma)\vDash \Gamma$ and $z\in Z$, 
  if $T_P(\rho,\sigma) < +\infty$, then there exist \emph{couplings}
  $\delta_i : \< M_i\rho M_i^\dagger, N_i\sigma N_i^\dagger \>$ for each $i$, such that:
  $$T_P(\rho,\sigma) \ge \sum_i\tr(Q_i\delta_i).$$
\end{definition}

\begin{figure*}
  \textbf{Extra rules:}\quad
  % (assign)\ \ 
  % $
  %   \begin{prooftree}
  %     \infer0{\vdash Z: \rtriple{\sum_{ij} ((\ket{i}_{q \varsinone} \bra{0})\otimes(\ket{j}_{q \varsintwo} \bra{0}))P((\ket{0}_{q \varsinone} \bra{i})\otimes (\ket{0}_{q \varsintwo} \bra{j}))}{q:=\ket{0}}{q:=\ket{0}}{P}}
  %   \end{prooftree}
  % $\\[0.5cm]
  % \hspace*{2.1cm} (apply) \ \
  % $
  %   \begin{prooftree}
  %     \infer0{\vdash Z: \rtriple{(U_{\bar{q} \varsinone} \otimes U'_{\bar{q}' \varsinone})^\dagger P (U_{\bar{q} \varsinone} \otimes U'_{\bar{q}' \varsinone})}{\bar{q} := U[\bar{q}]}{\bar{q} := U'[\bar{q}]}{P}}
  %   \end{prooftree}
  % $\\[0.5cm]
  % \hspace*{2.1cm} 
  (if)\ \
    $
      \begin{prooftree}
        \hypo{\Gamma \vDash Z: \{P\} M\approx M' \{R_k\}}
        \hypo{\forall\,k,\ \vdash Z : \rtriple{R_k}{S_k}{S_k'}{Q}}
        \infer2{\Gamma \vdash Z: \rtriple{P}{\ifb\ (\guard k\cdot M[\oq] = k \to S_k)\ \ife}{\ifb\ (\guard k\cdot M'[\oq] = k \to S_k')\ \ife}{Q}}
        \end{prooftree}
    $\\[0.5cm]
  \hspace*{2.1cm} (while)\ \ 
    $
      \begin{prooftree}
        \hypo{\vDash Z: \{P\} M\approx M' \{Q_0,Q_1\}}
        \hypo{\vdash Z : \rtriple{Q_1}{S}{S'}{P}}
        \infer2{\vdash Z: \rtriple{P}{\while\ M[\oq]=1\ \wdo\ S\ \wod}{\while\ M'[\oq]=1\ \wdo\ S'\ \wod}{Q_0}}
        \end{prooftree}
    $\\[0.5cm]
  \hspace*{2.1cm} (seq+)\ \ 
    $
    \begin{prooftree}
      \hypo{\Gamma\vdash Z : \rtriple{P}{S_1}{S_1'}{Q}}
      \hypo{\Gamma'\vdash Z : \rtriple{Q}{S_2}{S_2'}{R}}
      \hypo{\Gamma\stackrel{(S_1,S_1')}{\vDash}\Gamma'}
      \infer3{\Gamma\vdash Z : \rtriple{P}{S_1;S_2}{S_1';S_2}{R}}
    \end{prooftree}
    $
    \\
    \caption{Extra two-side rules for $\qqRHL$.}
    \label{fig:qqrhl_extra}
\end{figure*}

We can now defined two-sided rules in Fig. \ref{fig:qqrhl_extra} and prove their soundness.
\begin{theorem}[Soundness of Two-Sided Rules]
  \label{thm:soundness extra}
  The extra rules for $\qqRHL$ in Fig. \ref{fig:qqrhl_extra} are sound regarding the notion of validity.
\end{theorem}

\section{Infinite-Valued and Projective Predicates}

It might seem curious why we chose to present everything in terms of infinite-valued predicates. What exactly do they buy us? In this section, we answer this question by showcasing the expressiveness of infinite-valued predicates, by showing how it enables a complete semantic embedding of projector-based quantum relational Hoare logics in $\qqRHL$. In the context of $\qqRHL$, this gives us complete characterisations of non-trivial properties like program equivalence, for free. In the wider field of quantum program logics, this gives us a general way of unifying the two types of predicates (projective and quantitative) in the same logic.
% It might seem curious why we present everything in terms of infinite-value predicates. 

\subsection{Projective Predicates}
Our logic, $\qqRHL$, follows a quantitative paradigm: we use (generalised) positive semi-definite operators as predicates, and reason about the `extent' to which quantum states satisfy those predicates by the expectations of the operators over the states. 
The alternative approach, followed by \cite{qRHL_Unruh_2019,approx_rel_reasoning_quantum_programs} and parts of \cite{barthe_rqpd}, uses subspaces (or equivalently projectors) as assertions: a state $\rho$ satisfies $X \in \cS(\HH)$ if $\supp(\rho) \subseteq X$. In the setting of quantum relational Hoare logics, this corresponds to a notion of validity as follows:
\begin{definition}[Logic for Projective Predicates]
  We write $\vDash_{\pqRHL}: \rtriple{X}{S_1}{S_2}{Y}$, where $X, Y \in \cS(\HH_1 \otimes \HH_2)$, if for any initial state $\rho$ with $\supp (\rho) \subseteq X$, there exists a coupling $\sigma : \langle \sem{S_1}(\tr_2(\rho)), \sem{S_2}(\tr_1(\rho)) \rangle$ such that $\supp (\sigma) \subseteq Y$. 
\end{definition}
This formulation has several advantages compared to its quantitative counterpart: the resulting logic often has simpler rules, and several non-trivial properties have much simpler formulations. Crucially, this is possible only because $\pqRHL$ allows one to enforce \textit{projective preconditions}, i.e., membership of the initial state in a particular subspace. For example, equivalence between two programs $S_1$, $S_2$, or, equivalently, the property that $\forall \rho \in \DD(\HH). \sem{S_1}(\rho) = \sem{S_2}(\rho)$ can be expressed as the judgement 
\[
\vDash_{\pqRHL} \rtriple{\symeq}{S_1}{S_2}{\symeq}
\]
where importantly, the precondition $\symeq$ forces the arbitrary initial state $\rho$ to satisfy $\tr_2(\rho) = \tr_1(\rho)$. Unfortunately, similar constraints on the initial state/coupling are not known to be expressible in the bounded quantitative case. As a consequence, it takes much more effort to characterise properties like program equivalence using only positive semi-definite operators as predicates, as we will later show in \cref{thm:equal rule}. 

% it turned out to be much more difficult to characterise properties like equivalence using only finite-valued PSD predicates, and, as shown in \cref{thm:equal rule}, 

% In logics based on (bounded) quantitative predicates (finite-valued PSD operators), while it is possible to enforce a projective \textit{post}condition (i.e. that the final coupling be included in some subspace) by setting the precondition to $0$, it is not known to be possible to enforce projective \textit{pre}conditions. Thus, similar characterisations of non-trivial properties are not known to be possible.

\subsection{Enforcing Projective Preconditions Using Infinite-Valued Predicates}
% But what if we had some way of enforcing projective preconditions in quantitative logics? It turns 
It turns out that things are different when we allow infinite-valued predicates. Consider a $\qqRHL$ judgement of the form $\vDash \rtriple{P}{S_1}{S_2}{Q}$ where $P$ is of the form $X|A = \infty \cdot X^\perp + A$. For any initial coupling $\rho$, if $\supp(\rho) \subseteq X$, then the judgement acts as if $P=A$; if $\supp(\rho) \nsubseteq X$ however, then the judgement is rendered trivially true. In other words, $X|A$ is the same thing as a normal quantitative precondition $A$ constrained by a projective precondition $X$! A direct consequence of this insight is a semantic embedding of $\pqRHL$ in $\qqRHL$ as follows:
\begin{proposition}
  \label{prop:semanticembedding projector logic}
  For AST programs $S_1$, $S_2$, and $X,Y\in\cS(\HH_{S_1}\otimes\HH_{S_2})$ be projectors. 
  The following holds:
  \[
  \vDash \rtriple{X\mid 0}{S_1}{S_2}{Y^\bot} \iff\ \vDash_{\pqRHL} \rtriple{X}{S_1}{S_2}{Y}.
  \]
\end{proposition}
Noting that the postcondition here is bounded, by the completeness theorem (\Cref{thm:completeness post cond}), we directly obtain a complete embedding of the projector logic into our logic for AST programs, as shown in the following theorem.
\begin{theorem}
  For AST programs $S_1$, $S_2$, and $X,Y\in\cS(\HH_{S_1}\otimes\HH_{S_2})$ be projectors, we can completely characterise any property defined by the judgement $\vDash_{\pqRHL} \rtriple{X}{S_1}{S_2}{Y}$ in $\qqRHL$.
  %the following holds:
  % \[
  % \vDash_{\pqRHL} \rtriple{X}{S_1}{S_2}{Y} \in \mathcal{Y} \iff \vdash (Y_1, Y_2, n) \rtriple{X|nI}{S_1}{S_2}{Y_1 \otimes I + I\otimes (nI - Y_2)}
  % \]
  % where $\mathcal{Y}$ is defined as in \Cref{lem:judgment strassen} with $C_i = X|0$ and $C_o = Y$.
\end{theorem}
%\tw{I'm purposefully not giving the explicit characterisation to save space}
Therefore, as a corollary, we obtain a complete characterisation of program equivalence for AST programs -- this has not been achieved so far in existing quantitative quantum relational Hoare logics \cite{barthe_rqpd,qRHLWithExpectations}.

\subsection{Wider Consequences}
Infinite-valued predicates provide a general recipe for unifying quantitative and projective quantum predicates.
We have seen how it works in the relational case; the same approach also works in the non-relational case. Indeed, if we define a quantum Hoare logic using projective predicates:
\begin{definition}
  Let $S$ be a $\qWhile$ program and $X, Y \in \cS(\HH_S)$. We define $\vDash_{\mathsf{pqHL}} \triple{X}{S}{Y}$ to mean $\forall \rho \in \DD(\HH_S). \ \supp(\rho) \subseteq X \implies \supp \sem{S}(\rho) \subseteq Y$.
\end{definition}
A logic similar to \cite{Ying11} but using infinite-valued quantitative predicates can also be defined:
\begin{definition}
  Let $S$ be a program and $P, Q \in \PosI(\HH_S)$. We define $\vDash_{\mathsf{iqHL}} \triple{P}{S}{Q}$ to mean $\forall \rho \in \DD(\HH_S). \ \tr(P\rho) \geq \tr(Q\sem{S}\rho)$.
\end{definition}
Following a similar reasoning, 
we could conclude a semantic embedding result:
\begin{theorem}
  For $S$ AST, $X, Y \in \cS(\HH_S)$, the following holds:
  \[
    \vDash_{\mathsf{pqHL}}\triple{X}{S}{Y} \iff \vDash_{\mathsf{iqHL}} \triple{X\mid 0}{S}{Y^\perp}.
  \]
\end{theorem}
Note that this is not the first or the unique possible embedding of $\mathsf{pqHL}$ in a quantitative quantum Hoare logic. In fact, in the simple, non-relational case, the naive embedding is complete \cite{ZhouYuYing2019}:
\[
\vDash_{\mathsf{pqHL}} \triple{X}{S}{Y} \iff \vDash_{\mathsf{qHL}} \triple{X^\perp}{S}{Y^\perp}
\]
where $\mathsf{qHL}$ is a special case of $\mathsf{iqHL}$ where all predicates are bounded.
The advantage of our approach lies in its generality: it works even when the naive embedding does not apply, as is the case of (quantum) relational logics \cite{barthe_rqpd}.

% Recall that there are two standard types of predicates used by 
% relational Hoare logics for quantum programs. By using (generalised) PSD predicates, our approach follows the quantitative route: the extent to which $\rho$ satisfies $P \in \PosI(\HH)$ is determined by $\tr(P\rho)$. A commonly studied alternative (\cite{qRHL_Unruh_2019,barthe_rqpd,approx_rel_reasoning_quantum_programs}) uses discrete, projective predicates in $\cS(\HH)$: $\rho$ satisfies $X \in \cS(\HH)$ if $\supp(\rho) \subseteq X$. Translated to our setting, 

% \begin{definition}[Logic for Projective Predicates]
%   We write $\vDash_{\pqRHL} Z: \rtriple{X}{S_1}{S_2}{Y}$, where $X, Y$ are closed subspaces of $\HH_1 \otimes \HH_2$, to mean that for all initial state $\rho$ such that $\supp \rho \subseteq X$, there exists a coupling $\sigma$ of $\langle \sem{S_1}(\tr_2(\rho)), \sem{S_2}(\tr_1(\rho)) \rangle$ such that $\supp \sigma \subseteq Y$. 
% \end{definition}

\section{Applications}
We present more applications of our completeness results in characterizing non-trivial relational properties, including quantum non-interference, quantum differential privacy, as well as an alternative characterization of equivalence that only needs bounded predicates. Interestingly, most of these results are direct consequences of completeness for split postconditions and do not require the duality rule.

% We present more applications of our completeness theorems in 

% Split Postconditions were crucial to proving completeness for $\qqRHL$. They also turn out to be even more widely applicable. In this section, we present further applications of split postconditions in 

\subsection{Program Equivalence}
% \tw{Add Equivalence Rules in Fig 1}
In the previous section, we showed how equivalence can be characterized with the help of infinite-valued predicates. But can we do it only with bounded quantitative predicates? This question is of particular interest, for it allows us to obtain a complete characterization of equivalence with a more minimal extension to existing quantum relational Hoare logics \cite{barthe_rqpd}. We show that this is indeed possible, and it relies on deep results in QOT.

% We first revisit the problem of characterising equivalence. 

% One of the big limitations of quantum relational Hoare logics using quantitative (Hermitian) predicates \cite{barthe_rqpd,qRHLWithExpectations} (compared to those using projective predicates \cite{barthe_rqpd,qRHL_Unruh_2019,approx_rel_reasoning_quantum_programs}) has been that there is no known sound and complete encoding of program equivalence. This section addresses this exact limitation: we show an encoding of program equivalence in the complete fragment of $\qqRHL$. 

% Let $P_{sym}(\HH) = \frac{1}{2}(I + \swap(\HH))$ (for short, we sometimes simply write $\symeq$) where $\swap(\HH) = \sum_{ij}|ij\>\<ji|$, and parameter $\HH$ is omitted if it is clear from the context.

\begin{theorem}
  \label{thm:equal rule}
    Let $S_1, S_2$ be AST programs acting on the same Hilbert spaces, $\HH_{S_1} = \HH_{S_2} = \HH$. $S_1$ and $S_2$ are semantically equivalent, i.e., $\sem{S_1} = \sem{S_2}$, if and only if,
    \begin{align}
    \label{eqn:equal program sym1}
        \vdash (Y_1,&Y_2,n)\in\mathcal{Y} : \{nI + P_{sym}^\bot\}\nonumber\\
        &S_1 \sim S_2 \{ \tr_2(Y_1)\otimes I + I \otimes(nI - \tr_2(Y_2)) \}.
    \end{align}
where $\mathcal{Y} = \{(Y_1, Y_2\in \Pos(\HH\otimes\HH_2), n\in\mathbb{N}) \mid 0\sqsubseteq Y_1, 0\sqsubseteq 2Y_2\sqsubseteq nI, P_{sym}^\bot[\HH\otimes\HH_2] \ge 2(Y_1\otimes I - I\otimes Y_2)\}$.
%\label{thm: equivalence sym1}
\end{theorem}
\begin{proof}
  Immediate consequence of \Cref{thm:weakcompleteness} and \Cref{thm:strassen QOT}.
\end{proof}

Therefore, the fragment of $\qqRHL$ using only finite-valued predicates is complete for program equivalence for AST programs.

\subsection{Trace Distance and Diamond Norm}
\label{section:tracedistance}

Another application of our completeness result for split postconditions would be a notion of completeness with respect to the diamond norm of quantum channels, which builds upon the encoding of the trace distance -- the quantum analogue of the total variation distance.
% Diamond norm is closely related to channel discrimination: it gives the maximum probability of successfully distinguishing between two quantum channels in a single-shot scenario with the help of auxiliary systems.
% The diamond norm is closely tied to channel discrimination, as it quantifies the maximum probability of successfully distinguishing between two quantum channels in a single-shot scenario, potentially with the aid of auxiliary systems
The diamond norm is closely related to channel discrimination,
as it quantifies the maximum probability of successfully distinguishing between two quantum channels in a single-shot scenario with the help of auxiliary systems.
As such, it serves as the foundation in reasoning about the robustness~\cite{robustnessquantum2019} and error analysis~\cite{Gleipnir2021} of quantum programs, particularly important in the current noisy intermediate-scale quantum (NISQ) era and beyond~\cite{Preskill2018quantumcomputingin}.
We first recall some relevant definitions and properties.

\begin{definition}[Trace Distance (see e.g. \cite{Wilde_2017} Definition 9.1.2)]
  Let $\rho_1, \rho_2$ be density operators over $\HH$. Then their trace distance is defined as
  $
  \TD(\rho, \sigma) \eqdef \frac{1}{2}\|\rho - \sigma\|_1,
  $
  where $\|\cdot\|_1$ is the trace norm defined by 
  $\|M\|_1 = \tr(\sqrt{M^\dagger M})$.
\end{definition}

Trace distance is also referred to as a quantum generalisation of total variation distance, as it can be alternatively characterised by the maximum (see Lemma 9.1.1 in~\cite{Wilde_2017}):
\[
  \TD(\rho, \sigma) = \max_{0 \leqlow P \leqlow I} \tr(P(\rho - \sigma)).
\]

We have the following characterization of the trace distance.

\begin{proposition}[Encoding of Trace Distance]
\label{prop:encodingoftracedistance}
  The following are equivalent for all AST programs $S_1, S_2$ such that\footnote{We could also just ask all programs to be interpreted over $\HH = \HH_{\text{all variables}}$, or over $\HH = \HH_{\var(S_1) \cup \var(S_2)}$.} $\HH_{S_1} = \HH_{S_2}$:
  \begin{enumerate}
    \item $\TD(\sem{S_1}(\rho_1), \sem{S_2}(\rho_2)) \leq \tr(\Phi_1 \rho_1) + \tr(\Phi_2 \rho_2)$ for all $z \in Z$ and $\rho_1 X^{\#} \rho_2$, for some given 
    subspace $X$;
    \item $\vDash  0 \leqlow P \leqlow I: \rtriple{X\mid (I+\Phi_1\otimes I + I\otimes \Phi_2)}{S_1}{S_2}{P \otimes I + I \otimes (I-P)}$.
  \end{enumerate}
\end{proposition}

We now introduce the notion of diamond norms of quantum channels.

\begin{definition}[Diamond Norm, Definition 8 in~\cite{AKN98}]
\label{def:diamond-norm}
    Let $\Phi: M_n(\mathbb{C})\to M_m(\mathbb{C})$ be 
    a linear transformation, where $M_n(\mathbb{C})$
    denote the set of $n\times n$ complex matrices,
    and let 
    $id_n:M_n(\mathbb{C})\to M_n(\mathbb{C})$
    be the identity map.
    Then, the diamond norm 
    (also known as the completely bounded trace norm) 
    of $\Phi$
    is given by
    \[
    \|\Phi\|_{\diamond}
    = \max_{X;\|X\|_1\le 1}
    \|(\Phi\otimes id_n) X\|.
    \]
\end{definition}
    
The diamond norm induces the diamond distance.
For completely positive, trace non-increasing maps
$\EE_1$ and $\EE_2$ with domain $\Pos(\HH)$,
%where$\HH$ is a Hilbert space of dimension $n$,
their diamond distance could be written as
\[
    \|\EE_1-\EE_2\|_{\diamond}
    = \max_{\rho \in \DD^1(\HH\otimes\HH)}
    \|(\EE_1\otimes \mathcal{I})(\rho) -(\EE_2\otimes \mathcal{I})(\rho)  \|_1,
\]
where $\mathcal{I}$ is the identity quantum channel on $\HH$.
% To see what properties concerning trace distance
% can be encoded via the above proposition,
% we could consider the following simple example:
Now, consider setting
$X = P_{sym}[\HH\otimes \HH]$, and $\Phi_1 = \Phi_2 = cI/2$
in \Cref{prop:encodingoftracedistance},
where $c\ge 0$ is a constant.
The property we are trying to encode becomes
\[
    \TD((\sem{S_1}\otimes \mathcal{I})(\rho_1), (\sem{S_2}\otimes \mathcal{I})(\rho_2))
    \le c,\ \ \forall \rho_1 (=_{sym})^{\#} \rho_2.
\]
Noting that $\rho_1 (=_{sym})^{\#} \rho_2$ iff $\rho_1 = \rho_2$, this gives an encoding
of the diamond distance between 
$\sem{S_1}$ and $\sem{S_2}$, which we formally 
stated as follows.
% that running $S_1$ and $S_2$ on the same input $\rho_1$, the trace distance between the outputs of the programs is no more than $c$.

% By setting $S_1$ and $S_2$ to $\EE_1 \otimes I$ and $\EE_2 \otimes I$, this gives an encoding that the diamond norm of $\EE_1-\EE_2$ is upper bouned by $a$.

%\lz{define diamond norm}
% \begin{theorem}[Completeness with respect to Trace Distance]
%   \label{thm:completenesstracedist}
%   The $\qqRHL$ is complete with respect to trace distance, for AST programs.
% \end{theorem}
% \begin{proof}
%   Follows from \Cref{prop:encodingoftracedistance}, \Cref{thm:soundness} and \Cref{thm:weakcompleteness}.
% \end{proof}

\begin{proposition}[Encoding of Diamond Norm]
\label{prop:encodingofdiamondnorm}
  Let $c\in\mathbb{R}^+$. The following are equivalent for all AST programs $S_1, S_2$ such that $\HH = \HH_{S_1} = \HH_{S_2}$:
  \begin{enumerate}
    \item $\|\sem{S_1} - \sem{S_2}\|_\diamond \leq 2c$;
    \item $\vDash  0 \leqlow P \leqlow I_{\HH\otimes\HH}: \rtriple{P_{sym}[\HH\otimes\HH]\mid (1+c)I}{S_1}{S_2}{P \otimes I + I \otimes (I-P)}$.
  \end{enumerate}
\end{proposition}

\begin{theorem}[Completeness with respect to Diamond Norm]
  \label{thm:completeness diamond norm}
  The $\qqRHL$ is complete with respect to diamond norm, for AST programs.
\end{theorem}

\paragraph*{Comparison to~\cite{robustnessquantum2019,Gleipnir2021}} The program logic introduced in~\cite{robustnessquantum2019,Gleipnir2021} provides a sound method for reasoning about the upper bound of $(Q,\lambda)$-diamond norm between a noisy program and its ideal counterpart. However, its completeness remains unknown.
\Cref{thm:completeness diamond norm} can be extended to establish complete reasoning for the upper bound of $(X,1)$-diamond norm where $X\in\cS(\HH)$ is a subspace.

\subsection{Quantum Wasserstein Semi-Distance}

% The Wasserstein metric, or the earth mover's 
% distance, is a distance function between
% two probability distributions.
% This distance is of great importance as it
% characterizes the minimal
% cost of transforming one probability 
% distribution into the other 
% in the context of optimal transport.
% Several quantum generalizations
% of the Wasserstein metric have
% been proposed, 
% but people are only able to show
% they are semi-distance for density matrices.
% In this work, 
% we study the following definition 
% of the quantum Wasserstein semi-distance.

The Wasserstein metric, also known as the earth mover's distance, 
is a measure of distance between two probability distributions. 
It is important because it characterizes the minimal cost 
required to transform one probability distribution 
into the other
in the context of optimal transport.
Several quantum generalizations of the Wasserstein metric have been proposed. 
However, so far, these generalizations have only been shown 
to satisfy the properties of a semi-distance for density matrices.
In this work, we adopt the following definition of the quantum Wasserstein semi-distance discussed in \cite{QuantumOptimalTransport_Cole_2023}.

Let $\rho, \sigma \in \mathcal{D}^1(\mathcal{H})$
be two density operators.
Their quantum $2$-Wasserstein semi-distance 
$W(\rho, \sigma)$ is defined as 
\[
    W(\rho, \sigma) = \sqrt{T(\rho, \sigma)},
\]
where $T(\rho, \sigma) = T_{P_{sym}^{\bot}}(\rho, \sigma)$ is the QOT
between $\rho$ and $\sigma$ with the cost function 
$P_{sym}^{\bot}$, see \Cref{sec: QOT equal} for details.

Verifying properties of programs related to 
the above quantum Wasserstein semi-distance
can be easily encoded in our logic.
% Specifically, we consider the Lipschitz
% property of the programs with respect 
% to the quantum Wasserstein semi-distance,
% which states that the quantum Wasserstein semi-distance between the outputs of a program
% is bounded by a constant $\lambda$
% times the quantum Wasserstein semi-distance between the inputs.
Specifically, we investigate the Lipschitz property of programs with respect to the quantum Wasserstein semi-distance. This property asserts that the quantum Wasserstein semi-distance between a program's outputs is bounded by the quantum Wasserstein semi-distance between its inputs scaled by a constant $\lambda$.
This property can be directly encoded and verified
in our logic:

% \mb{$\HH_{S_1} = \HH_{S_2}$ or $\HH_{S_1} \cong \HH_{S_2}$?
% Similar for the encoding of trace distance}
\begin{proposition}[Encoding of Quantum Wasserstein Semi-Distance]%
\label{prop:encoding-of-wasserstein-distances-lipschitz}
  Let $\lambda>0$. 
  The following are equivalent for all AST programs $S_1, S_2$ such that 
  $\HH_{S_1} = \HH_{S_2}$:
  \begin{enumerate}
    \item $W(\sem{S_1}(\tr_2(\rho)), \sem{S_2}(\tr_1(\rho))) \leq \lambda \cdot W(\tr_2(\rho), \tr_1(\rho))$ for all $\rho \in \DD(\HH_{S_1} \otimes \HH_{S_2})$;
    \item $ \vDash  \rtriple{\lambda^2 P_{sym}^{\bot}}{S_1}{S_2}{P_{sym}^{\bot}}$.
  \end{enumerate}
\end{proposition}

% By combining \Cref{prop:encoding-of-wasserstein-distances-lipschitz}, \Cref{thm:soundness} and \Cref{thm:weakcompleteness}, we immediately get:
%Combining above with our previous results, we have:
\begin{theorem}[Completeness with respect to Quantum Wasserstein Semi-Distance]%
\label{thm:completeness-quantum-wasserstein-semi-distance}
  The $\qqRHL$ is complete with respect to the Lipschitz property
  of quantum Wasserstein semi-distance for AST programs.
\end{theorem}
% \begin{proof}
%   Follows from \Cref{prop:encoding-of-wasserstein-distances-lipschitz}, \Cref{thm:soundness} and \Cref{thm:weakcompleteness}.
% \end{proof}

% \begin{proof}
%     By definition, we know that 
%     the second condition is equivalent to: 
%     for every $\rho \in \DD(\HH_{S_1} \otimes \HH_{S_2})$,
%     there is a coupling 
%     $\sigma :\langle \sem{S_1}(\tr_2(\rho)), \sem{S_2}(\tr_1(\rho))\rangle$
%     such that
%     \[
%     \lambda^2 \tr (\rho P_{sym}^{\bot})
%     \ge \tr (\sigma P_{sym}^{\bot}).
%     \]
%     Note that 
%     $\sigma$ only depends on
%     $\tr_2(\rho), \tr_1(\rho)$.
%     Thus, 
%     for fixed $\rho_1$ and $\rho_2$
%     we can write the second 
%     condition equivalently as
%     \[
%     \min_{\rho:\langle\rho_1, \rho_2\rangle} \max_{\sigma:\langle \sem{S_1}(\rho_1), \sem{S_2}(\rho_2)\rangle} \lambda^2  \tr (\rho P_{sym}^{\bot})
%     -  \tr (\sigma P_{sym}^{\bot}) \ge 0.
%     \]
%     This can be simplified to 
%     \[
%      \lambda^2\min_{\rho:\langle\rho_1, \rho_2\rangle}  \tr (\rho P_{sym}^{\bot})
%     \ge \min_{\sigma:\langle \sem{S_1}(\rho_1), \sem{S_2}(\rho_2)\rangle} \tr (\sigma P_{sym}^{\bot}),
%     \]
%     which is equivalent to the first condition by definition.
% \end{proof}

\subsection{Quantum Non-Interference}
Another application of our logic
involves characterizing the concept of quantum 
non-interference.
Intuitively,
non-interference refers to a critical property
where the actions of one group of agents
in a computer system
do not influence the actions of another group 
of agents.
This concept was later extended
to quantum settings in \cite{Ying2013Quantum}.
The key components for defining quantum non-interference 
in their work 
include quantum computer systems, 
a pseudo-distance for measuring output distributions,
and the degree of interference. 
We introduce these notions as follows.

% Thus, to define quantum 
% non-interference, we need the following 
% notions of quantum computer systems and
% pseudo distances between mixed quantum states to quantify the influence of actions.
% In~\cite{Ying2013Quantum},
%the authors gave the following automata model to describe quantum systems for non-interference.

\begin{definition}[Quantum Systems, Definition 3.1 in \cite{Ying2013Quantum}]
    A quantum system is a $6$-tuple
    \[  
    \mathbb{S}=\left\langle\mathcal{H},\rho_0,A,C,do,measure \right\rangle,
    \]
    where 
    \begin{itemize}
        \item $\mathcal{H}$ is a Hilbert space specifying 
    the state space;
        \item  $\rho \in \mathcal{D}(\mathcal{H})$  specifying the initial state;
        \item  $A$ is a set of agents;
        \item $C$ is a set of commands;
        \item $do=\{\mathcal{E}_{a,c}|a\in A\text{ and }c\in C\}$ is a set of trace-preserving operations $\mathcal{E}_{a,c}$ which
        describes how states are updated 
        when agent $a$ executes command $c$;
        \item $measure=\{\mathbb{M}_a|a\in A\}$
        is a collection of sets of POVM measurements, where each
        $\mathbb{M}_a$ is allowable for agent $a$.
    \end{itemize}
\end{definition}

The quantum non-interference influence is measured 
by the following pseudo distance 
$d_{\mathbb{M}}$
induced by POVM measurements $\mathbb{M}$.
Let $d(p,q)$ denote the total variation
distance between two probability distributions 
$p$ and $q$ over the sample space $X$.
% is given by 
% $
%     d(p,q) = \frac{1}{2}\sum_{x\in X}\abs{p(x)-q(x)} 
% $.
For a density operator $\rho$
and a POVM measurement 
$E = \{E_{\lambda}|\lambda\in \Lambda\}$,
we define the 
probability distribution $p_{E,\rho}$ as
$p_{E,\rho}(\lambda) = \tr (E_{\lambda}\rho)$.
The pseudo distance  
$d_{\mathbb{M}}$ is formalized as follows.
%We now revisit the measure of influence
%for describing quantum non interference,
%namely a 
%pseudo distance 
%$d_{\mathbb{M}}$ induced by 
%a set of POVM measurements $\mathbb{M}$ defined in~\cite{Ying2013Quantum}.
%We first recall that,
%for two probability distributions 
%$p$ and $q$ over the same sample space $X$,
%their total variance distance $d(p,q)$ 
%is given by
% \[
%     d(p,q) = \frac{1}{2}\sum_{x\in X}\abs{p(x)-q(x)} = \max_{E\subseteq X} \left(p(E)-q(E)\right).
% \]
% Given a density operator $\rho$
% and a POVM measurement 
% $E = \{E_{\lambda}|\lambda\in \Lambda\}$
% we could naturally define a 
% probability distribution $p_{E,\rho}$ as
% \[
%     p_{E,\rho}(\lambda) = \tr (E_{\lambda}\rho).
% \]
% We will write the distribution as $p_E(\rho)$ to 
% emphasize its dependence on $\rho$.
% We could then measure the ``distance''
% of two density operators
% by considering the total variation distance 
% of the distributions given by 
% POVM measurements,
% which is formally stated as follows.

\begin{definition}[Definition 3.2 in \cite{Ying2013Quantum}]
    The pseudo distance $d_{\mathbb{M}}$ between 
    two density operators 
    $\rho, \sigma\in \mathcal{D}(\mathcal{H})$
    induced by 
    a set of POVM measurements $\mathbb{M}$
    is defined as
    \[
        d_{\mathbb{M}}(\rho, \sigma)
        = \sup_{E\in \mathbb{M}}
        d(p_E(\rho), p_E(\sigma)).
    \]
\end{definition}
For each agent $a \in A$ in a quantum system $\mathbb{S}$  , we write $d_a =
d_{\mathbb{M}_a}$
for the pseudo distance defined by the set $\mathbb{M}_a$ of POVM measurements.

Let $G\subseteq A$ be a group of agents,
$D\subseteq A$ be a set of commands.
For a sequence of actions
$\alpha = \alpha_1 \alpha_2 \cdots \alpha_n \in (A\times C)^*$,
we define a function $\mathsf{purge}_{G,D}$ for $\alpha$
that removes the actions in $D$ by the agents
in group $G$.
In other words,
\[
    \mathsf{purge_{G,D}}(\alpha) = 
    \alpha_1'\alpha_2'\cdots\alpha_n',
\]
where $\alpha_i' = \epsilon$ (the empty action) if $\alpha_i' = (a_i, c_i)$ with $a_i\in G$ and $c_i\in D$,
and $\alpha_i' = \alpha_i$ otherwise. 
The degree of (non-) interference is measured
by the pseudo distance $d_{\mathbb{M}}$ 
between the state operated by the original actions
and the state operated by the purged actions.

\begin{definition}[Interference Degree, Definition 3.3 in \cite{Ying2013Quantum}]
    For a quantum system $\mathbb{S} = \left\langle\mathcal{H},\rho_0,A,C,do,measure \right\rangle$, let $G_1, G_2\subseteq A$ be 
    two groups of agents, and $D\subseteq C$
    be a set of commands.
    Then, the degree that agents in $G_1$ with commands $D$ interfere agents in $G_2$ is
    \[
          Int(G_1, D|G_2) = 
           \sup_{\substack{\alpha\in(A\times C)^*,\\ a\in G_2}} \{d_a(\mathcal{E}_{\alpha}(\rho_0), \mathcal{E}_{\mathsf{purge}_{G_1, D}(\alpha)}(\rho_0))\}
    \]
\end{definition}

If $ Int(G_1, D|G_2) = 0$, we will denote this as
$G_1, D : |G_2$.
%Additionally, if $D = C$, we will
%simplify this notation to $G_1:| G_2$.

For a quantum system $\mathbb{S} = \left\langle\mathcal{H},\rho_0,A,C,do,measure \right\rangle$ 
with an initial state $\rho_0 = \ket{0}\bra{0}$, 
we represent trace-preserving operations $\mathcal{E}_{a,c}$
as corresponding AST programs $S_{a,c}$, 
and sequence of actions $\alpha$ as AST programs $S_{\alpha}$.
The following and theorem shows the quantum
non-interference property can be encoded and verified 
using our logic.

\begin{proposition}[Encoding of Quantum Non-Interference]\label{prop:encoding-of-non-interference}
    For a quantum system $\mathbb{S} = \left\langle\mathcal{H},\rho_0,A,C,do,measure \right\rangle$ with $\rho_0 = \ket{0}\bra{0}$,
    let $G_1, G_2\subseteq A$ be 
    two groups of agents, and $D\subseteq C$
    be a set of commands.
    The following are equivalent:
    \begin{itemize}
        \item $G_1, D : |G_2$.
        \item $\forall \alpha\in (A\times C)^{*}$,  
        \[
        \begin{aligned}
                    \vDash  a&\in G_2, E = \{E_{\lambda}|\lambda\in \Lambda_E\}\in \mathbb{M}_a,  T\subseteq \Lambda_E: \\
        &\rtriple{I}{q:=\ket{0};S_{\alpha}}{q:=\ket{0};S_{\mathsf{purge}_{G_1, D}(\alpha)}}{M}{},
        \end{aligned}
        \]
    \end{itemize}
    where $M= M_T\otimes I+I\otimes (I-M_T)$ with $M_T = \sum_{\lambda\in T} E_{\lambda}$.
\end{proposition}

% Following the same reasoning as in previous applications,
% we have:
\begin{theorem}[Completeness with respect to Quantum Non-Interference]%
\label{thm:completeness-quantum-non-interference}
  The $\qqRHL$ is complete with respect to the quantum non-interference
  property for AST programs.
\end{theorem}
% \begin{proof}
%   Follows from \Cref{prop:encoding-of-non-interference}, \Cref{thm:soundness} and \Cref{thm:weakcompleteness}.
% \end{proof}

% \begin{proof}
%    By the definition of validity, the second condition
%    can be equivalently written as the following.
%    $\forall \alpha\in (A\times C)^{*}, a\in G_2, E = \{E_{\lambda}|\lambda\in \Lambda_E\}\in \mathbb{M}_a$ and  $T\subseteq \Lambda_E$,
%    we have
%    \[
%    \tr (\rho I )\ge \tr(\sigma (M_T\otimes I+I\otimes (I-M_T))),
%    \]
%    which could be simplified to
%    \[
%    \tr(M_T(\sigma_1-\sigma_2))\le 0,
%    \]
%    with 
%    \[\sigma_1 = \tr_2(\sigma) = \mathcal{E}_{\alpha} (\ket{0}\bra{0}),
%    \] and 
%    \[\sigma_2 = \tr_1(\sigma) =\mathcal{E}_{\mathsf{purge}_{G_1, D}(\alpha)} (\ket{0}\bra{0}).
%    \]
%     Noting that
%     $\forall T\subseteq \Lambda_E$,
%     $\tr(M_T(\sigma_1-\sigma_2))\le 0$
%     is equivalent to
%     \[
%     \max_{T} \left(p_{E, \sigma_1}(T) - p_{E, \sigma_2}(T) \right)\le 0,
%     \]
%     we can rewrite the second condition as 
%     $\forall \alpha\in (A\times C)^{*}, a\in G_2$,
%      \[
%      d_a \left(\mathcal{E}_{\alpha} (\ket{0}\bra{0}) , \mathcal{E}_{\mathsf{purge}_{G_1, D}(\alpha)} (\ket{0}\bra{0})\right) \le 0.
%      \]
%      Then, it is equivalent to 
%      $G_1, D : |G_2$ by definition.
   
% \end{proof}

\subsection{Quantum Differential Privacy}

% With infinite-valued predicates introduced in
% our logic, we could characterize 
% an important notion called differential privacy.
% Differential privacy is a mathematical framework 
% about providing statistical properties about datasets
% without leaking information of individual objects.
% The intuition about differential privacy is that,
% if running a program with two ``neighboring'' inputs,
% the output distributions will be close in
% some sense.
% Researchers have been successfully 
% generalized this notion into quantum computing,
% proposing similar yet different notions of
% quantum differential privacy, such as in 
% \cite{AR19,ZhouYingDiffPriv}.
% In this work, we adopt the following definition.

Finally, with the help of infinite-valued predicates, we can characterize a quantum version of differential privacy.
Differential privacy is a mathematical framework 
about providing statistical properties about datasets
while preserving private information of individual objects.
The core intuition behind differential privacy is that the output distributions of a program should remain nearly indistinguishable when run on two ``neighboring'' inputs, usually differing only by a single element.
This concept has been successfully extended to quantum computing, resulting in the development of related but distinct notions of quantum differential privacy. These notions, as explored in works such as \cite{AR19, ZhouYingDiffPriv}, primarily differ in how they define and characterize ``neighboring inputs''.
% This concept has been successfully extended to quantum computing, leading to the development of related but distinct notions of quantum differential privacy, as explored in works such as \cite{AR19, ZhouYingDiffPriv}, which mainly
% differs 
% from the way of characterizing ``neighboring inputs''.
In this paper, we adopt the following definition proposed in \cite{AR19}.

% \begin{definition}[Quantum Differential Privacy, Definition 3 in~\cite{AR19}]%
% \label{def:quantum-differential-privacy}
% Let
%  $\varepsilon, \delta >0$ be constants. A
% quantum operation $\EE$ is $(\varepsilon, \delta)$-differentially private if for
% every POVM $M=\{M_m\}_{m \in \mathsf{Out}(M)}$, every
% $S \subseteq \mathsf{Out}(m)$ and every input $\rho, \sigma$
% such that there exists $i$ satisfying $\tr_{i}(\rho) = \tr_{i}(\sigma)$,
% it holds that
% \[
%   \Pr[\EE(\rho) \in_M S] \leq \exp(\varepsilon) \cdot \Pr[\EE(\sigma) \in_M S] + \delta,
% \]
% where $\Pr[\rho \in_M S] = \sum_{m \in S} \tr(M_m \rho)$.
% \end{definition}

\begin{definition}[Quantum Differential Privacy, Definition 3 in~\cite{AR19}]%
\label{def:quantum-differential-privacy}
Let
$\varepsilon, \delta >0$ be constants. 
A quantum operation $\EE$ on an $n$-qubit system
is $(\varepsilon, \delta)$-differentially private,
if for
every POVM $M=\{M_m\}_{m \in \mathsf{Out}(m)}$, every
$S \subseteq \mathsf{Out}(m)$,
and every input $\rho, \sigma$
that differs at most one qubit (i.e., there exists $i\in [n]$ that $\tr_{i}(\rho) = \tr_{i}(\sigma)$),
it holds that
\[
  \Pr[\EE(\rho) \in_M S] \leq \exp(\varepsilon) \cdot \Pr(\EE(\sigma) \in_M S) + \delta,
\]
where $\Pr(\rho \in_M S) = \sum_{m \in S} \tr(M_m \rho)$.
\end{definition}

Let $S_{1}$ be some quantum program
that corresponds to the
quantum operation $\mathcal{E}$.
We could verify whether
$\mathcal{E}$, or equivalent $S_{1}$ is $(\varepsilon, \delta)$-differentially private
in our logic,
as implied by the following proposition and theorem.

\begin{proposition}[Encoding of Differential Privacy]%
\label{prop:encoding-of-differential-privacy}
  The following are equivalent for all AST programs $S_1$ on an $n$-qubit system with:
  \begin{enumerate}
    \item $S_1$ is  $(\varepsilon, \delta)$-differentially private;
    \item $\vDash i\in [n], 0\sqsubseteq M \sqsubseteq I: \rtriple{P_{i,sym} \mid (\exp(\varepsilon)+ \delta) I}{S_1}{S_{1}}{M\otimes I  + \exp(\varepsilon) I\otimes (I-M)}$.
  \end{enumerate}
  Here $P_{i, sym} = P_{sym}[\mathcal{H}_{[n] - i}] \otimes (I_{i\varsinone}\otimes I_{i\varsintwo})$ for $i\in [n]$.
\end{proposition}

% Following the same reasoning as in previous applications,
% we have:
\begin{theorem}[A Complete Characterization of Quantum Differential Privacy]%
\label{thm:completeness-quantum-differential-privacy}
  The $\qqRHL$ is complete with respect to the quantum differential privacy
  property for AST programs.
\end{theorem}
% \begin{proof}
%   Follows from \Cref{prop:encoding-of-differential-privacy}, \Cref{thm:soundness} and \Cref{thm:weakcompleteness}.
% \end{proof}

% \begin{proof}
%   We first notice that,
%   for some fixed $i$,
%   consider any $\rho$ and $\sigma$ with a coupling
%   $\rho_{0}:\langle \rho, \sigma\rangle$.
%   If $\tr_{i}(\rho) \ne \tr_{i}(\sigma)$,
%   then $\tr(P_{i,sym}^{\bot}\rho_{0}) > 0 $,
%   and thus
%   $\tr(\rho_0 (P_{i,sym}\vert (\exp (\varepsilon)+\delta)I)) = +\infty$,
%   meaning that it is always valid in this case.
%   Therefore, we only need to consider the case
%   where $\tr_{i}(\rho) = \tr_{i}(\sigma)$.
%   In this case, the condition is
%   for any $M$ satisfying $0\sqsubseteq M \sqsubseteq I$,
%   \[
%     \exp(\varepsilon) + \delta \ge \tr(\sem{S_{1}}(\rho) M) + \exp(\varepsilon)(1- \tr(\sem{S_{1}}(\sigma) M)),
%   \]
%   which could be simplified to
%   \[
%         \tr(\sem{S_{1}}(\rho) M) \le  \exp(\varepsilon) \tr(\sem{S_{1}}(\sigma) M) + \delta.
%   \]
%   Note that when $M$ goes over all $0\sqsubseteq M \sqsubseteq I$,
%   it exactly goes over all POVMs $\sum_{m\in S} M_{m}$,
%   we know that the second condition is then equivalent to the first
%   as we want.
% \end{proof}

% \section{Concrete Examples}
% \lz{Example 1.1 from \cite{barthe_rqpd}. Expect to be one-column.}

\section{Probabilistic Duality}
There is a variety of relational program logics for probabilistic
programs~\cite{BartheGB09,Aguirre:POPL:21,GregersenAHTB24,erhl}. Similar to the quantum case, validity of these logics is based on probabilistic
couplings, their completeness has remained an open problem. Recent
work by Avanzini \textit{et al}.~\cite{erhl} defines a quantitative
relational Hoare logic, called $\eRHL$, and shows that it
achieves completeness for non-trivial classes of properties. Their
proof of completeness leverages a notion of split post-condition
similar to ours. However, their work lacks a general completeness
theorem. We show that their logic is in fact complete for bounded postconditions. Completeness follows from the classic Kantorovich-Rubinstein duality. Our phrasing of the theorem is stated w.r.t.\,
positive functions $c_1$ and $c_2$, to match the assumption that assertions in $\eRHL$ take positive values.
\begin{theorem}[Kantorovich-Rubinstein Duality]
\label{prop:Kantorovich-Rubinstein-duality}
Let $\mu,\nu$ be discrete probability distributions over $X$ and $Y$
respectively, and let $c:X\times Y \rightarrow [0,+\infty)$ be a
  bounded function. Then
$$\inf_{\theta \in \Gamma (\mu,\nu)} \E_\theta [c] =
  \sup_{(n,c_1,c_2)\in\mathcal{W}} (\E_\mu[c_1]+\E_\nu[c_2] -n)
  $$ where $\Gamma (\mu,\nu)$ denotes the set of probabilistic
  couplings of $\mu$ and $\nu$ and $(n, c_1,c_2) \in \mathcal{W}$ iff
  for every $x\in X$ and $y\in Y$, we have $0\leq c_1(x),c_2(y)$
  and $c_1(x)+c_2(y)\leq c(x,y)+n$.
\end{theorem}
It follows that every bounded
post-condition is logically equivalent to a universally quantified
split post-condition---where in the probabilistic setting a split
post-condition is simply the addition of two unary assertions on the
first and second state respectively. Therefore $\eRHL$ is
complete for all bounded postconditions.
\begin{proposition}
    \label{prop:completeness-erhl}
 $\eRHL$ is complete for all AST programs and bounded
 postconditions.
\end{proposition}
Note that completeness does not require adding a duality rule, due to a
difference of settings. Indeed, $\eRHL$ features a very general
rule of consequence, which allows using arbitrary theorems from the
theory of couplings. 

Interestingly,~\cite{erhl} establishes a completeness theorem for
judgments of the $\mathsf{pRHL}$ logic, using Strassen's
theorem~\cite{strassen1965existence}. This is very similar in spirit to
our use of the duality theorem. In fact, Strassen's theorem can be
seen as a specialized variant of the duality theorem for
boolean-valued cost functions. Furthermore, note that~\cite{erhl} uses the Kantorovich-Rubinstein duality to prove that $\eRHL$ characterizes Kantorovich distance, but fails to establish a link with completeness.

\section{Related work and Discussion}

\subsection{Quantum Relational Hoare Logics}
% \tw{Should I include $\rqPD$'s validity for clarity or is the following enough?}
Our logic can be seen as a generalization of Barthe \textit{et al}.'s $\rqPD$~\cite{barthe_rqpd}. It differs with $\rqPD$ in three ways: we consider possibly infinite-valued positive predicates (instead of finite-valued and subunital ones), we consider an upper-bounding instead of a lower-bounding semantics, and we allow for logical variables. That said, $\rqPD$ can be semantically embedded in our logic as follows:
\begin{proposition}
  \label{prop:semanticembeddingrqPD}
  Let $S_1$, $S_2$ be AST programs and $0\leqlow P, Q \leqlow I$ be predicates. Then 
  $$\vDash_{\rqPD}\rtriple{P}{S_1}{S_2}{Q} \iff \vDash \rtriple{I-P}{S_1}{S_2}{I-Q}.$$
\end{proposition}
Moreover, most of our proof rules are adapted from $\rqPD$'s rules, and it turns out that straightforwardly generalizing $\rqPD$'s one-sided rules to our setting is sufficient to achieve completeness without needing two-sided rules with complex semantic conditions like measurement conditions. Nevertheless, the said rules being still sound and potentially more usable, we adapt them and further show that measurement conditions can be reasoned about \textit{within} our logic.

Barthe \textit{et al}.~\cite{barthe_rqpd} also discuss a logic using projective (instead of quantitative) predicates similar to $\pqRHL$, as well as an incomplete embedding of that logic into $\rqPD$. This work achieves a complete embedding with the help of infinite-valued predicates.

Another line of work \cite{qRHL_Unruh_2019,qRHLWithExpectations} is based on separable couplings (instead of general couplings like in \cite{barthe_rqpd} and our work). Unruh~\cite{qRHL_Unruh_2019} defines the first sound relational program logic for quantum programs based on projective predicates and separable couplings. The primary motivation for using separable couplings is that it is possible to prove soundness of a frame rule. Li and Unruh~\cite{qRHLWithExpectations} define an expectation-based variant of~\cite{qRHL_Unruh_2019}. The soundness of their logic is also proved with a notion of validity based on separable couplings. Interestingly, the motivation for using separable couplings in this case is soundness---there is no frame rule in this logic. Neither of these logics are known to be complete. 

Comparing our logic to separable-coupling-based ones is not the focus of our work. As already extensively discussed by \cite{barthe_rqpd} and \cite{qRHLWithExpectations}, while similar sets of proof rules can be sound for both general-coupling-based and separable-coupling-based notions of validity, it is unclear how these notions of validity actually relate to each other. Moreover, it is unclear whether it is possible to adapt ideas in our work (like Strassen's theorem) to obtain similar completeness results.

Finally, in \cite{approx_rel_reasoning_quantum_programs}, Yan \textit{et al}.\ study approximative relational reasoning by giving a logic similar to \cite{barthe_rqpd}'s projective logic, but based on approximate couplings (where a $\rho$ is an $\epsilon$-coupling of $\rho_1$ and $\rho_2$ if both $\TD(\rho_1,\tr_2(\rho))$ and $\TD(\rho_2, \tr_1(\rho))$ are upper-bounded by $\epsilon$, where $\TD$ is the trace distance). In comparsion, while our logic is quantitative and completely characterizes various distance metrics like the Wasserstein semi-distance, we do require couplings to be exact. It would indeed be interesting to explore how expressive our logic is for approximative reasoning compared to their work.

% Our logic bears the most resemblance with the latter work's expectation-based logic. In fact, we can define a notion of validity based on separable couplings as follows, and show that all our rules (barring (Strassen)) are sound with respect to it.
% \begin{definition}[$\qqRHL$ Validity for separable coupling]
%   The judgment $\vdash_s Z: \rtriple{P}{S_1}{S_2}{Q}$ is valid, written
%   \[
%     \vDash_s Z : \rtriple{P}{S_1}{S_2}{Q}
%   \]
%   if for every $z \in Z$, and separable $\rho \in \DD(\HH_{S_1} \otimes \HH_{S_2})$, there exists a separable quantum coupling $\sigma$ for $\langle \sem{S_1}(\tr_2(\rho)), \sem{S_2}(\tr_1(\rho)) \rangle$ such that
%   \[
%   \tr(P_z \rho) \geq \tr(Q_z \sigma).
%   \]
% \end{definition}
% Consistently with the observation made by both \cite{barthe_rqpd} and \cite{qRHLWithExpectations}, the following holds:
% \tw{Does this actually hold btw?}
% \begin{proposition}[Soundness for separable couplings]
%     Fig. \ref{fig:qqrhlrules} excluding (Strassen) is sound for separable couplings.
% \end{proposition}
% Unfortunately, we are not aware of an equivalent of Strassen's Theorem for separable couplings. Therefore, there is no clear path to make a logic for separable couplings that is complete for non-trivial properties.

% \gb{what about the examples? can they be done in these logics?}

\subsection{Quantum Optimal Transport}
Over the past decades, efforts have been made to generalize the optimal transport problem to the quantum setting. 
Early attempts~\cite{QOT_monge_Zyczkowski_1998} defined the cost between two quantum states using the (probabilistic) Monge distance based on their corresponding Husimi distributions, and then explored the physical consequences including unitary evolution and decoherence~\cite{QOT_monge_Zyczkowski_2001}.
More recent approaches have framed quantum optimal transport in terms of expectations over couplings.
%More recent work has approached quantum optimal transport from the perspective of expectations on couplings.
Specifically, given quantum states $\rho_1$ and $\rho_2$, they have studied the expectation $\tr(C\rho)$ over possible "couplings" $\rho$ of $\rho_1$ and $\rho_2$, for both general and specific types of costs:

\cite{QOT_continuous_2016} explored the QOT based on quantum coupling (\Cref{def:coupling} but on continuous space), with the cost specified as energies involving position and momentum operators, primarily for investigating applications in quantum mean-field theory and its classical limits. Related results include such as inequalities involving QOT and related potentials or metrics~\cite{QOT_continuous_2016, QOT_continuous_2022}, Kantorovich type duality theorem~\cite{QOT_continuous_2022a}, and showing that QOT can be cheaper than the classical one~\cite{QOT_continuous_2020}.

Another approach~\cite{QOT_channel_2021, QOT_channel_2024a} examines QOT from a view of changing one state to another, i.e, focusing on the properties of quantum channels $\EE$ that satisfy $\EE(\rho_1)=\rho_2$. It was shown that there is a one-to-one correspondence between such channels and the couplings of $\<\rho_1^T,\rho_2\>$, offering an alternative definition of couplings. They further studied the cost function $\sum_i(R_i\otimes I_2-I_1\otimes R_i^T)^2$, and established related entropic and concentration inequalities~\cite{QOT_channel_2023, QOT_channel_2024}, with an application showing the limitations of Variational Quantum Algorithms in the presence of noise~\cite{QOT_channel_2023}.

% Not limited in special costs, \cite{QuantumOptimalTransport_Cole_2023, Friedland2022Quantum, QOT_nogo_2022} have systematically explored the general properties and associated metric or distance, for example, \cite{mullerhermes2022monotonicity} defined a stable version of quantum optimal transport (QOT), as briefly summarized in \Cref{sec: QOT}. This line of research aligns well with our goals, as we focus on the general relational properties of program states rather than specific costs or physically oriented properties.

In addition to special cost functions, 
prior works have systematically explored the general properties and associated metrics or distances in quantum optimal transport (QOT) \cite{QuantumOptimalTransport_Cole_2023, Friedland2022Quantum, QOT_nogo_2022}. 
For instance, \cite{mullerhermes2022monotonicity} introduced a stable version of QOT, as briefly summarized in \Cref{sec: QOT}. 
This line of research aligns closely with our objectives, as we prioritize the general relational properties of program states over specific costs or physically oriented properties.

\subsection{Reasoning about Program Equivalence}

Outside of quantum relational Hoare logics, reasoning techniques about equivalence between quantum programs/circuits have been widely studied. 

%  look at notions of quantum computing with existing and well-defined notions of program equivalence (e.g. denotational semantics). Their focus is to \textit{axiomatise} them and enable equational reasoning
Several lines of works create axiomatizations of program equivalence in the form of equational theories. Similarly to our work, they look at notions of quantum computing with \textit{existing} and well-defined denotational semantics, with respect to which they would prove properties like soundness and completeness. However, they only focus on program equivalence and have a very different goal of enabling equational reasoning. One line of work stems from the categorical quantum mechanics (CQM) programme \cite{abramsky2007categoricalsemanticsquantumprotocols,PQP_Coecke_Kissinger_2017,heunen_vicary_book} and leverages categorical formalisms developed therein to produce string-diagrammatic axiomatizations of quantum theory like the ZX calculus and others \cite{ZX_1_10.1007/978-3-540-70583-3_25,ZH_Calculus_Backens_2019, ZW_calculus_hadzihasanovic2015diagrammaticaxiomatisationqubitentanglement}. Similarly to these calculi, \cite{Xu_Barthe_Zhou_2025} develop a tool for reasoning about expressions formally written in Dirac notation. Other works, including ours, have a specific focus on quantum computing (CP and trace non-decreasing maps), a proper subset of quantum theory (CP maps). In \cite{statonpopl15}, Staton presents an algebraic theory of qubit quantum computing and proved that it completely axiomatizes a standard model of quantum computation. Recently, \cite{Quantum_PI_Carette_2024} complements \cite{statonpopl15} by giving an equational theory for unitary quantum circuits that is complete for various gate sets, by extending $\Pi$, a model of classical reversible computing. 

% One line of work stems from the categorical quantum mechanics (CQM) programme [], which leverages the graphical tools derived from compact closed categories and related constructions to reason about general quantum theory. The insights produced by CQM led to several string-diagrammatic calculi (ZX calculus [], ZH calculus [], ZW calculus [], etc.). Unlike our work, this line of work focuses on axiomatisations of general quantum theory (e.g. all complete positive maps), of which quantum programs (complete positive trace non-increasing maps) are a proper subset. 

Other works focus on models of quantum computation (e.g. quantum concurrency) that do not yet have a readily available notion of program equivalence. Quantum process calculi \cite{QPAlg_lalire2004processalgebraicapproachconcurrent,CQP_gay2004communicatingquantumprocesses,lqCCS_10.1145/3632885, qCCS_ying2010algebraquantumprocesses}, for instance, model quantum concurrency by extending classical process calculi with quantum primitives. Several notions of behavioural equivalence were then developed based on bisimulation. However, as \cite{lqCCS_10.1145/3632885} points out, each of these proposals is subtly different from the others, and there is not yet a consensus on which is the right one. Compared to our work, quantum bisimulation handles a richer programming language, but once again only handles equivalence. 

% \tw{Help needed for the bit on process calculi.}

% However, as pointed out by [lqccs], however, there is not yet a consensus on what is the `right' notion of quantum bisimulation; each of the proposals differ in subtle ways. 

% Another interesting approach to quantum program equivalence is bisimulation. Notions of quantum bisimulation were originally defined to characterise program equivalence for quantum process calculi [...], used to model quantum concurrency. However, as pointed out by the authors in [], characterising equivalence for quantum process calculi is difficult because of the interactions between concurrency, synchronisation and quantumness. 

\section{Conclusion}
We have defined a sound and complete relational proof system for quantum programs. Our proof system achieves completeness from the duality theorem of quantum optimal theorem. In addition, with the help of infinite-valued predicates, we have given a complete embedding of projective assertions into our logic.

\bibliographystyle{IEEEtran}
\bibliography{refs}

\begin{thebibliography}{10}
\providecommand{\url}[1]{#1}
\csname url@samestyle\endcsname
\providecommand{\newblock}{\relax}
\providecommand{\bibinfo}[2]{#2}
\providecommand{\BIBentrySTDinterwordspacing}{\spaceskip=0pt\relax}
\providecommand{\BIBentryALTinterwordstretchfactor}{4}
\providecommand{\BIBentryALTinterwordspacing}{\spaceskip=\fontdimen2\font plus
\BIBentryALTinterwordstretchfactor\fontdimen3\font minus
  \fontdimen4\font\relax}
\providecommand{\BIBforeignlanguage}[2]{{%
\expandafter\ifx\csname l@#1\endcsname\relax
\typeout{** WARNING: IEEEtran.bst: No hyphenation pattern has been}%
\typeout{** loaded for the language `#1'. Using the pattern for}%
\typeout{** the default language instead.}%
\else
\language=\csname l@#1\endcsname
\fi
#2}}
\providecommand{\BIBdecl}{\relax}
\BIBdecl

\bibitem{qRHL_Unruh_2019}
\BIBentryALTinterwordspacing
D.~Unruh, ``Quantum relational hoare logic,'' \emph{Proceedings of the ACM on
  Programming Languages}, vol.~3, no. POPL, p. 1–31, Jan. 2019. [Online].
  Available: \url{http://dx.doi.org/10.1145/3290346}
\BIBentrySTDinterwordspacing

\bibitem{barthe_rqpd}
G.~Barthe, J.~Hsu, M.~Ying, N.~Yu, and L.~Zhou, ``Relational proofs for quantum
  programs,'' \emph{Proc. ACM Program. Lang.}, vol.~4, no. POPL, Dec. 2019.

\bibitem{qRHLWithExpectations}
Y.~Li and D.~Unruh, ``{Quantum Relational Hoare Logic with Expectations},'' in
  \emph{Proceedings of the 48th International Colloquium on Automata,
  Languages, and Programming}, vol. 198, 2021, pp. 136:1--136:20.

\bibitem{erhl}
\BIBentryALTinterwordspacing
M.~Avanzini, G.~Barthe, D.~Davoli, and B.~Gr\'{e}goire, ``A quantitative
  probabilistic relational hoare logic,'' \emph{Proc. ACM Program. Lang.},
  vol.~9, no. POPL, Jan. 2025. [Online]. Available:
  \url{https://doi.org/10.1145/3704876}
\BIBentrySTDinterwordspacing

\bibitem{Nielsen_Chuang_2010}
M.~A. Nielsen and I.~L. Chuang, \emph{Quantum Computation and Quantum
  Information: 10th Anniversary Edition}.\hskip 1em plus 0.5em minus
  0.4em\relax Cambridge University Press, 2010.

\bibitem{strassen1965existence}
\BIBentryALTinterwordspacing
V.~Strassen, ``The existence of probability measures with given marginals,''
  \emph{The Annals of Mathematical Statistics}, pp. 423--439, 1965. [Online].
  Available: \url{http://projecteuclid.org/euclid.aoms/1177700153}
\BIBentrySTDinterwordspacing

\bibitem{zhou2018quantumcouplingstrassentheorem}
\BIBentryALTinterwordspacing
L.~Zhou, S.~Ying, N.~Yu, and M.~Ying, ``Strassen's theorem for quantum
  couplings,'' \emph{Theoretical Computer Science}, vol. 802, pp. 67--76, 2020.
  [Online]. Available:
  \url{https://www.sciencedirect.com/science/article/pii/S0304397519305225}
\BIBentrySTDinterwordspacing

\bibitem{QuantumOptimalTransport_Cole_2023}
\BIBentryALTinterwordspacing
S.~Cole, M.~Eckstein, S.~Friedland, and K.~{\.{Z}}yczkowski, ``On quantum
  optimal transport,'' \emph{Mathematical Physics, Analysis and Geometry},
  vol.~26, no.~2, p.~14, Jun 2023. [Online]. Available:
  \url{https://doi.org/10.1007/s11040-023-09456-7}
\BIBentrySTDinterwordspacing

\bibitem{monge1781memoire}
G.~Monge, ``M{\'e}moire sur la th{\'e}orie des d{\'e}blais et des remblais,''
  \emph{Mem. Math. Phys. Acad. Royale Sci.}, pp. 666--704, 1781.

\bibitem{Bistron_2023}
\BIBentryALTinterwordspacing
R.~Bistro{\'{n}}, M.~Eckstein, and K.~{\.{Z}}yczkowski, ``Monotonicity of a
  quantum 2-wasserstein distance,'' \emph{Journal of Physics A: Mathematical
  and Theoretical}, vol.~56, no.~9, p. 095301, feb 2023. [Online]. Available:
  \url{https://dx.doi.org/10.1088/1751-8121/acb9c8}
\BIBentrySTDinterwordspacing

\bibitem{mullerhermes2022monotonicity}
\BIBentryALTinterwordspacing
A.~M{\"u}ller-Hermes, ``On the monotonicity of a quantum optimal transport
  cost,'' 2022. [Online]. Available: \url{https://arxiv.org/abs/2211.11713}
\BIBentrySTDinterwordspacing

\bibitem{quantumearthmover}
\BIBentryALTinterwordspacing
L.~Zhou, N.~Yu, S.~Ying, and M.~Ying, ``{Quantum earth mover’s distance, a
  no-go quantum Kantorovich–Rubinstein theorem, and quantum marginal
  problem},'' \emph{Journal of Mathematical Physics}, vol.~63, no.~10, p.
  102201, 10 2022. [Online]. Available: \url{https://doi.org/10.1063/5.0068344}
\BIBentrySTDinterwordspacing

\bibitem{koashi2004monogamy}
M.~Koashi and A.~Winter, ``Monogamy of quantum entanglement and other
  correlations,'' \emph{Physical Review A}, vol.~69, no.~2, p. 022309, 2004.

\bibitem{Ying11}
\BIBentryALTinterwordspacing
M.~Ying, ``Floyd--hoare logic for quantum programs,'' \emph{ACM Trans. Program.
  Lang. Syst.}, vol.~33, no.~6, Jan. 2012. [Online]. Available:
  \url{https://doi.org/10.1145/2049706.2049708}
\BIBentrySTDinterwordspacing

\bibitem{Ying16}
------, \emph{Foundations of Quantum Programming}, 1st~ed.\hskip 1em plus 0.5em
  minus 0.4em\relax San Francisco, CA, USA: Morgan Kaufmann Publishers Inc.,
  2016.

\bibitem{approx_rel_reasoning_quantum_programs}
P.~Yan, H.~Jiang, and N.~Yu, ``Approximate relational reasoning for quantum
  programs,'' in \emph{Computer Aided Verification}, A.~Gurfinkel and
  V.~Ganesh, Eds.\hskip 1em plus 0.5em minus 0.4em\relax Cham: Springer Nature
  Switzerland, 2024, pp. 495--519.

\bibitem{ZhouYuYing2019}
\BIBentryALTinterwordspacing
L.~Zhou, N.~Yu, and M.~Ying, ``An applied quantum hoare logic,'' in
  \emph{Proceedings of the 40th ACM SIGPLAN Conference on Programming Language
  Design and Implementation}, ser. PLDI 2019.\hskip 1em plus 0.5em minus
  0.4em\relax New York, NY, USA: Association for Computing Machinery, 2019, p.
  1149–1162. [Online]. Available:
  \url{https://doi.org/10.1145/3314221.3314584}
\BIBentrySTDinterwordspacing

\bibitem{robustnessquantum2019}
\BIBentryALTinterwordspacing
S.-H. Hung, K.~Hietala, S.~Zhu, M.~Ying, M.~Hicks, and X.~Wu, ``Quantitative
  robustness analysis of quantum programs,'' \emph{Proc. ACM Program. Lang.},
  vol.~3, no. POPL, Jan. 2019. [Online]. Available:
  \url{https://doi.org/10.1145/3290344}
\BIBentrySTDinterwordspacing

\bibitem{Gleipnir2021}
\BIBentryALTinterwordspacing
R.~Tao, Y.~Shi, J.~Yao, J.~Hui, F.~T. Chong, and R.~Gu, ``Gleipnir: toward
  practical error analysis for quantum programs,'' in \emph{Proceedings of the
  42nd ACM SIGPLAN International Conference on Programming Language Design and
  Implementation}, ser. PLDI 2021.\hskip 1em plus 0.5em minus 0.4em\relax New
  York, NY, USA: Association for Computing Machinery, 2021, p. 48–64.
  [Online]. Available: \url{https://doi.org/10.1145/3453483.3454029}
\BIBentrySTDinterwordspacing

\bibitem{Preskill2018quantumcomputingin}
\BIBentryALTinterwordspacing
J.~Preskill, ``Quantum {C}omputing in the {NISQ} era and beyond,''
  \emph{{Quantum}}, vol.~2, p.~79, Aug. 2018. [Online]. Available:
  \url{https://doi.org/10.22331/q-2018-08-06-79}
\BIBentrySTDinterwordspacing

\bibitem{Wilde_2017}
M.~M. Wilde, \emph{Quantum Information Theory}, 2nd~ed.\hskip 1em plus 0.5em
  minus 0.4em\relax Cambridge University Press, 2017.

\bibitem{AKN98}
\BIBentryALTinterwordspacing
D.~Aharonov, A.~Kitaev, and N.~Nisan, ``Quantum circuits with mixed states,''
  in \emph{Proceedings of the Thirtieth Annual ACM Symposium on Theory of
  Computing}, ser. STOC '98.\hskip 1em plus 0.5em minus 0.4em\relax New York,
  NY, USA: Association for Computing Machinery, 1998, p. 20–30. [Online].
  Available: \url{https://doi.org/10.1145/276698.276708}
\BIBentrySTDinterwordspacing

\bibitem{Ying2013Quantum}
M.~Ying, Y.~Feng, and N.~Yu, ``Quantum information-flow security:
  Noninterference and access control,'' in \emph{2013 IEEE 26th Computer
  Security Foundations Symposium}, 2013, pp. 130--144.

\bibitem{AR19}
S.~Aaronson and G.~N. Rothblum, ``Gentle measurement of quantum states and
  differential privacy,'' in \emph{Proceedings of the 51st Annual ACM SIGACT
  Symposium on Theory of Computing}, 2019, p. 322–333.

\bibitem{ZhouYingDiffPriv}
L.~Zhou and M.~Ying, ``Differential privacy in quantum computation,'' in
  \emph{2017 IEEE 30th Computer Security Foundations Symposium (CSF)}, 2017,
  pp. 249--262.

\bibitem{BartheGB09}
G.~Barthe, B.~Gr\'{e}goire, and S.~Zanella~B\'{e}guelin, ``Formal certification
  of code-based cryptographic proofs,'' in \emph{Proceedings of the 36th Annual
  ACM SIGPLAN-SIGACT Symposium on Principles of Programming Languages}, 2009,
  p. 90–101.

\bibitem{Aguirre:POPL:21}
\BIBentryALTinterwordspacing
A.~Aguirre, G.~Barthe, J.~Hsu, B.~L. Kaminski, J.~Katoen, and C.~Matheja, ``A
  pre-expectation calculus for probabilistic sensitivity,'' \emph{Proc. {ACM}
  Program. Lang.}, vol.~5, no. {POPL}, pp. 1--28, 2021. [Online]. Available:
  \url{https://doi.org/10.1145/3434333}
\BIBentrySTDinterwordspacing

\bibitem{GregersenAHTB24}
\BIBentryALTinterwordspacing
S.~O. Gregersen, A.~Aguirre, P.~G. Haselwarter, J.~Tassarotti, and L.~Birkedal,
  ``Asynchronous probabilistic couplings in higher-order separation logic,''
  \emph{Proc. {ACM} Program. Lang.}, vol.~8, no. {POPL}, pp. 753--784, 2024.
  [Online]. Available: \url{https://doi.org/10.1145/3632868}
\BIBentrySTDinterwordspacing

\bibitem{QOT_monge_Zyczkowski_1998}
\BIBentryALTinterwordspacing
K.~Zyczkowski and W.~Slomczynski, ``The monge distance between quantum
  states,'' \emph{Journal of Physics A: Mathematical and General}, vol.~31,
  no.~45, p. 9095, nov 1998. [Online]. Available:
  \url{https://dx.doi.org/10.1088/0305-4470/31/45/009}
\BIBentrySTDinterwordspacing

\bibitem{QOT_monge_Zyczkowski_2001}
\BIBentryALTinterwordspacing
------, ``The monge metric on the sphere and geometry of quantum states,''
  \emph{Journal of Physics A: Mathematical and General}, vol.~34, no.~34, p.
  6689, aug 2001. [Online]. Available:
  \url{https://dx.doi.org/10.1088/0305-4470/34/34/311}
\BIBentrySTDinterwordspacing

\bibitem{QOT_continuous_2016}
F.~Golse, C.~Mouhot, and T.~Paul, ``On the mean field and classical limits of
  quantum mechanics,'' \emph{Communications in Mathematical Physics}, vol. 343,
  no.~1, pp. 165--205, Apr. 2016.

\bibitem{QOT_continuous_2022}
\BIBentryALTinterwordspacing
F.~Golse and T.~Paul, ``{Quantum and Semiquantum Pseudometrics and
  applications},'' \emph{{Journal of Functional Analysis}}, 2022. [Online].
  Available: \url{https://hal.science/hal-03136855}
\BIBentrySTDinterwordspacing

\bibitem{QOT_continuous_2022a}
\BIBentryALTinterwordspacing
E.~Caglioti, F.~Golse, and T.~Paul, ``{Towards Optimal Transport for Quantum
  Densities},'' \emph{{Annali della Scuola Normale Superiore di Pisa, Classe di
  Scienze}}, 2022. [Online]. Available: \url{https://hal.science/hal-01963667}
\BIBentrySTDinterwordspacing

\bibitem{QOT_continuous_2020}
------, ``Quantum optimal transport is cheaper,'' \emph{Journal of Statistical
  Physics}, vol. 181, no.~1, pp. 149--162, Oct. 2020.

\bibitem{QOT_channel_2021}
G.~De~Palma and D.~Trevisan, ``Quantum optimal transport with quantum
  channels,'' \emph{Annales Henri Poincar{\'e}}, vol.~22, no.~10, pp.
  3199--3234, Oct. 2021.

\bibitem{QOT_channel_2024a}
\BIBentryALTinterwordspacing
------, \emph{Quantum Optimal Transport: Quantum Channels and Qubits}.\hskip
  1em plus 0.5em minus 0.4em\relax Cham: Springer Nature Switzerland, 2024, pp.
  203--239. [Online]. Available:
  \url{https://doi.org/10.1007/978-3-031-50466-2_4}
\BIBentrySTDinterwordspacing

\bibitem{QOT_channel_2023}
\BIBentryALTinterwordspacing
G.~De~Palma, M.~Marvian, C.~Rouz\'e, and D.~S. Fran\ifmmode~\mbox{\c{c}}\else
  \c{c}\fi{}a, ``Limitations of variational quantum algorithms: A quantum
  optimal transport approach,'' \emph{PRX Quantum}, vol.~4, p. 010309, Jan
  2023. [Online]. Available:
  \url{https://link.aps.org/doi/10.1103/PRXQuantum.4.010309}
\BIBentrySTDinterwordspacing

\bibitem{QOT_channel_2024}
\BIBentryALTinterwordspacing
G.~D. Palma and D.~Pastorello, ``Quantum concentration inequalities and
  equivalence of the thermodynamical ensembles: an optimal mass transport
  approach,'' 2024. [Online]. Available: \url{https://arxiv.org/abs/2403.18617}
\BIBentrySTDinterwordspacing

\bibitem{Friedland2022Quantum}
S.~Friedland, M.~Eckstein, S.~Cole, and K.~\ifmmode~\dot{Z}\else
  \.{Z}\fi{}yczkowski, ``Quantum monge-kantorovich problem and transport
  distance between density matrices,'' \emph{Physical Review Letters}, vol.
  129, p. 110402, Sep 2022.

\bibitem{QOT_nogo_2022}
\BIBentryALTinterwordspacing
L.~Zhou, N.~Yu, S.~Ying, and M.~Ying, ``Quantum earth mover’s distance, a
  no-go quantum kantorovich–rubinstein theorem, and quantum marginal
  problem,'' \emph{Journal of Mathematical Physics}, vol.~63, no.~10, p.
  102201, 10 2022. [Online]. Available: \url{https://doi.org/10.1063/5.0068344}
\BIBentrySTDinterwordspacing

\bibitem{abramsky2007categoricalsemanticsquantumprotocols}
\BIBentryALTinterwordspacing
S.~Abramsky and B.~Coecke, ``A categorical semantics of quantum protocols,''
  2007. [Online]. Available: \url{https://arxiv.org/abs/quant-ph/0402130}
\BIBentrySTDinterwordspacing

\bibitem{PQP_Coecke_Kissinger_2017}
B.~Coecke and A.~Kissinger, \emph{Picturing Quantum Processes: A First Course
  in Quantum Theory and Diagrammatic Reasoning}.\hskip 1em plus 0.5em minus
  0.4em\relax Cambridge University Press, 2017.

\bibitem{heunen_vicary_book}
\BIBentryALTinterwordspacing
C.~Heunen and J.~Vicary, \emph{{Categories for Quantum Theory: An
  Introduction}}.\hskip 1em plus 0.5em minus 0.4em\relax Oxford University
  Press, 11 2019. [Online]. Available:
  \url{https://doi.org/10.1093/oso/9780198739623.001.0001}
\BIBentrySTDinterwordspacing

\bibitem{ZX_1_10.1007/978-3-540-70583-3_25}
B.~Coecke and R.~Duncan, ``Interacting quantum observables,'' in
  \emph{Automata, Languages and Programming}, L.~Aceto, I.~Damg{\aa}rd, L.~A.
  Goldberg, M.~M. Halld{\'o}rsson, A.~Ing{\'o}lfsd{\'o}ttir, and
  I.~Walukiewicz, Eds.\hskip 1em plus 0.5em minus 0.4em\relax Berlin,
  Heidelberg: Springer Berlin Heidelberg, 2008, pp. 298--310.

\bibitem{ZH_Calculus_Backens_2019}
\BIBentryALTinterwordspacing
M.~Backens and A.~Kissinger, ``Zh: A complete graphical calculus for quantum
  computations involving classical non-linearity,'' \emph{Electronic
  Proceedings in Theoretical Computer Science}, vol. 287, p. 23–42, Jan.
  2019. [Online]. Available: \url{http://dx.doi.org/10.4204/EPTCS.287.2}
\BIBentrySTDinterwordspacing

\bibitem{ZW_calculus_hadzihasanovic2015diagrammaticaxiomatisationqubitentanglement}
\BIBentryALTinterwordspacing
A.~Hadzihasanovic, ``A diagrammatic axiomatisation for qubit entanglement,''
  2015. [Online]. Available: \url{https://arxiv.org/abs/1501.07082}
\BIBentrySTDinterwordspacing

\bibitem{Xu_Barthe_Zhou_2025}
\BIBentryALTinterwordspacing
Y.~Xu, G.~Barthe, and L.~Zhou, ``Automating equational proofs in dirac
  notation,'' \emph{Proc. ACM Program. Lang.}, vol.~9, no. POPL, Jan. 2025.
  [Online]. Available: \url{https://doi.org/10.1145/3704878}
\BIBentrySTDinterwordspacing

\bibitem{statonpopl15}
\BIBentryALTinterwordspacing
S.~Staton, ``Algebraic effects, linearity, and quantum programming languages,''
  in \emph{Proceedings of the 42nd Annual ACM SIGPLAN-SIGACT Symposium on
  Principles of Programming Languages}, ser. POPL '15.\hskip 1em plus 0.5em
  minus 0.4em\relax New York, NY, USA: Association for Computing Machinery,
  2015, p. 395–406. [Online]. Available:
  \url{https://doi.org/10.1145/2676726.2676999}
\BIBentrySTDinterwordspacing

\bibitem{Quantum_PI_Carette_2024}
\BIBentryALTinterwordspacing
J.~Carette, C.~Heunen, R.~Kaarsgaard, and A.~Sabry, ``With a few square roots,
  quantum computing is as easy as pi,'' \emph{Proceedings of the ACM on
  Programming Languages}, vol.~8, no. POPL, p. 546–574, Jan. 2024. [Online].
  Available: \url{http://dx.doi.org/10.1145/3632861}
\BIBentrySTDinterwordspacing

\bibitem{QPAlg_lalire2004processalgebraicapproachconcurrent}
\BIBentryALTinterwordspacing
M.~Lalire and P.~Jorrand, ``A process algebraic approach to concurrent and
  distributed quantum computation: Operational semantics,'' 2004. [Online].
  Available: \url{https://arxiv.org/abs/quant-ph/0407005}
\BIBentrySTDinterwordspacing

\bibitem{CQP_gay2004communicatingquantumprocesses}
\BIBentryALTinterwordspacing
S.~Gay and R.~Nagarajan, ``Communicating quantum processes,'' 2004. [Online].
  Available: \url{https://arxiv.org/abs/quant-ph/0409052}
\BIBentrySTDinterwordspacing

\bibitem{lqCCS_10.1145/3632885}
\BIBentryALTinterwordspacing
L.~Ceragioli, F.~Gadducci, G.~Lomurno, and G.~Tedeschi, ``Quantum bisimilarity
  via barbs and contexts: Curbing the power of non-deterministic observers,''
  \emph{Proc. ACM Program. Lang.}, vol.~8, no. POPL, jan 2024. [Online].
  Available: \url{https://doi.org/10.1145/3632885}
\BIBentrySTDinterwordspacing

\bibitem{qCCS_ying2010algebraquantumprocesses}
\BIBentryALTinterwordspacing
M.~Ying, Y.~Feng, R.~Duan, and Z.~Ji, ``An algebra of quantum processes,''
  2010. [Online]. Available: \url{https://arxiv.org/abs/0707.0330}
\BIBentrySTDinterwordspacing

\bibitem{Gilles19Relational}
G.~Barthe, J.~Hsu, M.~Ying, N.~Yu, and L.~Zhou, ``Relational proofs for quantum
  programs,'' \emph{Proceedings of the ACM on Programming Languages}, vol.~4,
  no. POPL, Dec. 2019.

\bibitem{Mel24}
A.~A. Mele, ``Introduction to {H}aar {M}easure {T}ools in {Q}uantum
  {I}nformation: {A} {B}eginner's {T}utorial,'' \emph{{Quantum}}, vol.~8, p.
  1340, 2024.

\bibitem{Villani2008Optimal}
C.~Villani, \emph{Optimal transport -- Old and new}.\hskip 1em plus 0.5em minus
  0.4em\relax Springer Berlin, Heidelberg, 2008, vol. 338.

\end{thebibliography}

\newpage 
\onecolumn

\begin{appendices}

\section{Deferred Proofs in ``Notations and Preliminaries'' Section}
\label{app: preliminary}
We first briefly review some basic concepts and propositions
in linear algebra and quantum computing.
% that are needed
% in the following proofs.

\textit{Quantum States}. 
The state space of a 
quantum system is described by a complex
Hilbert space
$\mathcal{H}$, which we ususally assume to be 
finite-dimensional. 
A pure state is a unit (column) vector
in the Hilbert space. For example, 
an $d$-dimensional quantum state 
in $\mathbb{C}^d$
has the form
\[
    v = (v_0, v_1, v_2, \dots, v_{d-1})^{\intercal},
\]
usually denoted as $\ket{v}$ with 
the Dirac symbol $\ket{\cdot}$.
The computational basis of $\mathbb{C}^d$
is denoted as $\{\ket{i}\}_{i=0}^{d-1}$,
where the $j$-th coordinate of $\ket{i}$
is $1$ if $j = i$, and $0$ otherwise.
The inner product between states $\ket{u}$
and $\ket{v}$ is denoted as $\<u|v\>$,
which is the standard inner product
in the Hilbert space,
where $\bra{u}$ stands for the 
conjugate transpose of $\ket{u}$.
Let $\ket{u}\in \mathbb{C}^{d_1}$
and $\ket{v}\in \mathbb{C}^{d_2}$
be two states.
Their tensor product, written as 
$\ket{u}\otimes \ket{v}$, or $\ket{u}\ket{v}$
in short, 
is defined to as the following vector
$\ket{u}\ket{v} = 
(u_0v_0, u_0v_1, u_0v_2, \dots, u_{d_1}v_{d_2})^{\intercal}$.

A quantum bit (or qubit for short),
is a quantum state $\ket{\phi}
= \alpha\ket{0} + \beta \ket{1} \in \mathbb{C}^2$,
with $\abs{\alpha}^2 + \abs{\beta}^2 = 1$.
An $n$-qubit state is in the Hilbert space
of dimension $2^n$.

Let $\HH$ be a $d$-dimensional Hilbert space.
The (linear) operators on $\HH$ are just 
$d\times d$ matrices. For an operator $A$, the trace of $A$, $\tr(A)$
is defined as $\sum_{j=0}^{d-1}\bra{j}A\ket{j}$.
It can be shown that $\tr(AB) = \tr(BA)$
and thus $\tr(UAU^{\dagger}) = \tr(A)$
holds for linear operators $A, B$ and unitary 
operator $U$, meaning that 
the trace of an operator does not rely on
the choice of basis.
A density operator $\rho$ in $\HH$, is a positive semi-definite operator with trace $1$.
Applying the spectral decomposition theorem, 
we could write it as $\rho = \sum_i p_i\ket{\phi_i}\bra{\phi_i}$,
meaning that it can be regarded as a 
distribution over pure states.
We could also regard pure states $\ket{\phi}$
as rank $1$ density matrices $\ket{\phi}\bra{\phi}$.
We call a positive semi-definite operator
with trace no more than $1$ a partial density operator.

\textit{Quantum Operations}. Evolutions on
pure quantum states on $\HH$ 
are described by the unitary operators $U$ on $\HH$ with $UU^{\dagger} = U^{\dagger}U =I$.
Given any state $\ket{\phi}$,
the evolution of $U$ gives $U\ket{\phi}$.
A more general concept of quantum operations
is quantum channel,
which could be seen as completely-positive
and trace-preserving
linear maps from the set of density operators
in a Hilbert space into 
another set of density operators in some Hilbert space.
For a linear map $\EE: \DD(\HH)\to \DD(H)$, we say 
it is completely positive if 
for all $\HH'$,
the map 
$\EE \otimes I$ maps 
positive semi-definite
operators on $\HH\otimes \HH'$ to positive semi-definite operators,
and it is trace-preserving
if $\tr(\EE(\rho)) = \tr(\rho)$
for any density operator $\rho$.

In this work, we only consider projective measurements.
For a projective measurement $\mathcal{M} = \{P_{i}\}$,
we require all the $P_i$'s are projectors,
and $\sum_i P_i = I$.
Given any density operator $\rho$,
after performing the projective measurement $\mathcal{M} = \{P_{i}\}$,
we will observe the event $i$
with probability $\tr(P_i \rho)$,
with the post measured state 
$P_i\rho P_i/\tr(P_i \rho)$.

An observable, or a quantum predicate,
is 
a positive semi-definite operator 
on $\HH$.
For an observable $A$,
its spectral decomposition
can be written as
$A = \sum \lambda P_{\lambda}$,
where $P_{\lambda}$ is 
the projector onto the eigenspace
corresponding to the eigenvalue $\lambda$.
The 
expection of $A$
in a state $\rho$
is given by 
$\sum \lambda \tr(P_{\lambda} \rho) = \tr(A\rho)$.

The notion partial trace is a special quantum 
operation that discard the state of a system.
Formally, 
for
the Hilbert space $\HH_1\otimes \HH_2$,
we define partial trace over $\HH_1$,
written as $\tr_1$
as a mapping from 
$\Pos(\HH_1\otimes \HH_2)$
to 
$\Pos(\HH_2)$,
that satisfies
\[
\tr_1(A) = \sum_i (\bra{i}\otimes I) A (\ket{i}\otimes I)
\]
where $\ket{i}$ ranges over all computational
basis of $\HH_1$.

For more introductions and explanations
of the above notions,
we refer the readers to \cite{Nielsen_Chuang_2010}.

We now introduce some concepts
that are
useful in the study of quantum programs.

The L\"owner order is a partial order of  positive semi-definite
operators, defined as follows:
for $A, B\in \Pos(\HH)$,
$A\sqsubseteq B$ if and only if $B-A$
is positive semi-definite.

\begin{definition}[Support]\label{def:support}
	Let $P\in \Pos(\mathcal{H})$ be a positive semi-definite operator.
  Then, the support of $P$, denoted as $\supp(P)$, is the subspace of $\mathcal{H}$
  that are spanned by the eigenvectors of $P$ associated with non-zero
  eigenvalues.
\end{definition}

\begin{definition}\label{def:vee-wedge-perp}
  Let $\mathcal{H}$ be a Hilbert space, 
  $X$ and $Y$ be its subspaces.
  The join of $X$ and $Y$ is defined as
  $X\vee Y = \overline{\spanv \{X\cup Y\}}$, 
  where $\overline{\cdot}$ represents the operation 
  of taking the topological closure.
  The meet of $X$ and $Y$ is defined as
  $X\wedge Y = X\cap Y$.
  The orthogonal complement of $X$ is
  $X^{\bot} = \{\ket{\psi}\in \mathcal{H}| \ket{\psi} \perp \ket{\phi}, \forall \ket{\phi}\in X \}$.
\end{definition}

Here are some of the properties of the support 
that are needed for proofs in the subsequent parts.
\begin{proposition}[Properties of the Support]\label{prop:support}
  We have the following properties:
  \begin{itemize}
    \item Let $P, Q\in \Pos(\HH)$, then $\supp(P+Q) = \supp(P) \vee \supp(Q)$;
    \item Let $X_1$ and $X_2$ be subspaces of 
    a Hilbert space $\HH$, then
    $\mathcal{E}(X_{1}\vee X_{2}) = \mathcal{E}(X_{1}) \vee \mathcal{E}(X_{2})$ 
    for any CP maps $\mathcal{E}$.
  \end{itemize}
\end{proposition}

\noindent\textit{Infinite-valued predicates.} \\[-0.2cm]

In the following, we introduce the basic notions
that are related to the infinite-valued predicates.
We first formally define the notion of
infinite-valued predicates as follows.

%\lz{re-edit once preliminary is done.}
\begin{definition}[Infinite-Valued Predicates]%
\label{def:infinite-valued-predicates}
    Given a Hilbert space $\mathcal{H}$,
$A$ is called an infinite-valued 
predicate on $\mathcal{H}$, 
if it has 
a unique spectral decomposition 
$\{(\lambda_i,X_i)\}_i$,
where $\lambda_i\in \mathbb{R}^{+\infty}$ are its eigenvalues, 
$X_i$'s are projections onto eigenspaces 
that are pairwise orthogonal, 
with $\sum_i X_i = I$. 

The set of all infinite-valued 
predicates is denoted as
$\PosI(\HH)$.
\end{definition}

For the arithmetic operations
related to the $+\infty$,
we make the conventions 
that 
$(+\infty) \cdot 0 = 0 \cdot (+\infty) = 0$,
$(+\infty) + a = a + (+\infty) = +\infty$ for $a\in \mathbb{R}^{+\infty}$, 
$0 / 0 = 0$ (this is used in normalization of quantum states, i.e., $\rho = \tr(\rho)(\rho / \tr(\rho))$ even if $\rho = 0$),
and $+\infty\le +\infty$.

% We define $\tr(A) = \sum_i\lambda_i \dim(X_i)$.

The following lemma enables us 
to represent the infinite-valued predicates 
as two parts, 
namely the ``finite'' and ``infinite'' part.
\begin{lemma}%
\label{lem:IVP-rep}
  For any $A\in \PosI(\HH)$, it can be uniquely represented as $(P_A,X_A)$, written $A\triangleq(P_A,X_A)$, where $P_A\in\Pos(\HH)$, $X_A\in \cS(\HH)$ such that $P_A X_A = 0$, and we will write it as $A = P_A + (+\infty\cdot X_A)$.
\end{lemma}
\begin{proof}
  Suppose $A = \sum_{j} \lambda_{j} X_{j}$.
  We set $P_{A} = \sum_{j: \lambda_{j} < +\infty} \lambda_{j} X_{j}$
  and $X_{A} = \sum_{j:\lambda_{j} = +\infty} X_{j}$.
  It is clear that $P_{A} X_{A} = 0$.
  Now we prove the uniqueness.
  Suppose $A$ could also be written as $P_{A}'+ \infty \cdot X_{A}'$.
  We first prove $X_{A} = X_{A}'$.
  If not, there must exist a non-zero vector $\ket{\psi}\in X_{A}\cap X_{A}^{\prime\bot}$,
  which satisfies $\bra{\psi} A\ket{\psi} = \bra{\psi} P_{A}  \ket{\psi} + \infty  \bra{\psi} X_{A}\ket{\psi} = +\infty $,
  and $\bra{\psi} A\ket{\psi} = \bra{\psi} P_{A}'  \ket{\psi} + \infty  \bra{\psi} X_{A}'\ket{\psi} = \bra{\psi} P_{A}'  \ket{\psi} < +\infty $,
  a contradiction.
  Then, we have $P_{A} = X_{A}^{\bot} A X_{A}^{\bot} = X_{A}^{\prime \bot} A X_{A}^{\prime\bot} = P_{A}'$ as we desired.
\end{proof}

Here we introduces some basic operations for infinite-valued predicates:

\begin{definition}[Operations of Infinite-Valued Predicate]%
\label{def:operations-infinite-valued-pre}
  The basic
  operations of infinite-valued predicates can be defined as follows.
    \begin{itemize}
    \item The addition for two infinite-valued predicates $A_1, A_2$ is:
    \[ A_1 + A_2 \triangleq ( X^{\perp}\left(P_{A_1} + P_{A_2} \right)X^{\perp}, X),
    \] where $X = X_{A_1} \vee X_{A_2}$.
    \item The tensor product for two infinite-valued predicates $A_1, A_2$ 
    is
    \[ A_1 \otimes A_2 \triangleq \left( P_{A_1} \otimes P_{A_2}, X \right),
    \] where
    $X = \left( \supp \left( P_{A_{1}} \right)\otimes X_{A_{2}}\right)\vee \left( X_{A_{1}} \otimes \supp \left( P_{A_{2}} \right) \right)\vee \left(X_{A_{1}}\otimes X_{A_{2}}\right)$.
    \item For any $\ket{\psi} \in \HH$
    and $A\in \PosI(\HH)$
    with spectral decomposition
    $\{(\lambda_i, X_i)\}$, the inner product $\bra{\psi}A\ket{\psi}$ is defined as
\[
  \bra{\psi}A\ket{\psi} \triangleq \sum_i\lambda_i \bra{\psi}X_{i}\ket{\psi}.
\]
    \item For a density operator $\rho \in \DD(\HH)$, its expectation value on
   an infinite-valued predicate $A$ is 
    \[ \tr (A \rho) \triangleq
    \begin{cases} \tr (P_A\rho),\text{ if } X_A\rho = 0 \\ \infty, \text{
otherwise}
    \end{cases}
    \]
    \item For subspace $X$, $X\mid A \triangleq X\cdot A\cdot X + (+\infty\cdot X^\bot)$, or equivalently,
    $X\mid A \triangleq ((X\vee X_A^\bot) P_{A}(X\vee X_A^\bot), X^\bot\vee X_A)$.
\end{itemize}
\end{definition}

\begin{lemma}[Basic Properties of Operations of Infinite-Valued Predicates]
\label{lemma:basic-properties-of-ops-of-ivp}
We have the following properties for $A, A_1, A_2\in \PosI(\HH)$:
\begin{itemize}
  \item Scalar product $cA$ for $c\in \mathbb{R}^{+\infty}$ is defined such that for all $\ket{\psi}$, $\bra{\psi}cA\ket{\psi}= c\bra{\psi}A\ket{\psi}$.
  \item Addition $A_1+A_2$ such that for all $|\psi\>\in \HH$, $\<\psi|(A_1+A_2)|\psi\> = \<\psi|A_1|\psi\>+\<\psi|A_2|\psi\>$.
  \item Tensor product $A_1\otimes A_2$ such that for all $|\psi_1\>, |\psi_2\>$, $(\<\psi_1|\otimes\<\psi_2|)(A_1\otimes A_2)(|\psi_1\>\otimes|\psi_2\>) = (\<\psi_1|A_1|\psi_1\>)\cdot(\<\psi_2|A_2|\psi_2\>)$.
  \item Let $M$ be a linear operator with $\HH$ as its domain, $M^\dagger AM$ can be defined such that for all $|\psi\>$, $\<\psi|(M^\dagger AM)|\psi\> = \<\phi|A|\phi\>$ where $|\phi\> = M|\psi\>$.
  \item For $P\in \Pos$ 
  with decomposition $P = \sum_ia_i|\psi_i\>\<\psi_i|$ ($0\le a_i$), 
  the trace is $\tr(AP) = \sum_i a_i \<\psi_i| A |\psi_i\>$. 
  Note that the value is unique for any decomposition.
        %\lz{Minbo: help prove it.}
  \item For $\EE\in \QO$ (more generally, CP maps) with Kraus operators $\{E_i\}$, $\EE^\dagger(A) = \sum_i E_i^\dagger A E_i$. Note that it is unique for arbitrary Kraus operators.
        %\lz{Minbo: help prove it.}
  \item $A_1 = A_2$ if for all $|\psi\>$, $\bra{\psi}A_1\ket{\psi} = \bra{\psi}A_2\ket{\psi}$.
  %\item $A_{1}\otimes A_{2} = B_{1}\otimes B_{2}$ if for all $\ket{\psi_{1}}\otimes \ket{\psi_{2}}$,
  \item $A_1\sqsubseteq A_2$ if for all $|\psi\>$, $\<\psi|A_1|\psi\> \le \<\psi|A_2|\psi\>$.
\end{itemize}
\end{lemma}

For readability, we postpone the proof of the above lemma
to Appendix~\ref{sec:postponed-technical-proofs}.

\begin{lemma}[Algebraic Properties]%
\label{lem: IVP algebraic}
  In the following, 
  let $a,b,c\in \mathbb{R}^{+\infty}$, $A,A_1,A_2 \in\PosI$, $M,M_1,M_2,\cdots\in \LL$, and $P,P_1,P_2\cdots\in\Pos$.
  We have the following properties:
  \begin{itemize}
    \item $0A = 0$, $1A = A$, $a(bA) = (ab)A$;
    \item $0+A = A+0 = A$, $A_1 + A_2 = A_2 + A_1$, $A_1 + (A_2 + A_3) = (A_1 + A_2) + A_3$;
    \item $0\otimes A = A\otimes 0 = 0$, $A_1\otimes (A_2 \otimes A_3) = (A_1 \otimes A_2) \otimes A_3$;
    \item $A\otimes (cA_1 + A_2) = c(A\otimes A_1) + (A\otimes A_2)$;
          $(cA_1 + A_2)\otimes A = c(A_1\otimes A) + (A_2\otimes A)$;
    \item $0^\dagger A0 = 0$, $M_2^\dagger (M_1^\dagger AM_1) M_2 = (M_1M_2)^\dagger A (M_1M_2)$; 
          $M^\dagger(cA_1+A_2)M = c(M^\dagger A_1M) + M^\dagger A_2M$;
    \item $(M_1\otimes M_2)^\dagger (A_1\otimes A_2) (M_1\otimes M_2) = 
           (M_1^\dagger A_1M_1)\otimes (M_2^\dagger A_2M_2)$;
    \item $\tr(A(cP_1+P_2)) = c\tr(AP_1) + \tr(AP_2)$;
          $\tr((cA_1 + A_2)P) = c\tr(A_1P) + \tr(A_2P)$;
    \item $\tr((A_1\otimes A_2)(P_1\otimes P_2)) = \tr(A_1P_1)\tr(A_2P_2)$;
          $\tr((M^\dagger AM)P) = \tr(A (MPM^\dagger))$.
    \item $\tr((A\otimes I)P) = \tr(A\tr_2(P))$;
          $\tr((I\otimes A)P) = \tr(A\tr_1(P))$;
    \item $\tr(A\ket{\phi}\bra{\phi}) = \bra{\phi}A\ket{\phi}$.
    \item $A_1=A_2$ iff for all $P\in\Pos$ (or $P\in\DD$) such that $\tr(A_1 P) = \tr(A_2 P)$;
    \item $A_1\sqsubseteq A_2$ iff for all $P\in\Pos$ (or $P\in\DD$) such that $\tr(A_1 P) \le \tr(A_2 P)$;
    \item $A_1\sqsubseteq A_2$ implies $M^\dagger A_1M\sqsubseteq M^\dagger A_2M$; $A_1\sqsubseteq A_2$ and $A_3\sqsubseteq A_4$ implies $cA_1+A_3\sqsubseteq cA_2+A_4$.
  \end{itemize}
  As direct corollaries, for CP map $\EE, \EE_1,\EE_2$,
  \begin{itemize}
    \item $\tr(A\EE(P)) = \tr(\EE^\dagger(A) P)$;
          $A_1\sqsubseteq A_2$ implies $\EE (A_1)\sqsubseteq \EE (A_2)$;
    \item $(c\EE_1 + \EE_2)(A) = c\EE_1(A) + \EE_2(A)$;
          $\EE(cA_1+A_2) = c\EE(A_1) + \EE(A_2)$;
    \item $\EE_2(\EE_1(A)) = (\EE_2\circ\EE_1)(A)$;
          $(\EE_1\otimes\EE_2)(A_1\otimes A_2) = \EE_1(A_1)\otimes \EE_2(A_2)$.
  \end{itemize}
\end{lemma}

For readability, we postpone the proof of the above lemma
to Appendix~\ref{sec:postponed-technical-proofs}.

\section{Deferred Proofs in ``Quantum Coupling'' Section}
\label{app: coupling}

%\subsection{Deferred Proofs in ``Basic Definition and Duality Theorems''}

\begin{theorem}[\Cref{thm:quantum strassen defect alter}]
%\label{thm:quantum strassen defect alter-appendix}
  For any $\rho_1 \in \DD(\HH_1)$ and $\rho_2 \in \DD(\HH_2)$ with $\tr(\rho_1) = \tr(\rho_2)$, for any defect $\epsilon \in \RR^{+\infty}$ and for any $X \in \Pos(\HH_1 \otimes \HH_2)$, the following are equivalent:
  \begin{enumerate}
    \item $\rho_1 X^\#_\epsilon \rho_2$;
    \item For any $Y_1 \in \Pos(\HH_1)$ and $Y_2 \in \Pos(\HH_2)$ such that $X \geqlow Y_1 \otimes I_2 - I_1 \otimes Y_2$, it holds that 
    \[
      \tr(Y_1\rho_1) \leq \tr(Y_2\rho_2) + \epsilon
    \]
  \end{enumerate}
\end{theorem}

\begin{proof}
%[Proof of \Cref{thm:quantum strassen defect alter}]
  If $\epsilon = +\infty$, then both (1) and (2) trivially hold. So we consider the case that $\epsilon\in \RR^+$, i.e., $\epsilon$ is finite.
  \begin{itemize}
      \item  $(1 \Longrightarrow 2)$. Suppose $\rho_1 X^\#_\epsilon \rho_2$, let $\rho : \< \rho_1,\rho_2 \>$ be the witness such that $\tr(X \rho) \leq \epsilon$. Then for any Hermitian $Y_1, Y_2$, if $X \geq Y_1 \otimes I - I \otimes Y_2$, we have $\tr(Y_1\rho_1) = \tr((Y_1 \otimes I)\rho) \leq \tr((X + I\otimes Y_2)\rho) \leq \tr((I \otimes Y_2)\rho)+\epsilon = \tr(Y_2\rho_2) + \epsilon$, where the second last step uses the assumption.
      \item $(2 \Longrightarrow 1)$. In this part of the proof, we write $\langle A, B \rangle$ to mean the Hilbert-Schmidt inner product $\langle A, B \rangle \eqdef \tr(A^\dagger B)$. The original proof \cite{zhou2018quantumcouplingstrassentheorem} considers a semi-definite program $(\Phi, A, B)$. The primal formulation is that of maximising $\langle A, Z\rangle$, subject to $\Phi(Z) = B, Z \in \Pos(\HH_1 \otimes \HH_2)$, and the dual one is that of minimising $\langle B, Y\rangle$ subject to $\Phi^\dagger(Y) \geq A, Y \in \Herm(\HH_1 \oplus \HH_2)$, where:
      \begin{gather*}
        A = I - X, B = \begin{pmatrix}
          \rho_1 & \\ & \rho_2
        \end{pmatrix}\\
        \Phi(Z) = \begin{pmatrix}
          \tr_2(Z) & \\ & \tr_1(Z)
        \end{pmatrix}\\
        \Phi^\dagger (Y) = \Phi^\dagger \begin{pmatrix}
          Y_1 & \cdot \\ \cdot & Y_2
        \end{pmatrix} = Y_1 \otimes I_2 + I_1 \otimes Y_2
      \end{gather*}
      The above formulation can be shown to satisfy strong duality, meaning that the optima for the primal and dual problems exist and are equal.
      
      Then, let us consider for all Hermitians $Y_1\in \Herm(\HH_1)$, $Y_2 \in \Herm(\HH_2)$ satisfying $Y_1 \otimes I_2 + I_1 \otimes Y_2 \geq I - X$, observe that
      \begin{align*}
          \<B,Y\> &= \tr(Y_1\rho_1 + Y_2\rho_2) \\
          &= \tr(Y_2' \rho_2) - \tr( Y_1' \rho_1) + \tr(\rho_1) \\
          &\geq \tr(\rho_1) - \epsilon
      \end{align*}
      where $Y_2' = Y_2 + nI$, $Y_1' = (n+1)I - Y_1$ with sufficiently large $n\in \RR^+$ (e.g., bigger than all singular values of $Y_1$ and $Y_2$) such that $Y_2'$ and $Y_1'$ are both positive. The second line is derived by using the assumption $\tr(\rho_1) = \tr(\rho_2)$. By condition (2) since $Y_1'\otimes I_2 - I_1\otimes Y_2' = I - (Y_1\otimes I_2 + I_1\otimes Y_2) \sqsubseteq I - (I - X) = X$, we have $\tr(Y_2' \rho_2) - \tr( Y_1' \rho_1)\geq -\epsilon$, and this leads to third line. \\
      By strong duality, we have $\<I - X, Z_{\max}\> \ge \tr(\rho_1) - \epsilon$, or equivalently, $\tr(Z_{\max}) - \tr(X Z_{\max}) \ge \tr(\rho_1) - \epsilon$. Since $\tr(Z_{\max}) = \tr(\rho_1)$, 
      we have $\tr(X Z_{\max}) \le \epsilon$,
      which says that $Z_{\max}$ is a witness of the lifting $\rho_1 X^\#_\epsilon \rho_2$.
  \end{itemize}
\end{proof}

%\subsection{Deferred Proofs in ``Partial Couplings''}
\begin{lemma}[\Cref{lemma:traceequiv}]
%\label{lemma:traceequiv}
  Let $\rho: \langle \rho_1, \rho_2 \rangle$. Then, $\tr(\rho) = \tr(\rho_1) = \tr(\rho_2)$.
\end{lemma}

\begin{proof}
%[Proof of \Cref{lemma:traceequiv}]
  This is direct by noting
  $\tr(\rho) = \tr(\tr_1(\rho)) = \tr(\rho_1) = \tr(\tr_2(\rho)) = \tr(\rho_2)$.
\end{proof}

% \noindent\textit{$\star$-Coupling and Proof of Theorem \ref{thm: strassen partial coupling}} \\[-0.2cm]

Before proving the relationship between $\star$-coupling and partial coupling and its variant of Strassen's theorem, i.e., \Cref{lemma: relation star coupling} and \Cref{thm: strassen partial coupling}, we first introduce some useful definitions:
\begin{itemize}
  \item For any Hilbert space $\HH$, we additionally extend it to $\HH^\star$ with one-dimension denoted by $|\star\>$. Let $P_\star = |\star\>\<\star|$ the projection of $\star$ space and $P_\star^\bot = I_\star - P_\star$ the projection to original space. (To avoid ambiguity, we write $I_\star$ for the identity of $\HH^\star$.)
  \item For any $\rho\in\DD(\HH)$, we define the star-extension $\rho^\star \triangleq (1-\tr(\rho))|\star\>\<\star| + \rho \in \DD^1(\HH^\star)$. Obviously, $P_\star^\bot \rho^\star P_\star^\bot = \rho$.
  \item For $\rho_1\in\DD(\HH_1)$ and $\rho_2\in\DD(\HH_2)$, we say $\rho\in\DD(\HH_1^\star\otimes\HH_2^\star)$ is a $\star$-coupling of $\rho_1$ and $\rho_2$, written $\rho:\<\rho_1,\rho_2\>_\star$, if $\rho:\<\rho_1^\star,\rho_2^\star\>$.
  \item For any $A\in\Pos(\HH_1\otimes\HH_2)$, we define $A^\star \in \Pos(\HH_1^\star\otimes\HH_2^\star)$ as the embedding of $A$.
  \item For $\rho\in\DD(\HH_1^\star\otimes\HH_2^\star)$, we define the projection:
  $$\Pi_{\star}^\bot(\rho) = (P_\star^\bot\otimes P_\star^\bot)\rho(P_\star^\bot\otimes P_\star^\bot) \in \DD(\HH_1\otimes\HH_2).$$
  \item For any $\rho\in\DD(\HH_1\otimes\HH_2)$ such that $\rho : \<\rho_1,\rho_2\>_p$, define 
  \begin{align*}
    \rho^\uparrow = &(1+\tr(\rho) - \tr(\rho_1) - \tr(\rho_2))|\star\star\>\<\star\star|\ + \\ 
    & (\rho_1 - \tr_2(\rho)) \otimes |\star\>\<\star| + |\star\>\<\star|\otimes (\rho_2 - \tr_1(\rho)) + \rho.
  \end{align*}
\end{itemize}

\begin{proposition}[Relation to $\star$-coupling]
  \label{lemma: relation star coupling}
  For given $\rho_1\in\DD(\HH_1)$, $\rho_2\in\DD(\HH_2)$ and $A\in\Pos(\HH_1\otimes\HH_2)$, we claim that:
  \begin{enumerate}
    \item Any $\star$-coupling provide a partial-coupling, i.e., $\Pi_{\star}^\bot(\rho) : \<\rho_1,\rho_2\>_p$ if $\rho : \<\rho_1,\rho_2\>_\star$.
    \item We can construct $\star$-coupling from a partial-coupling, i.e., $\rho^\uparrow : \<\rho_1,\rho_2\>_\star$ if $\rho : \<\rho_1,\rho_2\>_p$. In fact, $\Pi_{\star}^\bot(\rho^\uparrow) = \rho$.
    \item If $\rho : \<\rho_1,\rho_2\>_\star$, $\tr(A^\star\rho) = \tr(A\Pi_\star^\bot(\rho))$. As a consequence, if $\rho : \<\rho_1,\rho_2\>_p$, then $\tr(A^\star\rho^\uparrow) = \tr(A\rho)$.
  \end{enumerate}
\end{proposition}
\begin{proof}
\noindent(1). Suppose $\rho : \<\rho_1,\rho_2\>_\star$.
Compute
\begin{align*}
  \tr_2(\Pi_{\star}^\bot(\rho)) &= \tr_2((P_\star^\bot\otimes P_\star^\bot)\rho(P_\star^\bot\otimes P_\star^\bot)) \\
  &= P_\star^\bot \tr_2((I_\star\otimes P_\star^\bot)\rho(I_\star \otimes P_\star^\bot)) P_\star^\bot \\
  &\sqsubseteq P_\star^\bot \tr_2(\rho) P_\star^\bot \\
  % \otimes \<\star|P_\star^\bot)\rho (P_\star^\bot\otimes P_\star^\bot|\star\>) + \\
  % &\quad \sum_i (P_\star^\bot \otimes \<i|P_\star^\bot)\rho (P_\star^\bot\otimes P_\star^\bot|i\>) \\
  % &= P_\star^\bot \big(\sum_i (I \otimes \<i|)\rho (I \otimes|i\>)\big)P_\star^\bot \\
  % &\sqsubseteq P_\star^\bot\tr_2(\rho)P_\star^\bot\\
  &= P_\star^\bot \rho_1^\star P_\star^\bot\\
  &= \rho_1
\end{align*}
where we use the fact that $\tr_2(\rho) = \tr_2((I_\star \otimes P_\star)\rho (I_\star \otimes P_\star) + (I_\star \otimes P_\star^\bot)\rho (I_\star \otimes P_\star^\bot))$.
Similarly, $\tr_2(\Pi_{\star}^\bot(\rho))\sqsubseteq \rho_2$.
Furthermore, observe that
\begin{align*}
  \tr(\rho) =\, &\tr((P_\star \otimes P_\star)\rho (P_\star \otimes P_\star))
  + \tr((I_\star \otimes P_\star^\bot)\rho (I_\star \otimes P_\star^\bot))\, + \\
  &\tr((P_\star^\bot \otimes I_\star)\rho (P_\star^\bot \otimes I_\star))
  - \tr((P_\star^\bot \otimes P_\star^\bot)\rho (P_\star^\bot \otimes P_\star^\bot))\\
  =\, &\tr((P_\star \otimes P_\star)\rho (P_\star \otimes P_\star)) + \tr(\rho_2) + \tr(\rho_1) - \tr(\Pi_{\star}^\bot(\rho))
\end{align*}
Note that $\tr(\rho) = 1$, $0\le \tr((P_\star \otimes P_\star)\rho (P_\star \otimes P_\star))$, we have $ \tr(\rho_1) + \tr(\rho_2) \le 1 + \tr(\Pi_{\star}^\bot(\rho))$.

All above implies that $\Pi_{\star}^\bot(\rho) : \<\rho_1,\rho_2\>_p$.

\noindent(2). Suppose $\rho : \<\rho_1,\rho_2\>_p$. It is straightforward that :
\begin{align*}
  \tr_2(\rho^\uparrow) =\, &(1+\tr(\rho) - \tr(\rho_1) - \tr(\rho_2))|\star\>\<\star|\, + \\
  &(\rho_1 - \tr_2(\rho)) + \tr(\rho_2 - \tr_1(\rho)) |\star\>\<\star| + \tr_2(\rho) \\
  =\ &(1-\tr(\rho_1))|\star\>\<\star| + \rho_1 \\
  =\ &\rho_1^\star.
\end{align*}
Similarly, $\tr_1(\rho^\uparrow) = \rho_2^\star$. Thus, $\rho^\uparrow : \<\rho_1^\star, \rho_2^\star\>$, or equivalently, $\rho^\uparrow : \<\rho_1, \rho_2\>_\star$.
$\Pi_{\star}^\bot(\rho^\uparrow) = \rho$ is trivial by computation.

\noindent(3). Note that, $A^\star$ is preserved under the projection $P_\star^\bot\otimes P_\star^\bot$, so :
\begin{align*}
  \tr(A^\star \rho) &= 
    \tr(A^\star(P_\star^\bot\otimes P_\star^\bot)\rho(P_\star^\bot\otimes P_\star^\bot))) \\
    &= \tr(A \Pi_\star^\bot(\rho)).
\end{align*}
\end{proof}

% \begin{theorem}[Quantum Strassen's Theorem for partial coupling]
%   \label{thm: strassen partial coupling}
%   For any $\rho_1\in\DD(\HH_1)$, $\rho_2\in\DD(\HH_2)$, $A\in\Pos(\HH_1\otimes\HH_2)$ and $\epsilon \in \mathbb{R}^+$, the following are equivalent:
%   \begin{enumerate}
%     \item There exists partial-coupling $\rho : \<\rho_1,\rho_2\>_p$ such that $\tr(A\rho)\le \epsilon$;
%     \item For any $y_1,y_2\in \RR^+$, $Y_1\in \Pos(\HH_1), Y_2\in\Pos(\HH_2)$, such that $y_1\le y_2$, $Y_1\le y_2I_1$, $y_1I_2\le Y_2$ and $A\sqsupseteq Y_1\otimes I_2 - I_1\otimes Y_2$, it holds that:
%     $$y_1(1-\tr(\rho_1)) + \tr(Y_1\rho_1) \le y_2(1-\tr(\rho_2)) + \tr(Y_2\rho_2) + \epsilon.$$
%   \end{enumerate}
% \end{theorem}

\begin{proposition}[(Sub-)convex Combination of Partial Coupling]
  \label{lem: scale of partial coupling}
  Let $\{\lambda_i\}_{i\in I}$ be a subdistribution over index set $I$ (i.e., $0 \le \lambda_i\le 1$ for all $i$, and $\sum_i \lambda_i \le 1$),
  and $\rho_i\in \DD(\HH_1)$, $\sigma_i\in \DD(\HH_2)$, and partial couplings $\delta_i : \<\rho_i,\sigma_i\>_p$ with indices from $I$. Then 
  $$\sum_i\lambda_i\delta_i : \Big\<\sum_i\lambda_i\rho_i, \sum_i\lambda_i\sigma_i\Big\>_p.$$
  As an corollary, for any $\rho\in\DD(\HH_1)$, $\sigma\in\DD(\HH_2)$ and $0\le c \le 1$, if $\delta : \<\rho,\sigma\>_p$, then $c\delta : \<c\rho,c\sigma\>_p$.
\end{proposition}
\begin{proof}
%[Proof of \Cref{lem: scale of partial coupling}]
  First observe: 
  \begin{align*}
    &\tr_2\Big(\sum_i\lambda_i\delta_i\Big) = \sum_i\lambda_i\tr_2(\delta_i) 
    \sqsubseteq \sum_i\lambda_i\rho_i; \\
    &\tr_1\Big(\sum_i\lambda_i\delta_i\Big) = \sum_i\lambda_i\tr_1(\delta_i) 
    \sqsubseteq \sum_i\lambda_i\sigma_i.
  \end{align*}
  Further notice that,
  \begin{align*}
    \tr\Big(\sum_i\lambda_i\rho_i\Big) + \tr\Big(\sum_i\lambda_i\sigma_i\Big)
    &= \sum_i\lambda_i (\tr(\rho_i) + \tr(\sigma_i)) \\
    &\le \sum_i\lambda_i (1 + \tr(\delta_i))
    = \sum_i\lambda_i + \tr\Big(\sum_i\lambda_i\delta_i\Big) \\
    &\le 1 + \tr\Big(\sum_i\lambda_i\delta_i\Big)
  \end{align*}
  as $\{\lambda_i\}_{i\in I}$ is a subdistribution. This completes the proof.
  % Since $\rho : \<\rho_1,\rho_2\>_p$, so $\tr_2(\rho)\sqsubseteq \rho_1$, $\tr_1(\rho)\sqsubseteq \rho_2$ and $\tr(\rho_1) + \tr(\rho_2)\le 1 + \tr(\rho)$.

  % It is obvious that $\tr_2(c\rho)\sqsubseteq c\rho_1$, $\tr_1(c\rho)\sqsubseteq c\rho_2$, and
  % \begin{align*}
  %   1 + \tr(c\rho) &\ge 1-c + c(1 + \tr(\rho)) \\ 
  %   &\ge 1-c + c(\tr(\rho_1) + \tr(\rho_2)) \\
  %   &\ge \tr(c\rho_1) + \tr(c\rho_2).
  % \end{align*}
  % This completes the proof.
\end{proof}

\begin{theorem}[Quantum Strassen's Theorem for Partial Coupling]
\label{thm: strassen partial coupling}
  For any $\rho_1\in\DD(\HH_1)$, $\rho_2\in\DD(\HH_2)$, $A\in\Pos(\HH_1\otimes\HH_2)$ and $\epsilon \in \mathbb{R}^{+\infty}$, the following are equivalent:
  \begin{enumerate}
    \item There exists partial-coupling $\rho : \<\rho_1,\rho_2\>_p$ such that $\tr(A\rho)\le \epsilon$;
    \item For any $y_1,y_2\in \RR^+$, $Y_1\in \Pos(\HH_1), Y_2\in\Pos(\HH_2)$, such that $y_1\le y_2$, $Y_1\le y_2I_1$, $y_1I_2\le Y_2$ and $A\sqsupseteq Y_1\otimes I_2 - I_1\otimes Y_2$, it holds that:
    $$y_1(1-\tr(\rho_1)) + \tr(Y_1\rho_1) \le y_2(1-\tr(\rho_2)) + \tr(Y_2\rho_2) + \epsilon.$$
  \end{enumerate}
\end{theorem}
\begin{proof}
  %[Proof of Theorem \ref{thm: strassen partial coupling}]
  If $\epsilon = +\infty$, then both (1) and (2) trivially hold. So we consider the case that $\epsilon\in \RR^+$, i.e., $\epsilon$ is finite.
  We first introduce the following condition:
  \begin{itemize}
    \item[3)] There exists star-coupling $\rho : \<\rho_1,\rho_2\>_\star$, i.e., $\rho:\<\rho_1^\star,\rho_2^\star\>$, such that $\tr(A^\star\rho)\le \epsilon$.
  \end{itemize}

  \noindent$\bullet$ (1 $\Rightarrow$ 3). Let $\rho$ be the witness of (1), then $\rho^\uparrow : \<\rho_1,\rho_2\>_\star$ by Lemma \ref{lemma: relation star coupling}(2). According to Lemma \ref{lemma: relation star coupling}(3), $\tr(A^\star\rho^\uparrow) = \tr(A\rho) \le \epsilon$. 

  \noindent$\bullet$ (3 $\Rightarrow$ 1). Let $\rho$ be the witness of (3), then $\Pi_\star^\bot(\rho) : \<\rho_1,\rho_2\>_p$ by Lemma \ref{lemma: relation star coupling}(1). According to Lemma \ref{lemma: relation star coupling}(3), 
  $\tr(A\Pi_\star^\bot(\rho)) = \tr(A^\star\rho) \le \epsilon$. 

  Thus, (1) is equivalent to (3). (3) says that, $\rho_1^\star A^{\star\#}_\epsilon \rho_2^\star$, then by Theorem \ref{thm:quantum strassen defect alter}, it is then equivalent to:
  \begin{itemize}
    \item[4)] For any $Z_1 \in \Pos(\HH_1^\star)$ and $Z_2 \in \Pos(\HH_2^\star)$ such that $A^\star \geqlow Z_1 \otimes I_{2^\star} - I_{1^\star} \otimes Z_2$, it holds that 
    \[
      \tr(Z_1\rho_1^\star) \leq \tr(Z_2\rho_2^\star) + \epsilon.
    \]
  \end{itemize}
  What remaining to be shown is (2) equivalent to (4).

  \noindent$\bullet$ (4 $\Rightarrow$ 2). Set $Z_1 = y_1|\star\>\<\star| + Y_1$ and $Z_2 = y_2|\star\>\<\star| + Y_2$. Obviously, $Z_1\in\Pos(\HH_1^\star)$ and $Z_2\in\Pos(\HH_2^\star)$. Observe that: 
  \begin{align*}
    &A^\star - (Z_1 \otimes I_{2^\star} - I_{1^\star} \otimes Z_2) \\
    =\ & A - (y_1|\star\>\<\star| + Y_1) \otimes (|\star\>\<\star| + I_2) \\ 
    & + (|\star\>\<\star| + I_1) \otimes (y_2|\star\>\<\star| + Y_2) \\ 
    =\ &(A - (Y_1\otimes I_2 - I_1\otimes Y_2)) + (y_2-y_1)|\star\star\>\<\star\star| \\
    &+ (y_2I_1 - Y_1)\otimes|\star\>\<\star| + |\star\>\<\star|\otimes(Y_2 - y_1I_2) \\
    \sqsupseteq\  &0.
  \end{align*}
  So, $\tr(Z_1\rho_1^\star) \leq \tr(Z_2\rho_2^\star) + \epsilon$, or equivalently,
  $$ y_1(1-\tr(\rho_1)) + \tr(Y_1\rho_1) \le y_2(1-\tr(\rho_2)) + \tr(Y_2\rho_2) + \epsilon.$$

  \noindent$\bullet$ (2 $\Rightarrow$ 4). For any $Z_1 \in \Pos(\HH_1^\star)$ and $Z_2 \in \Pos(\HH_2^\star)$ such that $A^\star \geqlow Z_1 \otimes I_{2^\star} - I_{1^\star} \otimes Z_2$, by projecting it to $P_\star\otimes P_\star$, $P_\star^\bot\otimes P_\star$, $P_\star\otimes P_\star^\bot$ and $P_\star^\bot\otimes P_\star^\bot$, the L\"{o}wner preserves, and thus:
  \begin{align*}
    y_1 - y_2\le 0 &\qquad y_1I_2 - Y_2\sqsubseteq 0\qquad 
    Y_1 - y_2I_1\sqsubseteq 0\\
    &Y_1\otimes I_2 - I_1\otimes Y_2\sqsubseteq A
  \end{align*}
  where $Z_1 = \left(\begin{smallmatrix} y_1 & \cdot \\ \cdot & Y_1 \end{smallmatrix}\right)$ and $Z_2 = \left(\begin{smallmatrix} y_2 & \cdot \\ \cdot & Y_2 \end{smallmatrix}\right)$. Furthermore, observe that
  \begin{align*}
    &\tr(Z_1\rho_1^\star) - \tr(Z_2\rho_2^\star) \\
    =\ &y_1(1-\tr(\rho_1)) + \tr(Y_1\rho_1) - (y_2(1-\tr(\rho_2)) + \tr(Y_2\rho_2)) \\
    \le\ &\epsilon
  \end{align*}
  by employing (2), and this completes the proof.
\end{proof}

\section{Deferred Proofs in ``Quantum Optimal Transport'' Section}
\label{app: QOT}

\begin{proposition}\label{prop:partial-coupling}
    Given $\rho_1 \in \mathcal{D}(\mathcal{H}_1)$
    and $\rho_2 \in \mathcal{D}(\mathcal{H}_2)$
    where $\mathcal{H}_1$ and $\mathcal{H}_2$
    are finite-dimensional Hilbert spaces,
    the set of partial couplings
    $S = \{\rho \mid \rho:\langle \rho_1, \rho_2\rangle_p\}$
    is a non-empty, closed and convex set.    
    %\lz{Minbo: prove nonempty}
\end{proposition}
\begin{proof}
    For non-emptiness, 
    given $\rho_1 \in \mathcal{D}(\mathcal{H})_1$
    and $\rho_2 \in \mathcal{D(\mathcal{H}_2)}$.
    We claim $\rho = \rho_1\otimes \rho_2 \in S$.
    It is direct to see that 
    $\tr_2(\rho)\sqsubseteq \rho_1$
    and 
    $\tr_1(\rho)\sqsubseteq \rho_2$.
    For the trace constraint,
    notice that
    $\tr(\rho) = \tr(\rho_1) \tr(\rho_2)$,
    $\tr(\rho_1) \le 1$, and $\tr(\rho_2)\le 1$,
    we have
    $(1-\tr(\rho_1))(1-\tr(\rho_2))\ge$,
    meaning that
    $\tr(\rho_1) + \tr(\rho_2) \le 1 + \tr(\rho)$
    as we want.

    For closeness, 
    given $\rho^i \to \rho$ with $\rho^i \in S$,
    we show that $\rho \in S$.
    By definition, 
    we know
    $\tr_2(\rho^i) \sqsubseteq \rho_1$,
    or equivalently, 
    for any $\sigma \in \mathcal{D}(\HH_1)$,
    $\tr(\rho_i \sigma\otimes I) \le \tr(\rho_1 \sigma)$.
    Fixed $\sigma$, 
    the function $\tr((\sigma \otimes I) \cdot)$
    is linear and continuous.
    Therefore, 
    by $\rho^i\to \rho$ we know
    $\tr(\tr_2(\rho)\sigma) \le \tr(\rho_1\sigma)$.
    Since the above inequality holds for any $\sigma \in \mathcal{D}(\mathcal{H})$,
    we conclude that
    $\tr_2(\rho)\sqsubseteq \rho_1$.
    Similarly we can prove 
    $\tr_1(\rho)\sqsubseteq \rho_2$.
    By the continuity of the trace function,
    we can also conclude
    $\tr(\rho_1)+\tr(\rho_2)\le 1+ \tr(\rho)$.
    Therefore we have $\rho\in S$.

    For convexity,
    suppose $\sigma_1, \sigma_2\in S$
    and $\lambda\in (0,1)$.
    Consider $\rho = \lambda \sigma_1 + (1-\lambda)\sigma_2$.
    From $\tr_2(\sigma_1)\sqsubseteq \rho_1$
    and $\tr_2(\sigma_2)\sqsubseteq \rho_1$,
    we know
    $\lambda \tr_2(\sigma_1) + (1-\lambda)
    \tr_2(\sigma_2) \sqsubseteq \lambda \rho_1 
    + (1-\lambda)\rho_1$.
    Simplifying above, 
    we get
    $\tr_2(\rho)\sqsubseteq \rho_1$.
    Similarly we have
    $\tr_1(\rho) \sqsubseteq \rho_2$.
    From 
    $\tr(\rho_1) + \tr(\rho_2) 
    \le 1+ \tr(\sigma_1)$
    and $\tr(\rho_1) + \tr(\rho_2) 
    \le 1+ \tr(\sigma_2)$,
    we get
    $\tr(\rho_1) + \tr(\rho_2) 
    \le 1+ \lambda\tr(\sigma_1)
    +(1-\lambda)\tr(\sigma_2)$,
    meaning
    $\tr(\rho_1) + \tr(\rho_2) 
    \le 1+ \tr(\rho)$
    and $\rho \in S$ as we desired.
\end{proof}

\begin{proposition}[Jointly Convexity of QOT]
  \label{lem:QOT convexity}
  Let $\{\lambda_i\}_{i\in I}$ be a subdistribution over index set $I$, and $\rho_i\in \DD(\HH_1)$, $\sigma_i\in \DD(\HH_2)$ with indices from $I$. For any cost function $C$, It holds that:
  $$ T_C\Big(\sum_i\lambda_i\rho_i, \sum_i\lambda_i\sigma_i\Big)\le \sum_i\lambda_iT_C(\rho_i, \sigma_i).$$

  As a corollary, if $\{\lambda_i\}_{i\in I}$ is a distribution and $\sigma\in \DD(\HH_2)$, then  $ T_C\big(\sum_i\lambda_i\rho_i, \sigma\big)\le \sum_i\lambda_iT_C(\rho_i, \sigma).$
\end{proposition}
\begin{proof}
%[Proof of \Cref{lem:QOT convexity}]
  Select $\delta_i : \<\rho_i, \sigma_i\>_p$ such that $T_C(\rho_i, \sigma_i) = \tr(C\delta_i)$. By \Cref{lem: scale of partial coupling}, $\sum_i\lambda_i\delta_i : \big\<\sum_i\lambda_i\rho_i, \sum_i\lambda_i\sigma_i\big\>_p$, so
  \begin{align*}
    T_C\Big(\sum_i\lambda_i\rho_i, \sum_i\lambda_i\sigma_i\Big) \le \tr\Big(C\sum_i\lambda_i\delta_i\Big) = \sum_i\lambda_i\tr(C\delta_i) = \sum_i\lambda_iT_C(\rho_i, \sigma).
  \end{align*}
\end{proof}

The following lemma demonstrates that for monotonicity it suffices to
consider only density operators, which appears useful in the following proofs.
\begin{lemma}[Alternative Characterization]
\label{lem:validity quantum operation alter}
    Given two quantum operations $\EE_1\in\QO(\HH_1),\EE_2\in\QO(\HH_2)$, input and output costs , the following statement are equivalent:
    \begin{enumerate}
        \item $(\EE_1,\EE_2)$ is monotone w.r.t. $C_i$ and $C_o$;
        \item For all $\rho_1\in\DD^1(\HH_1)$ and $\rho_2\in\DD^1(\HH_2)$, 
        $$T_{C_o}(\EE_1(\rho_1),\EE_2(\rho_2)) \le T_{C_i}(\rho_1,\rho_2).$$
        % \item If $\EE_1,\EE_2$ are quantum channels. For all $\rho\in\DD(\HH_1\otimes\HH_2)$, 
        % $$T_{C_o}(\EE_1(\tr_2(\rho)),\EE_2(\tr_1(\rho))) \le T_{C_i}(\tr_2(\rho),\tr_1(\rho)).$$
    \end{enumerate}
\end{lemma}
\begin{proof}
(1 $\Rightarrow$ 2) is trivial. For (2 $\Rightarrow$ 1), by definition, it is sufficient to show that, for all $\rho_1\in\DD(\HH_1)$, $\rho_2\in\DD(\HH_2)$ and $\rho : \<\rho_1,\rho_2\>_p$, there exists $\sigma : \<\EE_1(\rho_1),\EE_2(\rho_2)\>_p$ such that :
$$\tr(C_o\sigma)\le\tr(C_i\sigma).$$

Since $\rho$ is a partial coupling, then $\rho_1'\triangleq\tr_2(\rho)\sqsubseteq \rho_1$, $\rho_2'\triangleq\tr_1(\rho)\sqsubseteq \rho_2$, $1 + \tr(\rho)\ge\tr(\rho_1) + \tr(\rho_2)$. Set $\tr(\rho) = c$. If $c = 0$, then $\tr(\EE_1(\rho_1)) + \tr(\EE_2(\rho_2))\le \tr(\rho_1) + \tr(\rho_2)\le 1$, thus $0 : \<\EE_1(\rho_1), \EE_2(\rho_2)\>_p$, and obviously $\tr(C_i\rho) = \tr(C_o0) = 0$. If $c > 0$, by taking $\rho_1'/c$ and $\rho_2'/c$ in (2), there must exist a partial coupling 
$$\sigma : \< \EE_1(\rho_1'/c), \EE_1(\rho_2'/c)\>_p$$
such that $\tr(C_i(\rho/c))\ge\tr(C_o\sigma)$. By \Cref{lem: scale of partial coupling}, $c\sigma : \<\rho_1',\rho_2'\>_p$, and $\tr(C_i\rho)\ge \tr(C_o(c\sigma))$, so it is sufficient to show $c\sigma : \<\EE_1(\rho_1), \EE_2(\rho_2)\>_p$.
First observe that, 
$$\tr_2(c\sigma) = c\tr(\sigma)\sqsubseteq c \EE_1(\rho_1'/c) =\EE_1(\rho_1')\sqsubseteq \EE_1(\rho_1) $$
and similarly, $\tr_1(c\sigma)\sqsubseteq \EE_2(\rho_2)$. On the other hand,
\begin{align*}
    1+\tr(c\sigma) &= 1-c+c(1+\tr(\sigma)) \\
    &\ge 1-c+c(\tr(\EE_1(\rho_1'/c)) + \tr(\EE_2(\rho_2'/c))) \\
    &= 1 - c + (\tr(\EE_1(\rho_1')) + \tr(\EE_2(\rho_2'))).
\end{align*}
Notice that $\tr(\rho_1-\rho_1')\ge\tr(\EE_1(\rho_1-\rho_1'))$ since $\rho_1'\sqsubseteq\rho_1$ and $\EE_1$ is a quantum operation, and similarly holds for $\rho_2, \rho_2'$, we get:
\begin{align*}
    &\tr(\EE_1(\rho_1')) + \tr(\EE_2(\rho_2')) \\
    \ge\ &\tr(\EE_1(\rho_1)) + \tr(\EE_2(\rho_2)) - (\tr(\rho_1) + \tr(\rho_2)) \\
    &+ \tr(\rho_1') + \tr(\rho_2') \\
    \ge\ &\tr(\EE_1(\rho_1)) + \tr(\EE_2(\rho_2)) - (1+\tr(\rho)) + 2\tr(\rho) \\
    =\ &\tr(\EE_1(\rho_1)) + \tr(\EE_2(\rho_2)) - 1 + c
\end{align*}
Combine these two inequalities, we obtain:
$$ 1+\tr(c\sigma) \ge \tr(\EE_1(\rho_1)) + \tr(\EE_2(\rho_2)),$$
which completes the proof.
\end{proof}

\begin{lemma}[Monotonicity for Quantum Channels]
\label{lem:validity quantum channel}
  Suppose $\EE_1\in\QC(\HH_1),\EE_2\in\QC(\HH_2)$ are two quantum channels.
  Then $(\EE_1, \EE_2)$ is monotone w.r.t. $C_i$ and $C_o$ if and only if for every $\rho \in \DD(\HH_1 \otimes \HH_2)$, there exists a coupling $\sigma : \langle \EE_1(\tr_2(\rho)), \EE_2(\tr_1(\rho)) \rangle$ such that
  \[
  \tr(C_i \rho) \geq \tr(C_o \sigma).
  \]
\end{lemma}
\begin{proof}
  (if) part. For all $\rho_1\in\DD^1(\HH_1)$ and $\rho_2\in\DD^1(\HH_2)$, set $\rho : \<\rho_1,\rho_2\>$ which obtains $\tr(C_i\rho) = T_{C_i}(\rho_1,\rho_2)$. Then by assumption, there exists a coupling $\sigma : \langle \EE_1(\tr_2(\rho)), \EE_2(\tr_1(\rho)) \rangle$ (i.e., $\sigma : \langle \EE_1(\rho_1), \EE_2(\rho_2) \rangle$) such that:
  $$T_{C_o}(\EE_1(\rho_1),\EE_2(\rho_2)) \le \tr(C_o\sigma) \le \tr(C_i\rho) = T_{C_i}(\rho_1,\rho_2).$$
  Then by \Cref{lem:validity quantum operation alter} we finish this part.

  (only if) part. Since coupling is preserved under scaling, and by \Cref{lem: IVP algebraic}, we only need to focus on $\rho\in\DD^1(\HH_1\otimes\HH_2)$.
  Choose $\sigma : \langle \EE_1(\tr_2(\rho)), \EE_2(\tr_1(\rho)) \rangle$ which $\tr(C_o\sigma) = T_{C_o}(\EE_1(\tr_2(\rho)),\EE_2(\tr_1(\rho)))$.
  By assumption, we have:
  $$\tr(C_o\sigma) = T_{C_o}(\EE_1(\tr_2(\rho)),\EE_2(\tr_1(\rho)))\le T_{C_i}(\tr_2(\rho), \tr_1(\rho))\le\tr(C_i\rho).$$
\end{proof}

\begin{proposition}[\Cref{prop:monotone basic}]
%\label{prop:monotone basic-appendix}
The monotonicity satisfies several desired properties for data processing:
\begin{enumerate}
  \item 
  %\label{lem:judgment two side}
  \emph{Backward}.
  $(\EE_1,\EE_2)$ is monotone w.r.t. $(\EE_1^\dag\otimes \EE_2^\dag)(C)$ and $C$. Here, $\EE^\dag$ is the dual of $\EE$, which satisfies $\tr(A\EE(B)) = \tr(\EE^\dag(A)B)$ for all linear operator $A,B$.
  \item
  %\label{lem:judgment csq}
  \emph{Consequence}.
  Suppose $(\EE_1,\EE_2)$ is monotone w.r.t. $C_i'$ and $C_o'$, and $C_i'\sqsubseteq C_i$, $C_o\sqsubseteq C_o'$, then $(\EE_1,\EE_2)$ is monotone w.r.t. $C_i$ and $C_o$.
  \item 
  %\label{lem:judgment seq}
  \emph{Sequential composition}.
  Suppose $(\EE_1,\EE_1')$ is monotone w.r.t. $C_i$ and $C_m$, and $(\EE_2,\EE_2')$ is monotone w.r.t. $C_m$ and $C_o$, then $(\EE_2\circ \EE_1,\EE_2'\circ \EE_1')$ is monotone w.r.t. $C_i$ and $C_o$.
  \end{enumerate}
  For any super-operator $\EE$, its dual $\EE^\dag$ is another super-operator. Whenever $\EE$ is a quantum operation with Kraus operator $\{E_i\}$, then $\EE^\dag$ has Kraus representation $\{E_i^\dag\}$.
  $\circ$ is the composition of two quantum operations, i.e., for all $\rho$, $(\EE_1\circ \EE_2)(\rho) \triangleq \EE_1(\EE_2(\rho))$.
\end{proposition}
\begin{proof}
%[Proof of \Cref{prop:monotone basic}]
  (1) 
  By \Cref{lem:validity quantum operation alter}, for any $\rho\in\DD^1(\HH_1\otimes\HH_2)$, set $\sigma = (\EE_1\otimes\EE_2)(\rho)$. By the property of dual map, we have: 
  $$\tr((\EE_1^\dag\otimes \EE_2^\dag)(C)\rho) = \tr(C(\EE_1\otimes \EE_2)(\rho)) = \tr(C\sigma).$$
  It is then sufficient to show that $\sigma : \< \EE_1(\tr_2(\rho)),\EE_2(\tr_1(\rho)) \>_p$, which is completed by noticing that:
  \begin{align*}
    \tr_2(\sigma) &= \tr_2((\EE_1\otimes\EE_2)(\rho)) \sqsubseteq \EE_1(\tr_2(\rho)), \\
    \tr_1(\sigma) &= \tr_1((\EE_1\otimes\EE_2)(\rho)) \sqsubseteq \EE_2(\tr_1(\rho)), \\
    1 + \tr(\sigma) &= \tr(\rho) + \tr((I\otimes I)(\EE_1\otimes\EE_2)(\rho)) \\
    &= \tr((I\otimes I + \EE_1^\dag(I) \otimes \EE_2^\dag(I))\rho) \\
    &\ge \tr((\EE_1^\dag(I)\otimes I + I \otimes \EE_2^\dag(I))\rho) \\
    &= \tr(\EE_1^\dag(I)\tr_2(\rho)) + \tr(\EE_2^\dag(I)\tr_1(\rho)) \\
    &= \tr(\EE_1(\tr_2(\rho))) + \tr(\EE_2(\tr_1(\rho))).
  \end{align*}
  First two hold since $\EE_1$ and $\EE_2$ are quantum operations, and furthermore, $\EE_1^\dag(I)\sqsubseteq I$ and $\EE_2^\dag(I)\sqsubseteq I$, thus, $0\sqsubseteq (I - \EE_1^\dag(I))\otimes (I - \EE_2^\dag(I))$ which then leads to the inequality of fifth line.
  
  (2)
  By \Cref{lem:validity quantum operation alter}, for any $\rho\in\DD^1(\HH_1\otimes\HH_2)$, by assumption, there exists $\sigma : \< \EE_1(\tr_2(\rho)), \EE_2(\tr_1(\rho))\>_p$ such that $\tr(C_i'\rho)\ge\tr(C_o'\sigma)$. Thus,
  $$ \tr(C_i\rho) \ge \tr(C_i'\rho)\ge \tr(C_o'\sigma) \ge \tr(C_o\sigma).$$

  (3)
  For any $\rho_1$ and $\rho_1'$, and any partial coupling $\rho : \<\rho_1,\rho_1'\>_p$, by the first assumption, there exists a partial coupling $\sigma : \<\EE_1(\rho_1),\EE_1'(\rho_1')\>_p$ such that $\tr(C_i\rho) \geq \tr(C_m\sigma)$. 
  From the second assumption there is a partial coupling $\sigma': \< \EE_2(\EE_1(\rho_1)), \EE_2'(\EE_1'(\rho_1')) \>_p$, or equivalently, $\sigma': \< (\EE_1\circ\EE_2)(\rho_1), (\EE_1'\circ\EE_2')(\rho_1') \>_p$ such that $\tr(C_m\sigma) \geq \tr(C_o\sigma')$, which concludes the proof by noticing $\tr(C_i\rho) \geq \tr(C_o\sigma')$.
\end{proof}

\begin{proposition}[Split Cost]
\label{lem:judgment two side wlp}
%\label{prop:monotone basic-appendix}
The monotonicity can further be checked via splitting output cost. Formally,
  Suppose $\EE_1, \EE_2$ are quantum channels. 
  $(\EE_1,\EE_2)$ is monotone w.r.t. $C$ and $Q_1\otimes I + I\otimes Q_2$ if and only if 
  $$C\sqsupseteq \EE_1^\dag(Q_1)\otimes I + I \otimes\EE_2^\dag(Q_2).$$
% \begin{lemma}[Two-side weakest precondition]
%   \label{lem:judgment two side wlp}
%   Suppose $\EE_1, \EE_2$ are quantum channels. Then 
%   $\vDash \EE_1 \sim \EE_2 : P \Rightarrow Q_1\otimes I + I\otimes Q_2$ if and only if 
%   $$P\sqsupseteq \EE_1^\dag(Q_1)\otimes I + I \otimes\EE_2^\dag(Q_2).$$
% \end{lemma}
\end{proposition}
\begin{proof}
  The if part can be proved by combining \Cref{prop:monotone basic} (\ref{lem:judgment two side}) and (\ref{lem:judgment csq}). For the only if part, by employing \Cref{lem:validity quantum channel}, we have for all $\rho\in\DD(\HH_1\otimes\HH_2)$, there exists $\sigma : \< \EE_1(\tr_2(\rho)), \EE_2(\tr_1(\rho)) \>$ such that 
  \begin{align*}
    \tr(C\rho) &\ge \tr((Q_1\otimes I + I\otimes Q_2)\sigma) \\
    &= \tr(Q_1\tr_2(\sigma)) + \tr(Q_2\tr_1(\sigma)) \\
    &= \tr(Q_1\EE_1(\tr_2(\rho))) + \tr(Q_2\EE_2(\tr_1(\rho)) \\
    &= \tr((\EE_1^\dag(Q_1)\otimes I + I\otimes\EE_2^\dag(Q_2))\rho).
  \end{align*}
  Since this holds for all $\rho$, we must have: $C\sqsupseteq \EE_1^\dag(Q_1)\otimes I + I \otimes\EE_2^\dag(Q_2)$.
\end{proof}

\begin{theorem}[\Cref{lem:judgment strassen}]
%\label{thm:judgment strassen-appendix}
  Suppose $\EE_1$ and $\EE_2$ are quantum channels and costs $C_i\in\PosI$ and $C_o\in \Pos$ (i.e., $C_o$ is finite). Then the following statements are equivalent:
  \begin{enumerate}
    \item $(\EE_1,\EE_2)$ is monotone w.r.t. $C_i$ and $C_o$
    \item for all $(Y_1,Y_2,n)\in\mathcal{Y}$, $(\EE_1,\EE_2)$ is monotone w.r.t. $C_i + nI$ and $Y_1\otimes I + I\otimes (nI - Y_2)$,
    where $\mathcal{Y} \eqdef \{(Y_1, Y_2, n) \ | \ n \in \NN; 0\leqlow Y_1; 0\leqlow Y_2 \leqlow nI; C_o \geqlow Y_1 \otimes I - I \otimes Y_2\}$.
  \end{enumerate}
\end{theorem}

\begin{proof}
%[Proof of \Cref{lem:judgment strassen}]
  As $\EE_1$ and $\EE_2$ are quantum channels, we employ \Cref{lem:validity quantum channel} to interpret monotonicity.

  (1) says that for all $\rho \in \DD({\HH_1 \otimes \HH_2})$ there is a $\sigma: \langle \EE_1(\tr_2(\rho)), \EE_2(\tr_1)(\rho) \rangle$ such that $\tr(C_o\sigma) \leq \tr(C_i \rho)$, or equivalently, 
  \[\EE_1(\tr_2(\rho)) {C_o}^\#_\epsilon \EE_2(\tr_1)(\rho)\]
  where $\epsilon = \tr(C_i \rho)$.

  (2) says that for all $\rho \in \DD({\HH_1 \otimes \HH_2})$, positive Hermitians $Y_1, Y_2$ and $n$ such that $Y_2 \leqlow nI$ (where such an $n$ always exists for any $Y_2$) and $Y_1\otimes I - I\otimes Y_2\sqsubseteq C_o$, there exists a coupling $\sigma: \langle \EE_1(\tr_2(\rho)), \EE_2(\tr_1)(\rho) \rangle$ such that 
  \[
    \tr(C_i \rho) + n\tr(\rho) \geq \tr((Y_1 \otimes I)\sigma) + n\tr(\sigma) - \tr((I \otimes Y_2)\sigma)
  \]
  which, by noticing $\tr(\sigma) = \tr(\rho)$ since $\EE_1$ and $\EE_2$ are quantum channels, is equivalent to 
  \[
  \tr(Y_1(\EE_1\tr_2(\rho))) \leq \tr(Y_2(\EE_2 \tr_1(\rho))) + \tr(C_i \rho).
  \]
  The equivalence then follows from \Cref{thm:quantum strassen defect alter}.
\end{proof}

Symmetric space is useful for describing equivalence of states, as studied in \cite{zhou2018quantumcouplingstrassentheorem}:
%The standard approach is given by the following lemma:
\begin{lemma}[State Equivalence (e.g. \cite{barthe_rqpd} prop 3.2)]
  \label{lemma:stateequiv}
  Let $\symeq$ be the projector $\frac{1}{2}(I + \swap)$. Then, for any $\rho_1, \rho_2 \in \DD(\HH)$,
  \[
  \rho_1 = \rho_2 \iff \rho_1 (\symeq)^\# \rho_2.
  \]
\end{lemma}

\begin{proposition}[\Cref{lem: channel equivalence}]
%\label{prop: channel equivalence-appendix}
    Two quantum channels $\EE_1$ and $\EE_2$ are equivalent if and only if for all density operators $\rho_1,\rho_2$, $T_s(\EE_1(\rho_1),\EE_2(\rho_2))\le T_s(\rho_1,\rho_2).$
\end{proposition}

\begin{proof}
%[Proof of \Cref{lem: channel equivalence}]
  The ``only if'' part holds directly by the monotonicity of stable QOT under data processing~\cite{mullerhermes2022monotonicity}.
  
  For the ``if'' part, 
  for any $\rho\in \DD^1(\HH)$, 
  from the assumption, we have 
  $0\le T_s(\rho,\rho)\le T(\rho,\rho) = 0$.
  Thus, $T_s(\EE_1(\rho),\EE_2(\rho)) = 0$, 
  or equivalently, $T(\EE_1(\rho)\otimes \halfI,\EE_2(\rho)\otimes\halfI) = 0$, indicating $\EE_1(\rho)\otimes \halfI = \EE_2(\rho)\otimes \halfI$. 
  Then,
  we know 
  $\EE_1(\rho) = \EE_2(\rho)$ for all $\rho$, 
  meaning that $\EE_1=\EE_2$.
\end{proof}

\begin{proposition}[\Cref{thm:strassen QOT}]
%\label{prop:strassen QOT-appendix}
  Given $\rho_1,\rho_2\in\DD^1(\HH)$ and $\epsilon\in \mathbb{R}^+$, the following are equivalent:
  \begin{enumerate}
    \item $T_s(\rho_1,\rho_2)\le \epsilon$;
    \item For all $Y_1, Y_2\in \Pos(\HH\otimes\HH_2)$ such that $P_{sym}^\bot[\HH\otimes\HH_2] \ge 2(Y_1\otimes I - I\otimes Y_2)$, it holds that:
    $$\tr(\tr_2(Y_1)\rho_1) \le \tr(\tr_2(Y_2))\rho_2) + \epsilon.$$
  \end{enumerate}
\end{proposition}

\begin{proof}
%[Proof of \Cref{thm:strassen QOT}]
    By the property of $T_s$, we first observe:
    $$T_s(\rho_1,\rho_2) = \inf_{\tau:\<\rho_1\otimes \halfI, \rho_2\otimes\halfI\>}(P_{sym}^\bot[\HH\otimes\HH_2]\tau),$$
    which implies that, (1) is equivalent to $(\rho_1\otimes \halfI)\big(P_{sym}^\bot[\HH\otimes\HH_2]\big)^\#_\epsilon(\rho_2\otimes \halfI)$. Employing Theorem \ref{thm:quantum strassen defect alter}, (1) is further equivalent to:
    \begin{enumerate}
        \item[3)] for all $Y_1, Y_2\in \Pos(\HH\otimes\HH_2)$ such that $P_{sym}^\bot[\HH\otimes\HH_2] \ge Y_1\otimes I - I\otimes Y_2$, it holds that:
            $$\tr\Big(Y_1\big(\rho_1\otimes \halfI\big)\Big) \le \tr\Big(Y_2\big(\rho_2\otimes \halfI\big)\Big) + \epsilon.$$
    \end{enumerate}
    Notice that, $\tr(Y_i(\rho_i\otimes \halfI)) = \tr(\tr_2(\frac{Y_i}{2})\rho_i)$ for $i=1,2$, direct substitutions of $Y_i' = \frac{Y_i}{2}$ translate (3) to (2).
\end{proof}

\section{Deferred Proofs in ``A Quantum Relational Hoare Logic'' Section}

\begin{lemma}[\Cref{lem:validity and monotone}]
%\label{lem:validity and monotone-appendix}
  $\vDash Z : \rtriple{P}{S_1}{S_2}{Q}$ if and only if for all $z\in Z$, $(\sem{S_1},\sem{S_2})$ is monotone w.r.t. $P$ and $Q$.
\end{lemma}
\begin{proof}
%[Proof of \Cref{lem:validity and monotone}]
  We employ \Cref{lem:validity quantum operation alter} to interpret monotone. 

  (if) part. For all $z\in Z$, and $\rho \in \DD^1(\HH_{S_1} \otimes \HH_{S_2})$, by assumption we have:
  $$T_{Q}(\sem{S_1}(\tr_2(\rho)),\sem{S_2}(\tr_1(\rho))) \le T_{P}(\tr_2(\rho),\tr_1(\rho))\le \tr(P\rho).$$
  Set $\sigma : \langle \sem{S_1}(\tr_2(\rho)), \sem{S_2}(\tr_1(\rho)) \rangle_p$ such that $T_{Q}(\sem{S_1}(\tr_2(\rho)),\sem{S_2}(\tr_1(\rho))) = \tr(Q\sigma)$. Then $\tr(Q\sigma)\le \tr(P\rho)$.

  (only if) part. For all $z\in Z$, $\rho_1\in\DD^1(\HH_1)$ and $\rho_2\in\DD^1(\HH_2)$, set $\rho : \<\rho_1,\rho_2\>$ such that $T_{P}(\rho_1,\rho_2) = \tr(P\rho)$. By assumption, there exists $\sigma : \langle \sem{S_1}(\tr_2(\rho)), \sem{S_2}(\tr_1(\rho)) \rangle_p$ (i.e., $\sigma : \langle \sem{S_1}(\rho_1), \sem{S_2}(\rho_2) \rangle_p$) such that:
  $$ T_{P}(\rho_1,\rho_2) = \tr(P \rho) \geq \tr(Q \sigma)\ge T_{Q}(\sem{S_1}(\rho_1),\sem{S_2}(\rho_2)).$$
\end{proof}

\begin{lemma}[$\qqRHL$ Validity Alternative]
  \label{lem:validity alter}
    $\vDash Z: \rtriple{P}{S_1}{S_2}{Q}$ if and only if for every $z \in Z$, $\rho_1 \in \DD(\HH_{S_1})$, $\rho_2 \in \DD(\HH_{S_2})$ and partial coupling $\rho : \<\rho_1,\rho_2\>_p$, there exists a partial coupling $\sigma : \langle \sem{S_1}(\rho_1), \sem{S_2}(\rho_2) \rangle_p$ such that
    \[
    \tr(P \rho) \geq \tr(Q \sigma).
    \]
\end{lemma}

\begin{proof}
  (if) part. For every $z \in Z$, $\rho_1 \in \DD(\HH_{S_1})$, $\rho_2 \in \DD(\HH_{S_2})$, select partial coupling $\rho : \<\rho_1,\rho_2\>_p$ such that $T_P(\rho_1,\rho_2) = \tr(P\rho)$, by assumption there exists a partial coupling $\sigma : \langle \sem{S_1}(\rho_1), \sem{S_2}(\rho_2) \rangle_p$ such that
  $$T_Q(\sem{S_1}(\rho_1), \sem{S_2}(\rho_2)) \le \tr(Q\sigma) \le \tr(P\rho) = T_P(\rho_1,\rho_2),$$
  so $(\sem{S_1}, \sem{S_2})$ is monotone w.r.t. $P$ and $Q$, then by \Cref{lem:validity and monotone}.

  (only if) part. For every $z \in Z$, by \Cref{lem:validity and monotone}, $(\sem{S_1}, \sem{S_2})$ is monotone w.r.t. $P$ and $Q$. For $\rho_1 \in \DD(\HH_{S_1})$, $\rho_2 \in \DD(\HH_{S_2})$ and partial coupling $\rho : \<\rho_1,\rho_2\>_p$, select partial coupling $\sigma : \<\sem{S_1}(\rho_1), \sem{S_2}(\rho_2)\>_p$ such that $T_Q(\sem{S_1}(\rho_1), \sem{S_2}(\rho_2)) = \tr(Q\sigma)$, then:
  $$ \tr(Q\sigma) = T_Q(\sem{S_1}(\rho_1), \sem{S_2}(\rho_2)) \le T_P(\rho_1,\rho_2) \le \tr(P\rho). $$
\end{proof}

\begin{lemma}[$\qqRHL$ Validity for AST Programs]
\label{lem:validity AST program}
  Suppose $S_1$ and $S_2$ are AST programs. Then $\vDash Z: \rtriple{P}{S_1}{S_2}{Q}$ if and only if for every $z \in Z$, and $\rho \in \DD(\HH_{S_1} \otimes \HH_{S_2})$, there exists a coupling $\sigma : \langle \sem{S_1}(\tr_2(\rho)), \sem{S_2}(\tr_1(\rho)) \rangle$ such that
  \[
  \tr(P_z \rho) \geq \tr(Q_z \sigma).
  \]
\end{lemma}
\begin{proof}
%[Proof of \Cref{lem:validity AST program}]
  Direct by \Cref{lem:validity and monotone} and \Cref{lem:validity quantum channel}.
\end{proof}

\begin{theorem}[\Cref{thm:soundness}]
%\label{thm:soundness-appendix}
If $\vdash Z: \rtriple{P}{S_1}{S_2}{Q}$ then 
  $\vDash Z: \rtriple{P}{S_1}{S_2}{Q}$.
\end{theorem}

\begin{proof}
%[Proof of \Cref{thm:soundness}]
  By induction on the structural of the program, analogously to the soundness proof in \cite{barthe_rqpd}.

  \textbf{(skip), (assign-L), (apply-L)}: By employing \Cref{lem:validity and monotone} and \Cref{prop:monotone basic}(\ref{lem:judgment two side}) and the denotational semantics.

  \textbf{(seq)}: By employing \Cref{lem:validity and monotone} and \Cref{prop:monotone basic}(\ref{lem:judgment seq}) and the denotational semantics.  
  % We employ Lemma \ref{lem:validity alter} to interpret the validity. For any $\rho_1$ and $\rho_1'$, and any partial coupling $\rho : \<\rho_1,\rho_1'\>_p$, by the first assumption, there exists a partial coupling $\sigma : \<\sem{S_1}(\rho_1),\sem{S_1'}(\rho_1')\>_p$ such that $\tr(P\rho) \geq \tr(Q\sigma)$. From the second assumption there is a partial coupling $\sigma': \< \sem{S_2}(\sem{S_1}(\rho_1)), \sem{S_2'}(\sem{S_1'}(\rho_1')) \>_p$ such that $\tr(Q\sigma) \geq \tr(R\sigma')$. Finally, noting that $\sigma' : \< \sem{S_1; S_2}(\rho_1), \sem{S_1'; S_2'}(\rho_1')\>$ concludes the proof.

  \textbf{(if-L)}: By employing \Cref{lem:validity and monotone} and \Cref{lem:validity quantum operation alter}, it remains to be shown that for all $\rho\in\DD^1$ and $\sigma\in\DD^1$,
  $$ T_Q(\sem{\ifb}(\rho),\sigma) \le T_{\sum_m (M_m \otimes I)^\dagger_m P_m (M_m \otimes I)}(\rho,\sigma).$$
  Set $\rho_m = \EE_m(\rho) = M_m\rho M_m^\dagger$ and $p_m = \tr(\rho_m)$, so $\sem{\ifb}(\rho) = \sum_mp_m\sem{S_m}(\rho_m/p_m)$.
  By assumption and \Cref{lem:QOT convexity}, we have:
  \begin{align*}
    T_Q(\sem{\ifb}(\rho),\sigma) = T_Q(\sum_mp_m\sem{S_m}(\rho_m/p_m),\sigma) 
    \le \sum_mp_mT_Q(\sem{S_m}(\rho_m/p_m),\sigma) \le\sum_mp_mT_{P_m}(\rho_m/p_m,\sigma)
  \end{align*}
  On the other hand, select $\delta : \< \rho, \sigma \>$ such that 
  \begin{align*}
    T_{\sum_m (M_m \otimes I)^\dagger_m P_m (M_m \otimes I)}(\rho,\sigma) &= \tr((\sum_m (M_m \otimes I)^\dagger_m P_m (M_m \otimes I))\delta) \\
    &= \sum_m p_m\tr(P_m ((M_m \otimes I)\delta(M_m \otimes I)^\dagger_m /p_m)).
  \end{align*}
  Notice that, $(M_m \otimes I)\delta(M_m \otimes I)^\dagger_m /p_m$ is in fact a coupling of $\<\rho_m/p_m,\sigma\>$, so 
  $$\tr(P_m ((M_m \otimes I)\delta(M_m \otimes I)^\dagger_m /p_m))\ge T_{P_m}(\rho_m/p_m,\sigma),$$
  and this complete the proof.

  \textbf{(while-L)}: 
  Let $z\in Z$ and $\rho$ be the initial joint state.
  Set $\sigma_0 = \rho$, and inductively define 
  $$\sigma_{n+1}' : \<\sem{S}(\tr_2(\EE_1(\sigma_n)/\tr(\EE_1(\sigma_n)))), \tr_1(\EE_1(\sigma_n)/\tr(\EE_1(\sigma_n)))\>_p$$ 
  as the partial coupling obtained by applying (IH) on initial state $\EE_1(\sigma_n)/\tr(\EE_1(\sigma_n))$, and $\sigma_{n+1} = \tr(\EE_1(\sigma_n))\sigma_{n+1}$. 
  For simplicity, we write 
  $$R =((M_0 \otimes I)^\dagger P (M_0 \otimes I) + (M_1 \otimes I)^\dagger Q (M_1 \otimes I)).$$
  By (IH), we have 
  $\tr(Q \EE_1(\sigma_n)/\tr(\EE_1(\sigma_n))) \ge \tr(R\sigma_{n+1}')$, or equivalently, $\tr(Q \EE_1(\sigma_n)) \ge \tr(R\sigma_{n+1})$.

  We first check $\tr_2(\sigma_n) = (\sem{S} \circ \EE_1)^n (\tr_2(\rho))$ by induction on $n$. The base case $n = 0$ is trivial. For $n+1$, observe that:
  \begin{align*}
    \tr_2(\sigma_{n+1}) &= \tr(\EE_1(\sigma_n)) \tr_2(\sigma_{n+1}') \\
    &\sqsubseteq \tr(\EE_1(\sigma_n))\sem{S}(\tr_2(\EE_1(\sigma_n)/\tr(\EE_1(\sigma_n)))) \\
    &= \sem{S}(\tr_2(\EE_1(\sigma_n))) \\
    &= (\sem{S} \circ \EE_1) (\tr_2(\sigma_n)).
  \end{align*}
  On the other hand, 
  \begin{align*}
    1 + \tr(\sigma_{n+1}) &= 1 + \tr(\EE_1(\sigma_n))\tr(\sigma_{n+1}') \\
    &\ge 1 + \tr(\EE_1(\sigma_n))(\tr(\sem{S}(\tr_2(\EE_1(\sigma_n)/\tr(\EE_1(\sigma_n))))) \\
    &\quad\qquad\qquad + \tr(\tr_1(\EE_1(\sigma_n)/\tr(\EE_1(\sigma_n)))) - 1) \\
    &= 1 + \tr(\sem{S}(\tr_2(\EE_1(\sigma_n)))) \\
    &= 1 + \tr((\sem{S} \circ \EE_1) (\tr_2(\sigma_n)))
  \end{align*}
  These two together imply $\tr_2(\sigma_{n+1}) = (\sem{S} \circ \EE_1) (\tr_2(\sigma_n)) = (\sem{S} \circ \EE_1)^{n+1} (\tr_2(\rho))$.
  Furthermore, $\tr_1(\sigma_{n+1}')\sqsubseteq \tr_1(\EE_1(\sigma_n)/\tr(\EE_1(\sigma_n)))$, or equivalently, $\tr_1(\sigma_{n+1})\sqsubseteq \tr_1(\EE_1(\sigma_n))$.

  Set $\sigma = \sum_n \EE_0(\sigma_n)$. Its convergence is ensured by realizing that 
  \begin{align*}
    \tr(\sigma) &= \sum_n \tr(\EE_0(\tr_2(\sigma_n))) = \sum_n \tr(\EE_0((\sem{S} \circ \EE_1)^n (\tr_2(\rho)))) \\
    &= \tr(\sem{\while}(\tr_2(\rho))) \le \tr(\rho).
  \end{align*}

  Next, we check the inequality :
  \begin{align*}
    \tr(R \rho)
    =\ &\tr(P\EE_0(\sigma_0) + Q\EE_1(\sigma_0)) \\
    \ge\ &\tr(P\EE_0(\sigma_0)) + \tr(R\sigma_1) \\
    =\ &\tr(P\EE_0(\sigma_0)) + \tr(P\EE_0(\sigma_1) + Q\EE_1(\sigma_1)) \\
    =\ &\sum_{n=0}^k\tr(P\EE_0(\sigma_n)) +  \tr(Q\EE_1(\sigma_k)) \\
    \ge\ &\tr(P\sigma).
  \end{align*}
  Thus, it is sufficient to check $\sigma : \< \sem{\while}(\tr_2(\rho)), \tr_1(\rho)\>_p$ as follows:
  \begin{align*}
    \tr_2(\sigma) &= \sum_n \EE_0(\tr_2(\sigma_n)) = \sum_n\EE_0(\sem{S} \circ \EE_1)^n (\tr_2(\rho)) \\
    &= \sem{\while}(\tr_2(\rho)), \\
    \tr_1(\rho) &= \tr_1(\EE_0(\sigma_0) + \EE_1(\sigma_0)) \\
    &\sqsupseteq \tr_1(\EE_0(\sigma_0)) + \tr_1(\sigma_1) \\
    &= \tr_1(\EE_0(\sigma_0)) + \tr_1(\EE_0(\sigma_1) + \EE_1(\sigma_1))\\
    &\sqsupseteq \sum_{n=0}^k \tr_1(\EE_0(\sigma_n)) + \tr_1(\EE_1(\sigma_n))\\
    &\sqsupseteq \tr_1(\sigma) \\
    1 + \tr(\sigma) &= \tr(\tr_1(\rho)) +  \tr(\sem{\while}(\tr_2(\rho))).
  \end{align*}

  \textbf{(csq)}: By employing \Cref{lem:validity and monotone} and \Cref{lem:judgment csq}.  
  % Let $\rho$ be the initial joint state and $z \in Z$. Then, because $P \geqlow P'$, we get $\tr(P\rho) \geq \tr(P' \rho)$. The (IH) gives a $\sigma$ partial coupling from initial state $\rho$ such that $\tr(P'\rho) \geq \tr(Q'\sigma)$. Finally $Q \geqlow Q'$ implies $\tr(Q'\sigma) \geq \tr(Q\sigma)$. Therefore, $\tr(P\rho) \geq \tr(Q\sigma)$.

  \textbf{(Strassen)}: By employing \Cref{lem:validity and monotone} and \Cref{lem:judgment strassen}. Note that $Q$ should be finite which is inherited from \Cref{thm:quantum strassen defect alter}.

  % By assumption, both $S_1,S_2$ are AST programs, thus we employ Lemma \ref{lem:validity AST program} to interpret the judgment.
  % For any $z\in Z$, $\rho\in\DD(\HH_{S_1}\otimes\HH_{S_2})$, the conclusion is equivalent to $\sem{S_1}(\tr_2(\rho)) Q^\#_{\tr(\Phi\rho)}\sem{S_2}(\tr_1(\rho))$, by Theorem \ref{thm:quantum strassen defect alter}, it sufficient to show that for all $Y_1\in\Pos(\HH_1)$ and $Y_2\in\Pos(\HH_2)$ such that $X\sqsupseteq Y_1\otimes I_2 - I_1\otimes Y_2$, $\tr(Y_1\sem{S_1}(\tr_2(\rho)))\le \tr(Y_2\sem{S_2}(\tr_1(\rho))) + \tr(\Phi\rho)$.

  % We choose sufficient large $n\in\NN$ such that $Y_2\sqsubseteq nI$, by assumption, there exists a coupling $\sigma : \<\sem{S_1}(\tr_2(\rho)), \sem{S_2}(\tr_1(\rho))\>$ such that:
  % \begin{align*}
  %   &\tr((\Phi+nI)\rho) = \tr(\Phi\rho) + n \\
  %   \ge\ &\tr((Y_1\otimes I_2 + I_1\otimes (nI - Y_2))\sigma) \\
  %   =\ &\tr(Y_1(\tr_2(\sigma))) - \tr(Y_2(\tr_1(\sigma))) + n \\
  %   =\ &\tr(Y_1\sem{S_1}(\tr_2(\rho))) - \tr(Y_2\sem{S_2}(\tr_1(\rho))) + n,
  % \end{align*}
  % which completes the proof.
\end{proof}

\begin{lemma}[\Cref{lemma:onesidedweakestpre}]
%\label{lemma:onesidedweakestpre}
  For every AST program $S$, we have
  \[
  \vdash Z: \rtriple{(\sem{S}^\dagger \otimes I)(Q)}{S}{\skp}{Q}.
  \]
  % where for any CP map $\EE$, $\EE^\dagger$ is its adjoint with respect to the Hilbert-Schmidt inner product.
\end{lemma}

\begin{proof}
%[Proof of \Cref{lemma:onesidedweakestpre}]
  By induction. 

  Case $\skp$: direct from the definition and (skip).

  Case $q:= \ket{0}$: direct from the definition and (assign-L).

  Case $\bar{q} := U[\bar{q}]$: direct from the definition and (apply-L).

  Case $S_1;S_2$: Assume that for any $Q$ we have $\vdash Z: \rtriple{(\sem{S_i}^\dagger \otimes I)(Q)}{S_i}{\skp}{Q}$ for $i=1,2$, we directly get from (seq) since $((\sem{S_1;S_2}^\dagger) \otimes I)(Q) = (\sem{S_1}^\dagger\otimes I)[(\sem{S_2}^\dagger \otimes I)(Q)]$ by \Cref{lem: IVP algebraic}:
  \[
  \begin{prooftree}
    \hypo{\begin{array}{c}
      \vdash Z: \rtriple{(\sem{S_1}^\dagger\otimes I)[(\sem{S_2}^\dagger \otimes I)(Q)]}{S_1}{\skp}{(\sem{S_2}^\dagger \otimes I)(Q)}
      \\
      \vdash Z: \rtriple{(\sem{S_2}^\dagger \otimes I)(Q)}{S_2}{\skp}{Q}
    \end{array}}
    \infer1{\vdash Z: \rtriple{((\sem{S_1;S_2}^\dagger) \otimes I)(Q)}{S_1;S_2}{\skp}{Q}}
  \end{prooftree}
  \]

  Case $\ifb$: assume that for any $Q$ and $m$ we have $\vdash Z: \rtriple{\sem{S_m}^\dagger \otimes I (Q)}{S_m}{\skp}{Q}$. Then, it follows directly from (if-L) that 
  \[
  \vdash Z: \rtriple{\sum_m (M_m\otimes I)^\dagger[(\sem{S_m}^\dagger\otimes I)(Q)](M_m\otimes I)}{\ifb}{\skp}{Q}.
  \]
  By \Cref{lem: IVP algebraic}, knowing that 
  \begin{align*}
    &\sum_m (M_m\otimes I)^\dagger[(\sem{S_m}^\dagger\otimes I)(Q)](M_m\otimes I) \\
    =\ &\sum_m (\EE_m^\dagger \otimes I)[(\sem{S_m}^\dagger\otimes I)(Q)] \\
    =\ & ((\sum_m \EE_m^\dagger \circ \sem{S_m}^\dagger)\otimes I)(Q) \\
    =\ & ((\sum_m \sem{S_m} \circ \EE_m)^\dagger \otimes I)(Q)\\
    =\ & (\sem{\ifb}^\dagger \otimes I)(Q)
  \end{align*}
  then allows us to conclude that
  \[
  \vdash Z: \rtriple{(\sem{\ifb}^\dagger \otimes I)(Q)}{\ifb}{\skp}{Q}
  \]
  as required.

  Case $\while$: assume that for any $Q$ we have $\vdash Z: \rtriple{(\sem{S}^\dagger\otimes I)(Q)}{S}{\skp}{Q}$. Then, we get that 
  \[
    \vdash Z: \rtriple{ (\sem{S}^\dagger\otimes I)[(\sem{\while}^\dagger\otimes I)(Q)]}{S}{\skp}{(\sem{\while}^\dagger \otimes I) (Q)}
  \]
  % By \Cref{lem: IVP algebraic},
  Let $P = (\sem{S}^\dagger\otimes I)[(\sem{\while}^\dagger\otimes I)(Q)]$.
  Then, by the (least) fixed point property of $\sem{\while}$, i.e., $\sem{\while} = \EE_0 + \sem{\while} \circ (\sem{S}\circ\EE_1)$, we obtain that
  \begin{align*}
    (\sem{\while}^\dagger \otimes I) (Q) &= 
    [(\EE_0^\dagger + (\EE_1^\dagger \circ \sem{S}^\dagger) \circ \sem{\while}^\dagger) \otimes I] (Q) \\
    &= (\EE_0^\dagger \otimes I) (Q) + (\EE_1^\dagger \otimes I) (P)
  \end{align*}
  by \Cref{lem: IVP algebraic}.
  We thus obtain by (while-L) that 
  \[
  \begin{prooftree}
    \hypo{\vdash Z: \rtriple{P}{S}{\skp}{(\EE_0^\dagger \otimes I) (Q) + (\EE_1^\dagger \otimes I) (P)}}
    \infer1{\vdash Z: \rtriple{(\EE_0^\dagger \otimes I) (Q) + (\EE_1^\dagger \otimes I) (P)}{\while}{\skp}{Q}}
  \end{prooftree}
  \]
  which is the same as 
  \[
    \vdash Z: \rtriple{(\sem{\while}^\dagger \otimes I) (Q)}{\while}{\skp}{Q}
  \]
  as required.
\end{proof}

\begin{lemma}[\Cref{lemma:twosidedweakestpre}]
%\label{lemma:twosidedweakestpre-appendix}
  For every AST programs $S_1, S_2$, we have
  \[
    \vdash Z: \rtriple{(\sem{S_1}^\dagger \otimes \sem{S_2}^\dagger) (Q)}{S_1}{S_2}{Q}
  \]
\end{lemma}

\begin{proof}
%[Proof of \Cref{lemma:twosidedweakestpre}]
  By using \Cref{lemma:onesidedweakestpre}, its symmetric version and the rule (seq) we get:
  \[
    \begin{prooftree}
      \hypo{
        \begin{array}{c} 
          \vdash Z: \rtriple{(\sem{S_1}^\dagger \otimes I)[(I\otimes \sem{S_2}^\dagger) (Q)]}{S_1}{\skp}{(I \otimes \sem{S_2}^\dagger) (Q)}\\ \vdash Z: \rtriple{(I \otimes \sem{S_2}^\dagger) (Q)}{\skp}{S_2}{Q}
        \end{array}
      }
      \infer1[(seq)]{\vdash Z: \rtriple{(\sem{S_1}^\dagger \otimes \sem{S_2}^\dagger) (Q)}{S_1}{S_2}{Q}}
    \end{prooftree}
  \]
  as required by \Cref{lem: IVP algebraic}.
\end{proof}

\begin{theorem}[\Cref{thm:weakcompleteness}]
%\label{thm:weakcompleteness-appendix}
  For every AST programs $S_1, S_2$, we have:
  % any $Z$, any $Z$-parameterised predicate $P$ and $Z$-parameterised $Q_1 \in \PosI(\HH_{S_1})$ and $Q_2\in \PosI(\HH_{S_2})$, if 
  \[
    \vDash Z: \rtriple{P}{S_1}{S_2}{Q_1 \otimes I + I \otimes Q_2}
  \]
  implies
  \[
   \vdash Z: \rtriple{P}{S_1}{S_2}{Q_1 \otimes I + I \otimes Q_2}
  \]
\end{theorem}
\begin{proof}
%[Proof of \Cref{thm:weakcompleteness}]
  According to \Cref{lem:validity and monotone} and \Cref{lem:judgment two side wlp}, by assumption, it holds that
  \begin{align*}
    P &\geqlow (\sem{S_1}^\dagger (Q_1)) \otimes I + I \otimes (\sem{S_2}^\dagger (Q_2)) 
    \\ &= (\sem{S_1} \otimes \sem{S_2})^\dagger (Q_1 \otimes I + I \otimes Q_2)
  \end{align*}
  by employing \Cref{lem: IVP algebraic}, and since $S_1, S_2$ are AST, so $\sem{S_1}^\dagger$ and $\sem{S_2}^\dagger$ are unital maps, i.e., $\sem{S_1}^\dagger(I) = I$ and $\sem{S_2}^\dagger(I) = I$.

  % The assumption means that for all initial $\rho$ there exists a coupling $\sigma$ of $\langle \sem{S_1}(\tr_2(\rho)), \sem{S_2}(\tr_1(\rho)) \rangle$ such that 
  % \begin{align*}
  %   \tr(P\rho) &\geq \tr((Q_1 \otimes I + I \otimes Q_2)(\sigma))\\
  %   &= \tr(Q_1 (\tr_2 \sigma)) + \tr(Q_2 (\tr_1 \sigma))
  %   \\ &= \tr(Q_1(\sem{S_1}(\tr_2(\rho)))) + \tr(Q_2(\sem{S_2}(\tr_1(\rho))))
  % \end{align*}
  % This is then equivalent to 
  % \begin{align*}
  %   P &\geqlow (\sem{S_1}^\dagger (Q_1)) \otimes I + I \otimes (\sem{S_2}^\dagger (Q_2)) 
  %   \\ &= (\sem{S_1} \otimes \sem{S_2})^\dagger (Q_1 \otimes I + I \otimes Q_2)
  % \end{align*}
  % where the last step holds because $S_1, S_2$ are AST, which makes $\sem{S_1}$ and $\sem{S_2}$ \lz{CPTP maps (their dual are unital CP maps)}. 
  Then, by \Cref{lemma:twosidedweakestpre} and the (csq) rule, we can derive
  \[
    \begin{prooftree}
      \hypo{\begin{array}{c}
        P \geqlow (\sem{S_1} \otimes \sem{S_2})^\dagger (Q)\\
        \vdash Z: \rtriple{(\sem{S_1} \otimes \sem{S_2})^\dagger (Q)}{S_1}{S_2}{Q}
      \end{array}}
      \infer1{\vdash Z: \rtriple{P}{S_1}{S_2}{Q}}
    \end{prooftree}
  \]
  where $Q = Q_1 \otimes I + I \otimes Q_2$, as required.
\end{proof}

% \begin{theorem}[\Cref{thm:completeness post cond}]
% %\label{thm:completeness post cond}
%   For every AST $S_1, S_2$ programs and bounded predicate $Q\in
%   \Pos(\HH_{S_1} \otimes \HH_{S_2})$, we have
%   \[\vDash Z:
%   \rtriple{P}{S_1}{S_2}{Q}\] implies
%   \[\vdash Z:
%   \rtriple{P}{S_1}{S_2}{Q}\]
% \end{theorem}

\begin{theorem}[\Cref{thm:completeness post cond}]
  For every AST $S_1, S_2$ programs and bounded predicate $Q\in
  \Pos(\HH_{S_1} \otimes \HH_{S_2})$, we have
  \[\vDash Z:
  \rtriple{P}{S_1}{S_2}{Q}\] implies
  \[\vdash Z:
  \rtriple{P}{S_1}{S_2}{Q}\]
\end{theorem}
\begin{proof}
By \Cref{lem:validity and monotone} and \Cref{lem:judgment strassen}, we know that
$$\vDash Z, (Y_1, Y_2, n) \in \mathcal{Y}: \rtriple{P + nI}{\\ S_1}{S_2}{Y_1 \otimes I + I \otimes (nI - Y_2)}$$
where $\mathcal{Y}$ is defined as in \Cref{lem:judgment strassen}
with $C_i = P$ and $C_o = Q$.

Follows from  \Cref{thm:weakcompleteness} the completeness for split post-conditions, we know that:
$$\vdash Z, (Y_1, Y_2, n) \in \mathcal{Y}: \rtriple{P + nI}{\\ S_1}{S_2}{Y_1 \otimes I + I \otimes (nI - Y_2)}.$$
Finally by applying rule (duality), we have 
$$\vdash Z: \rtriple{P}{S_1}{S_2}{Q}.$$
\end{proof}

\begin{proposition}[\Cref{prop:semanticembeddingrqPD}]
  % \label{prop:semanticembeddingrqPD main}
  Let $S_1$, $S_2$ be AST programs and $0\leqlow P, Q \leqlow I$ be predicates. Then 
  $$\vDash_{\rqPD}\rtriple{P}{S_1}{S_2}{Q} \iff\ \vDash \rtriple{I-P}{S_1}{S_2}{I-Q}.$$
\end{proposition}
\begin{proof}
  From \cite{barthe_rqpd}, we know that
  $\vDash_\rqPD\rtriple{P}{S_1}{S_2}{Q}$
  iff for all $\rho\in\DD(\HH_{S_1}\otimes\HH_{S_2})$, there exists a coupling $\sigma : \<\sem{S_1}(\tr_2(\rho)), \sem{S_2}(\tr_1(\rho))\>$ such that 
  $$\tr(P\rho)\le\tr(Q\sigma).$$
  Since $S_1,S_2\in\AST$, $\tr(\rho) = \tr(\sigma)$, thus it is equivalent to
  $$\tr((I-Q)\sigma)\le\tr((I-P)\rho),$$
  which, according to \Cref{lem:validity AST program}, it equivalent to $\vDash\rtriple{I-P}{S_1}{S_2}{I-Q}$ in our logic.
  Since $P,Q\sqsubseteq I$, so it is guaranteed that $I-P, I-Q\in\Pos$.
\end{proof}

\begin{proof}[Proof of \Cref{prop:semanticembedding projector logic}]
  ($\Rightarrow$) part. For any $\rho\in\DD(\HH_{S_1}\otimes\HH_{S_2})$ such that $\supp(\rho)\subseteq X$, by assumption and \Cref{lem:validity AST program}, we know that there exists a coupling $\sigma : \<\sem{S_1}(\tr_2(\rho)), \sem{S_2}(\tr_1(\rho))\>$ such that 
  $$0 \le \tr(Y^\bot \sigma)\le \tr((X\mid 0)\rho) = 0,$$
  the last equation is due to \Cref{lemma:basic-properties-of-ops-of-ivp} as $\tr(X^\bot\rho) = 0$. This asserts that $\supp(\sigma)\subseteq Y$, which conclude that $\vDash_\pqRHL\rtriple{X}{S_1}{S_2}{Y}$.

  ($\Leftarrow$) part. By \Cref{lem:validity AST program}, for any $\rho\in\DD(\HH_{S_1}\otimes\HH_{S_2})$, if $\tr(X^\bot\rho)\neq 0$, then $\tr((X\mid 0)\rho) = +\infty$, by convention then it holds. Otherwise, $\tr(X^\bot\rho) = 0$, so $\tr((X\mid 0)\rho) = 0$ and $\supp(\rho)\subseteq X$, by assumption, there exists a coupling $\sigma : \<\sem{S_1}(\tr_2(\rho)), \sem{S_2}(\tr_1(\rho))\>$ such that $\supp(\sigma) \subseteq Y$, which leads to
  $$\tr(Y^\bot \sigma) = 0 \le \tr((X\mid 0)\rho),$$
  and this completes the proof.
\end{proof}

\begin{proposition}
  \label{prop:postcondition supp}
  For programs $S_1$, $S_2$, any $Z$, and Z-parameterised predicate P and Z-parameterised projector $X\in\cS(\HH_{S_1}\otimes\HH_{S_2})$,
  the following holds:
  \begin{equation*}
  \vDash Z : \rtriple{X\mid 0}{S_1}{S_2}{P} \iff\
  \vDash Z : \rtriple{X\mid 0}{S_1}{S_2}{\supp(P)}.
  \end{equation*}
  As a corollary, if $X\in\cS(\HH_{S_1}\otimes\HH_{S_2})$, then 
  \begin{equation*}
  \vDash Z : \rtriple{X\mid 0}{S_1}{S_2}{Y\mid 0} \iff\
  \vDash Z : \rtriple{X\mid 0}{S_1}{S_2}{Y^\bot}.
  \end{equation*}
\end{proposition}
\begin{proof}
  For any $z\in Z$, and $\rho\in\DD^1(\HH_{S_1}\otimes\HH_{S_2})$,
  if $\tr(X\rho) \neq 0$, then $\tr((X\mid 0)\rho) = +\infty$, thus both LHS and RHS holds, i.e., exists partial coupling $\sigma$ of outputs such that $\tr((X\mid 0)\rho)\ge \tr(P\sigma)$ and $\tr((X\mid 0)\rho)\ge \tr(\supp(P)\sigma)$. Otherwise, $\tr(X\rho) = 0$, then $\tr((X\mid 0)\rho) = 0$, then for any partial coupling $\sigma$ of outputs, $\tr(P\sigma)\le \tr((X\mid 0)\rho) = 0$ if and only if $\tr(\supp(P)\sigma)\le \tr((X\mid 0)\rho) = 0$, so LHS is equivalent to RHS.
\end{proof}

\section{Deferred Proofs in ``Two-Sided Rules'' Section}

\begin{theorem}[\Cref{thm:side-condition equal}]
%\label{thm:side-condition equal}
  For AST programs $S_1,S_2$, and measurements $M = \{M_1,\cdots,M_k\}$ and $N = \{N_1,\cdots,N_k\}$, the following are equivalent:
  \begin{enumerate}
    \item $\emptyset\stackrel{(S_1,S_2)}{\vDash}M\approx N$;
    \item $\vDash (Y_1,\cdots,Y_k,Z_1,\cdots,Z_k,n)\in\mathcal{Y}_k: \{nI\}S_1\sim S_2\big\{ (\mbox{$\sum_i$}M_i^\dagger Y_i M_i)\otimes I +
      I \otimes \big[nI - (\mbox{$\sum_i$}N_i^\dagger Z_i N_i)\big]\big\}$
      where $\mathcal{Y}_k = \{  (Y_1,\cdots,Y_k,Z_1,\cdots,Z_k,n) \mid
      \forall\,i,\ 0\sqsubseteq Y_i,\ 0\sqsubseteq Z_i\sqsubseteq nI, Y_i\otimes I - I\otimes Z_i\sqsubseteq 0,
      \forall\,j\neq i,\ Y_i\otimes I - I\otimes Z_j\sqsubseteq I\}.$
  \end{enumerate}
\end{theorem}

\begin{proof}
%[Proof of \Cref{thm:side-condition equal}]
  ($1\Rightarrow 2$). By \Cref{lem:validity AST program} , for any $\rho\in \DD(\HH_{S_1}\otimes\HH_{S_2})$, set $\rho_1 = \tr_2(\rho)$ and $\rho_2 = \tr_1(\rho)$, $\sigma_1 = \sem{S_1}(\rho_1)$ and $\sigma_2 = \sem{S_2}(\rho_2)$. By first assumption, we know that for all $i$, $\tr(M_i\sigma_1M_i^\dagger) = \tr(N_i\sigma_2N_i^\dagger)$ and denote it by $p_i$. Set the state 
  $$\sigma_1' = \begin{pmatrix}
      M_1\sigma_1M_1^\dagger & 0 & \cdots & 0 \\
      0 & M_2\sigma_1M_2^\dagger & \cdots & 0 \\
      \vdots & \vdots & \ddots & \vdots \\
      0 & 0 & \cdots & M_k\sigma_1M_k^\dagger
  \end{pmatrix}
  \qquad
  \sigma_2' = \begin{pmatrix}
      N_1\sigma_2N_1^\dagger & 0 & \cdots & 0 \\
      0 & N_2\sigma_2N_2^\dagger & \cdots & 0 \\
      \vdots & \vdots & \ddots & \vdots \\
      0 & 0 & \cdots & N_k\sigma_2N_k^\dagger
  \end{pmatrix}$$
  and obviously, $\sigma'_1 \in \DD(\HH_k\otimes\HH_{S_1})$ and $\sigma'_2 \in \DD(\HH_k\otimes\HH_{S_2})$ where $\HH_k$ is the $k$-dimensional Hilbert space. In other words, $\sigma_1' = \sum_i|i\>\<i|\otimes M_i\sigma_1M_i^\dagger$ and $\sigma_2' = \sum_i|i\>\<i|\otimes N_i\sigma_2N_i^\dagger$.
  Consider the PSD $A = \sum_{i\neq j} (|i\>\<i|\otimes I_{S_1}) \otimes (|j\>\<j|\otimes I_{S_2}) \in \Pos((\HH_k\otimes\HH_{S_1})\otimes(\HH_k\otimes\HH_{S_2}))$. We claim that $\sigma'_1A_0^\#\sigma'_2$, since we can construct coupling
  $$
  \sigma' = \sum_i (|i\>\<i|\otimes M_i\sigma_1M_i^\dagger)\otimes (|i\>\<i|\otimes N_i\sigma_2N_i^\dagger) / p_i,
  $$
  which $\sigma : \<\sigma'_1, \sigma'_2\>$ and $\tr(A\sigma') = 0$.
  By \Cref{thm:quantum strassen defect alter}, we know that for all $Y\in\Pos(\HH_k\otimes\HH_{S_1})$ and $Z\in\Pos(\HH_k\otimes\HH_{S_2})$ such that $A\sqsupseteq Y\otimes (I_k\otimes I_{S_2}) - (I_k\otimes I_{S_1})\otimes Z$, it holds that:
  $\tr(Y\sigma'_1)\le \tr(Z\sigma'_2)$.
  Now, back to (2) which we aim to prove, for any $(Y_1,\cdots,Y_k,Z_1,\cdots,Z_k,n)\in \mathcal{Y}_k$, set $Y = \sum_i|i\>\<i|\otimes Y_i$ and $Z = \sum_i|i\>\<i|\otimes Z_i$, we check that:
  \begin{align*}
      Y\otimes (I_k\otimes I_{S_1}) - (I_k\otimes I_{S_2})\otimes Z &= \Big(\sum_i|i\>\<i|\otimes Y_i\Big)\otimes\Big(\sum_j|j\>\<j|\otimes I_{S_2}\Big) - \Big(\sum_i|i\>\<i|\otimes I_{S_1}\Big)\otimes\Big(\sum_j|j\>\<j|\otimes Z_j\Big) \\
      &= \sum_i[(|i\>\<i|\otimes Y_i)\otimes (|i\>\<i|\otimes I_{S_2}) - (|i\>\<i|\otimes I_{S_1})\otimes (|i\>\<i|\otimes Z_i)] - \\
      &\quad \sum_{i\neq j}[(|i\>\<i|\otimes Y_i)\otimes (|j\>\<j|\otimes I_{S_2}) - (|i\>\<i|\otimes I_{S_1})\otimes (|j\>\<j|\otimes Z_j)]\\
      &= \sum_i(|i\>\<i|\otimes|i\>\<i|)\otimes(Y_i\otimes I_{S_2} - I_{S_1}\otimes Z_i) + \sum_{i\neq j}(|i\>\<i|\otimes|j\>\<j|)\otimes(Y_i\otimes I_{S_2} - I_{S_1}\otimes Z_j) \\
      &\sqsubseteq \sum_{i\neq j}(|i\>\<i|\otimes|j\>\<j|)\otimes(I_{S_1}\otimes I_{S_2}) \\
      &= A
  \end{align*}
  where, in the fourth and fifth line, we change the order of Hilbert space $(\HH_k\otimes\HH_{S_1})\otimes(\HH_k\otimes\HH_{S_2})\rightarrow(\HH_k\otimes\HH_k)\otimes(\HH_{S_1}\otimes\HH_{S_2})$ as they are isomorphic. Thus, it holds that $\tr(Y\sigma'_1)\le \tr(Z\sigma'_2)$, or equivalently,
  \begin{align*}
      0 \ge \tr(Y\sigma'_1) - \tr(Z\sigma'_2) &= \tr\Big(\Big(\sum_i|i\>\<i|\otimes Y_i\Big)\Big(\sum_i|i\>\<i|\otimes M_i\sigma_1M_i^\dagger\Big)\Big)
      - \tr\Big(\Big(\sum_i|i\>\<i|\otimes Z_i\Big)\Big(\sum_i|i\>\<i|\otimes N_i\sigma_2N_i^\dagger\Big)\Big) \\
      &= \sum_i\tr(Y_iM_i\sigma_1M_i^\dagger) - \sum_i\tr(Z_iN_i\sigma_2N_i^\dagger) \\
      &= \tr\big((\mbox{$\sum_i$}M_i^\dagger Y_iM_i)\sigma_1\big) - 
         \tr\big((\mbox{$\sum_i$}N_i^\dagger Z_iN_i)\sigma_2\big).
  \end{align*}
  Note that $\tr(\rho) = \tr(\sigma_1) = \tr(\sigma_2)$ and set it as $p$. Let $\sigma\triangleq \sigma_1\otimes\sigma_2/p$, realizing that $\sigma : \<\sigma_1,\sigma_2\>$, and observe that 
  \begin{align*}
      \tr((nI)\rho) &\ge \tr((nI)\sigma) + \tr\big((\mbox{$\sum_i$}M_i^\dagger Y_iM_i)\sigma_1\big) - 
         \tr\big((\mbox{$\sum_i$}N_i^\dagger Y_iN_i)\sigma_2\big) \\
        &= \tr((nI)\sigma) + \tr\big(((\mbox{$\sum_i$}M_i^\dagger Y_iM_i)\otimes I_{S_2})\sigma\big) - 
         \tr\big((I_{S_1}\otimes(\mbox{$\sum_i$}N_i^\dagger Z_iN_i))\sigma\big) \\
        &= \tr\big(\big((\mbox{$\sum_i$}M_i^\dagger Y_i M_i)\otimes I +
      I \otimes \big[nI - (\mbox{$\sum_i$}N_i^\dagger Z_i N_i)\big]\big)\sigma\big)
  \end{align*}
  As $\rho$ is arbitrary, so we finish the proof.
  
  ($2\Rightarrow 1$). 
  For any $\rho_1\in \DD^1(\HH_{S_1})$ and $\rho_2\in \DD^1(\HH_{S_2})$, by \Cref{lem:validity AST program} and \Cref{lem:validity quantum operation alter}, choose $\rho\triangleq \rho_1\otimes\rho_2\in\DD(\HH_{S_1}\otimes\HH_{S_2})$ which is a coupling of $\<\rho_1,\rho_2\>$. For any $i = 1,\cdots k$, select $Y_i = I$, $Z_i = I$, $n = 1$, $Y_j = 0$ and $Z_j = 0$ for all $j\neq i$. It is then obvious that $(Y_1,\cdots,Y_k,Z_1,\cdots,Z_k,1)\in\mathcal{Y}$, so by assumption, there exists a coupling $\sigma : \<\sem{S_1}(\tr_2(\rho)), \sem{S_2}(\tr_1(\rho))\>$ (or, equivalently, $\sigma : \<\sem{S_1}(\rho_1), \sem{S_2}(\rho_2)\>$) such that
  \begin{align*}
      n\tr(\rho) \ge\ &\tr\big(\big(\big(\mbox{$\sum_j$}M_j^\dagger Y_j M_j\big)\otimes I +
      I \otimes \big[nI - (\mbox{$\sum_j$}N_j^\dagger Z_j N_j)\big]\big)\big)\sigma\big) \\
      =\ &\tr(((M_i^\dagger M_i)\otimes I + nI - I\otimes(N_i^\dagger N_i))\sigma) \\
      =\ &\tr(M_i^\dagger M_i\tr_2(\sigma)) - \tr(N_i^\dagger N_i\tr_1(\sigma)) + n\tr(\sigma) \\
      =\ &\tr(M_i(\sem{S_1}(\rho_1))M_i^\dagger) - \tr(N_i(\sem{S_2}(\rho_2))N_i^\dagger) + n\tr(\rho),
  \end{align*}
  since $S_1,S_2\in\AST$. Thus, for all $i$, $\tr(M_i(\sem{S_1}(\rho_1))M_i^\dagger) \le \tr(N_i(\sem{S_2}(\rho_2))N_i^\dagger)$. Notice that 
  $$\sum_i\tr(M_i(\sem{S_1}(\rho_1))M_i^\dagger) = \tr(\sem{S_1}(\rho_1)) = \tr(\rho_1) = \tr(\rho_2) = \tr(\sem{S_2}(\rho_2)) = \sum_i\tr(N_i(\sem{S_2}(\rho_2))N_i^\dagger),$$
  so it must be the case that $\tr(M_i(\sem{S_1}(\rho_1))M_i^\dagger) = \tr(N_i(\sem{S_2}(\rho_2))N_i^\dagger)$ for all $i$, i.e., $(\sem{S_1}(\rho_1), \sem{S_2}(\rho_2))\vDash M\approx N$, and this completes the proof.
\end{proof}

% \begin{cor}
%   For AST programs $S_1,S_2$, and measurements $M, M'$ and $N,N'$ with the same output sets $\{1,\cdots, k\}$, then if $\vDash \mathcal{Y}_k: \{P\}S_1\sim S_2\{Q\}$ where
%   \begin{align*}
%     & P = \big(\mbox{$\sum_i$}M_i^\dagger Y_i M_i\big)\otimes I +
%     I \otimes \big[nI - \big(\mbox{$\sum_i$}M_i^{\prime\dagger} Z_i M'_i\big)\big]\\
%     & Q = (\mbox{$\sum_i$}N_i^\dagger Y_i N_i)\otimes I +
%     I \otimes \big[nI - (\mbox{$\sum_i$}N_i^{\prime\dagger} Z_i N'_i)\big]
%   \end{align*}
%   for $(Y_1,\cdots,Y_k,Z_1,\cdots,Z_k,n)\in\mathcal{Y}_k$, then 
%   $$M\approx M' \stackrel{(S_1,S_2)}{\vDash}N\approx N'.$$
% \end{cor}
% \begin{proof}
%     \lz{todo}
% \end{proof}

\begin{definition}[Measurement Property, c.f. \cite{barthe_rqpd}]
  Define $\Gamma \vDash Z : \{P \} M\approx N \{Q_k\}$
  if for all $\rho,\sigma\in\DD^1$ such that $(\rho,\sigma)\vDash \Gamma$ and $z\in Z$, 
  if $T_P(\rho,\sigma) < +\infty$, then there exists \emph{couplings}
  $\delta_k : \< M_k\rho M_k^\dagger, N_k\sigma N_k^\dagger \>$ such that:
  $$T_P(\rho,\sigma) \ge \sum_k\tr(Q_k\delta_k).$$
\end{definition}

\begin{proposition}
\label{prop:measurement properties embed}
1). $M\approx N\vDash_{\rqPD} A\Rightarrow \{B_m\}$ if and only if 
$M\approx N \vDash \{I - A \} M\approx N \{I - B_m\} $, 
where $0\sqsubseteq A, B_m\sqsubseteq I$.
2). $\vDash_{\pqRHL} M\approx N : X\Rightarrow \{Y_m\}$ if and only if 
$\vDash \{X\mid 0\} M\approx N \{Y_m^\bot\}$ where $X,Y_m\in\cS$.
\end{proposition}
% \begin{proposition}[\Cref{prop:measurement properties embed}]
% %\label{prop:measurement properties embed}
% $M_1\approx M_2\vDash_{\rqPD} A\Rightarrow \{B_m\}$ if and only if 
% $M_1\approx M_2 \vDash \{A \} M_1\approx M_2 \{B_m\} $, 
% where $0\sqsubseteq A, B_m\sqsubseteq I$.
% $\vDash_{\pqRHL} M_1\approx M_2 : X\Rightarrow \{Y_m\}$ if and only if 
% $\vDash \{X\mid 0\} M_1\approx M_2 \{Y_m^\bot\}$.
% \end{proposition}
\begin{proof}
%[Proof of \Cref{prop:measurement properties embed}]
We first prove clause (1).

    (if) part. By Def. 5.2 and 5.4 in \cite{barthe_rqpd}, for any $\rho\in\DD^1(\HH_1\otimes\HH_2)$ such that $\rho\vDash_\rqPD M\approx N$ (if $\rho$ is partial, then we just normalize it and everything still holds since coupling, trace etc are all scalable), i.e., $\forall\,i,\ \tr(M_i\tr_2(\rho)M_i^\dagger) = \tr(N_i\tr_1(\rho)N_i^\dagger)$, in other words, $(\tr_2(\rho),\tr_1(\rho))\vDash M\approx N$. By assumption, since $A\sqsubseteq I$, so $T_{I-A}(\tr_2(\rho),\tr_1(\rho)) < +\infty$, there exists couplings $\delta_i : \<M_i\tr_2(\rho)M_i^\dagger, N_i\tr_1(\rho)N_i^\dagger\>$ such that
    \begin{align*}
        1 - \tr(A\rho) = \tr((I-A)\rho)\ge T_{I-A}(\tr_2(\rho),\tr_1(\rho))\ge \sum_i\tr((I-B_i)\delta_i
        = 1 - \sum_i\tr(B_i\delta_i),
    \end{align*}
    that is, $\tr(A\rho)\le\sum_i\tr(B_i\delta_i)$, or equivalently, $M\approx N\vDash_{\rqPD} A\Rightarrow \{B_m\}$.

    (only if) part. For any $\rho_1,\rho_2\in \DD^1$, select coupling $\rho : \<\rho_1,\rho_2\>$ such that $\tr((I-A)\rho) = T_{I-A}(\rho_1,\rho_2)$. Since $(\rho_1,\rho_2)\vDash M\approx N$, so $\forall\,i,\ \tr(M_i\tr_2(\rho)M_i^\dagger) = \tr(N_i\tr_1(\rho)N_i^\dagger)$ as $\tr_2(\rho) = \rho_1$ and $\tr_1(\rho) = \rho_2$, which implies $\rho\vDash_\rqPD M\approx N$, by assumption, there exists couplings $\delta_i : \<M_i\tr_2(\rho)M_i^\dagger, N_i\tr_1(\rho)N_i^\dagger\>$ such that
    $\tr(A\rho) \le \sum_i\tr(B_i\delta_i)$, or equivalently,
    \begin{align*}
        T_{I-A}(\rho_1,\rho_2) = \tr((I-A)\rho) = 1 - \tr(A\rho) \ge 1 - \sum_i\tr(B_i\delta_i) = \sum_i\tr((I-B_i)\delta_i)
    \end{align*}
    as desired.

Now we prove clause (2).

    (if) part. For any $\rho_1,\rho_2\in\DD^1$ (the case for partial state is similar just by normalize everything) such that $\rho_1 X^\#\rho_2$. Let $\rho$ be the witness, thus $\supp(\rho)\subseteq X$, in other words, $\tr((X\mid 0)\rho) = \tr(0\rho) = 0$ by \Cref{def:operations-infinite-valued-pre}, thus $T_{X\mid 0}(\rho_1,\rho_2) = 0$. By assumption, there exists couplings $\delta_i : \<M_i\rho_1 M_i^\dagger, N_i\rho_2N_i^\dagger\>$ such that
    $$0 = T_{X\mid 0}(\rho_1,\rho_2) \ge \sum_i\tr(Y_i^\bot\delta_i),$$
    so $\tr(Y_i^\bot\delta_i) = 0$, or equivalently, $\supp(\delta_i)\subseteq Y_i$. So $(M_i\rho_1 M_i^\dagger)Y_i^\# (N_i\rho_2N_i^\dagger)$ for all $i$, i.e., $\vDash_{\pqRHL} M\approx N : X\Rightarrow \{Y_m\}$.

    (only if) part. For any $\rho_1,\rho_2\in \DD^1$, if $T_{X\mid 0}(\rho_1,\rho_2) = +\infty$, then it trivially holds. Otherwise, there exists a coupling $\rho : \<\rho_1,\rho_2\>$ such that $\tr((X\mid 0)\rho) < +\infty$, so $\supp(\rho) \subseteq X$, thus $\rho_1 X^\# \rho_2$, by assumption, $(M_i\rho_1 M_i^\dagger)Y_i^\# (N_i\rho_2N_i^\dagger)$, and set $\delta_i$ as the witness. Thus, $\supp(\delta_i)\subseteq Y_i$, or equivalently, $\tr(Y_i\delta_i) = 0$. Thus, $\sum_i\tr(Y_i^\bot\delta_i) = 0 \le T_{X\mid 0}(\rho_1,\rho_2)$, and this completes the proof.
\end{proof}

\begin{theorem}[\Cref{thm:soundness extra}]
%\label{thm:soundness extra}
  The extra rules for $\qqRHL$ in Fig. \ref{fig:qqrhl_extra} are sound regarding the notion of validity.
\end{theorem}
\begin{proof}
%[Proof of \Cref{thm:soundness extra}]
(assign) and (apply) are the same as applying the corresponding one-side rule twice on the left and right.

(seq+)
We employ \Cref{lem:validity and monotone} and \Cref{lem:validity quantum operation alter} to interpret judgements.
For any $z\in Z$, $\rho,\sigma\in\DD^1$ such that $(\rho,\sigma)\vDash\Gamma$, by first assumption, $T_P(\rho,\sigma)\ge T_Q(\sem{S_1}(\rho),\sem{S_1'}(\sigma))$, by entailment we know that $(\sem{S_1}(\rho),\sem{S_1'}(\sigma))\vDash\Gamma'$, then by the second assumption, it holds that 
$$T_P(\rho,\sigma)\ge T_Q(\sem{S_1}(\rho),\sem{S_1'}(\sigma))\ge T_R(\sem{S_2}(\sem{S_1}(\rho)),\sem{S_2'}(\sem{S_1'}(\sigma))) = T_R(\sem{S_1;S_2}(\rho),\sem{S_1';S_2'}(\sigma)).$$

(if)
For any $z\in Z$ and $\rho,\sigma\in\DD^1$ such that  $(\rho,\sigma)\vDash\Gamma$, if $T_P(\rho,\sigma) = +\infty$,
then obviously $T_P(\rho,\sigma)\ge T_A(\sem{\ifb}(\rho),\sem{\ifb'}(\sigma))$, and then follows by 
\Cref{lem:validity and monotone}.
If $T_P(\rho,\sigma)$ is finite, by the first assumption, 
$\tr(M_k\rho M_k^\dagger) = \tr(M'_k\sigma M_k^{\prime\dagger})$ 
(follows by the existence of coupling and let it by $\delta_k$) and denote it by $p_k$.
By the second assumption,
$$\tr(R_k\delta_k) \ge p_k T_{R_k}(M_k\rho M_k^\dagger/p_k, M'_k\sigma M_k^{\prime\dagger}/p_k) \ge p_k T_{Q}(\sem{S_k}(M_k\rho M_k^\dagger/p_k), \sem{S_k'}(M'_k\sigma M_k^{\prime\dagger}/p_k)).$$
Sum it up over $k$, we have:
\begin{align*}
T_C(\rho,\sigma) &\ge \sum_k\tr(R_k\delta_k) \\
&\ge \sum_k p_k T_{Q}(\sem{S_k}(M_k\rho M_k^\dagger/p_k), \sem{S_k'}(M'_k\sigma M_k^{\prime\dagger}/p_k)) \\
&\ge T_{Q}(\sum_k p_k\sem{S_k}(M_k\rho M_k^\dagger/p_k), \sum_k p_k\sem{S_k'}(M'_k\sigma M_k^{\prime\dagger}/p_k) \\
&= T_{Q}(\sem{\ifb}(\rho),\sem{\ifb'}(\sigma)).
\end{align*}
Where the first inequality follows from the first assumption, the third inequality is due to \Cref{lem:QOT convexity} since $\{p_k\}$ is a (sub)distribution. This completes the proof.

(while) Let $\EE_0(\cdot) = M_0(\cdot)M_0^\dagger$, $\EE_1(\cdot) = M_1(\cdot)M_1^\dagger$, $\EE'_0(\cdot) = M'_0(\cdot)M_0^{\prime\dagger}$, $\EE'_1(\cdot) = M'_1(\cdot)M_1^{\prime\dagger}$. Fix $z\in Z$. For any $\sigma\in \DD^1(\HH_{S_1}\otimes\HH_{S_2})$, let $\rho = \tr_2(\sigma)$ and $\rho' = \tr_1(\sigma)$. It is trivial if $\tr(P\sigma) = +\infty$. Otherwise, $\tr(P\sigma)<+\infty$, so $T_P(\rho,\rho')<+\infty$.
Set $\rho_0 = \rho$, $\rho'_0 = \rho'$, $p_0 = 1$, so $\rho_0, \rho'_0\in \DD^1$, and $T_p(\rho_0,\rho'_0) < +\infty$. We inductively construct $\rho_n,\rho'_n,\sigma_n, p_n, q_n$ as follows:
\begin{itemize}
    \item Since $\rho_n, \rho'_n\in \DD^1$ and $T_p(\rho_n,\rho'_n) < +\infty$, by the first assumption, there exists $\sigma_n : \<\EE_0(\rho_n),\EE'_0(\rho'_n)\>$ and $\sigma'_n : \<\EE_1(\rho_n),\EE'_1(\rho'_n)\>$ such that 
    $$T_P(\rho_n,\rho'_n)\ge \tr(Q_0\sigma_n) + \tr(Q_1\sigma'_n).$$
    Let $q_n = p_n\tr(\EE_0(\rho_n)) = p_n\tr(\EE'_0(\rho'_n))$ and $q = \tr(\EE_1(\rho_n)) = \tr(\EE'_1(\rho'_n))$.
    \item By the second assumption, we know that 
        $$T_{Q_1}(\EE_1(\rho_n)/q, \EE'_1(\rho'_n)/q) \ge T_P(\sem{S}(\EE_1(\rho_n)/q),\sem{S'}(\EE'_1(\rho'_n)/q).$$
        Select the partial coupling $\delta_n : \<\sem{S}(\EE_1(\rho_n)/q), \sem{S'}(\EE'_1(\rho'_n)/q)\>_p$ such that $T_P(\sem{S}(\EE_1(\rho_n)/q),\sem{S'}(\EE'_1(\rho'_n)/q) = \tr(P\delta_n)$.
        We set $\rho_{n+1} = \tr_2(\delta_n) / \tr(\delta_n)$ and $\rho'_{n+1} = \tr_1(\delta_n) / \tr(\delta_n)$, $p_{n+1} = p_n q\tr(\delta_n)$. Obviously, 
        \begin{align*}
            T_P(\rho_{n+1},\rho'_{n+1}) &\le \tr(P\delta_n/\tr(\delta_n)) 
            = T_P(\sem{S}(\EE_1(\rho_n)/q),\sem{S'}(\EE'_1(\rho'_n)/q)/\tr(\delta_n) \\
            &\le T_{Q_1}(\EE_1(\rho_n)/q, \EE'_1(\rho'_n)/q)/\tr(\delta_n) \\
            &\le \tr(Q_1\sigma'_n/q) / tr(\delta_n) 
            = (p_n/p_{n+1})\tr(Q_1\sigma'_n)\\
            &\le (p_n/p_{n+1}) T_P(\rho_n,\rho_n') \\
            &< +\infty.
        \end{align*}
\end{itemize}

Set $\sigma = \sum_i p_i\sigma_i$, it is sufficient to show 1) $T_P(\rho,\rho')\ge \tr(Q_0\sigma)$ and 2) $\sigma$ is a partial coupling of the outputs of two $\while$s. 

We first show (1) is true. First, by the construction above, we know that:
\begin{align*}
p_nT_p(\rho_n,\rho_n') &\ge p_n(\tr(Q_0\sigma_n) + \tr(Q_1\sigma'_n)) \\
&\ge 
\tr(Q_0(p_n\sigma_n)) + p_n(p_{n+1}/p_n)T_P(\rho_{n+1},\rho'_{n+1}) \\
&= \tr(Q_0(p_n\sigma_n)) + p_{n+1}T_P(\rho_{n+1},\rho'_{n+1}).
\end{align*}
Thus, we have:
\begin{align*}
    T_P(\rho,\rho') &= p_0T_P(\rho_0,\rho_0') \\
    &\ge \tr(Q_0(p_0\sigma_0)) + p_1T_P(\rho_1,\rho_1') \\
    &\ge \tr(Q_0(p_0\sigma_0 + p_1\sigma_1)) + p_2T_P(\rho_2,\rho_2') \\
    & \cdots \\
    &\ge \tr\Big(Q_0\Big(\sum_ip_i\sigma_i\Big)\Big) \\
    &= \tr(Q_0\sigma)
\end{align*}
To show $\sigma$ is a partial coupling, we first observe that $p_{n}\rho_n\sqsubseteq (\sem{S}\circ\EE_1)^n(\rho)$ since:
\begin{align*}
    p_{n+1}\rho_{n+1} = p_nq\tr(\delta_n)\tr_2(\delta_n)/\tr(\delta_n)
    \sqsubseteq p_nq\sem{S}(\EE_1(\rho_n)/q) = (\sem{S}\circ\EE_1)(p_n\rho_n)\sqsubseteq\cdots\sqsubseteq (\sem{S}\circ\EE_1)^{n+1}(p_0\rho_0).
\end{align*}
Similarly, $p_{n}\rho'_n\sqsubseteq (\sem{S'}\circ\EE'_1)^n(\rho')$. Thus, we have:
\begin{align*}
    p_{n+1}\tr_2(\sigma_{n+1}) &= p_{n+1}\EE_0(\rho_{n+1}) = p_{n+1}\EE_0(\tr_2(\delta_n)/\tr(\delta_n)) \\
    &\sqsubseteq p_{n+1}\EE_0(\sem{S}(\EE_1(\rho_n)/q)/\tr(\delta_n))
    = \EE_0((\sem{S}\circ\EE_1)(p_n\rho_n))) \\
    &\sqsubseteq \EE_0\circ (\sem{S}\circ\EE_1)^{n+1}(\rho))
\end{align*}
which leads to 
\begin{align*}
    \tr_2(\sigma) = \sum_np_n\tr_2(\sigma_n)\sqsubseteq \sum_n\EE_0\circ (\sem{S}\circ\EE_1)^n(\rho)) = \sem{\while[M,S]}(\rho).
\end{align*}
and similarly, $\tr_1(\sigma) \sqsubseteq \sem{\while[M,S]}(\rho')$.

We further observe that
\begin{align*}
    \tr(\rho) - \tr(\sem{\while[M,S](\rho)})
    &= \tr(\rho) - \sum_n\tr((\EE_0\circ (\sem{S}\circ \EE_1)^n)(\rho)) \\
    &= \tr(\rho) - \tr(\EE_0(\rho)) - \sum_n\tr((\EE_0\circ (\sem{S}\circ \EE_1)^{n+1})(\rho)) \\
    &= \tr(\EE_1(\rho)) - \tr((\sem{S}\circ \EE_1)(\rho)) + \\
    &\qquad \tr((\sem{S}\circ \EE_1)(\rho)) - \sum_n\tr((\EE_0\circ (\sem{S}\circ \EE_1)^n)((\sem{S}\circ \EE_1)(\rho))) \\
    &\ge \sum_n\tr((\EE_1 - \sem{S}\circ\EE_1)((\sem{S}\circ\EE_1)^n(\rho))) + \lim_n\tr((\sem{S}\circ\EE_1)^n(\rho)) \\
    &\ge \sum_n((\EE_1 - \sem{S}\circ\EE_1)(p_n\rho_n)) + \lim_n p_n
\end{align*}
where the last inequality comes from 1). $\tr((\mathcal{I} - \sem{S})\alpha)\le \tr((\mathcal{I} - \sem{S})\beta)$ where $\mathcal{I}$ is the identity quantum channel, if $\alpha,\beta\in\DD$ such that $\alpha\sqsubseteq\beta$; 2). $p_{n}\rho_n\sqsubseteq (\sem{S}\circ\EE_1)^n(\rho)$ as we proved above; 3). $\tr(p_n\rho_n) = p_n$. Similar result holds:
\begin{align*}
    \tr(\rho') - \tr(\sem{\while[M',S'](\rho')})
    \ge \sum_n((\EE'_1 - \sem{S'}\circ\EE'_1)(p_n\rho'_n)) + \lim_n p_n
\end{align*}
From the fact that $\delta_n$ is a partial coupling, we have the following equivalent forms:
\begin{align*}
    &\tr(\sem{S}(\EE_1(\rho_n)/q)) + \tr(\sem{S'}(\EE'_1(\rho'_n)/q)) \le 1 + \tr(\delta_n) \\
    \Longleftrightarrow\ &\tr((\sem{S}\circ\EE_1)(\rho_n)) + \tr((\sem{S'}\circ\EE'_1)(\rho'_n)) \le \tr(\EE_1(\rho_n)) + \tr(\EE'_1(\rho'_n)) - \tr(\EE_1(\rho_n)) + \tr(\EE_1(\rho_n))\tr(\delta_n) \\
    \Longleftrightarrow\ &
    0 \le 
    \tr((\EE_1 - \sem{S}\circ\EE_1)(\rho_n)) + \tr((\EE'_1 - \sem{S'}\circ\EE'_1)(\rho'_n)) - \tr(\EE_1(\rho_n)) + \tr(\EE_1(\rho_n))\tr(\delta_n) \\
    \Longleftrightarrow\ &
    \tr(\EE_1(p_n\rho_n)) - p_{n+1} \le 
    \tr((\EE_1 - \sem{S}\circ\EE_1)(p_n\rho_n)) + \tr((\EE'_1 - \sem{S'}\circ\EE'_1)(p_n\rho'_n)) \\
    \Longleftrightarrow\ &
    (p_n - p_{n+1}) - \tr(p_n\sigma_n) \le 
    \tr((\EE_1 - \sem{S}\circ\EE_1)(p_n\rho_n)) + \tr((\EE'_1 - \sem{S'}\circ\EE'_1)(p_n\rho'_n))
\end{align*}
since $\tr(p_n\sigma_n) = p_n\tr(\EE_0(\rho_n)) = p_n(\tr(\rho_n) - \tr(\EE_1(\rho_n)) = p_n - \tr(\EE_1(p_n\rho_n))$.

Combine these fact and back to what we aim to prove:
\begin{align*}
    &\tr(\sem{\while[M,S](\rho)}) + \tr(\sem{\while[M',S'](\rho')}) \\
    \le\ & \tr(\rho) + \tr(\rho') - \sum_n((\EE_1 - \sem{S}\circ\EE_1)(p_n\rho_n)) - \sum_n((\EE'_1 - \sem{S'}\circ\EE'_1)(p_n\rho'_n)) - 2\lim_n p_n \\
    \le\ & 2 - \sum_n((p_n - p_{n+1}) - \tr(p_n\sigma_n)) - 2\lim_n p_n \\
    =\ & 2 + \tr(\sigma) - (p_0 - \lim_np_n) - 2\lim_n p_n \\
    =\ & 1 + \tr(\sigma) - \lim_np_n \\
    \le\ &1 + \tr(\sigma).
\end{align*}
Take these all together, $\sigma$ is a partial coupling of $\<\sem{\while[M,S](\rho)}, \sem{\while[M',S'](\rho')}\>$ and this complete the proof.
\end{proof}

\section{Deferred Proofs in ``Applications''}

\begin{theorem}[\Cref{thm:equal rule}]
%\label{thm:equal rule}
    Let $S_1, S_2$ be AST programs acting on the same Hilbert spaces, $\HH_{S_1} = \HH_{S_2} = \HH$. $S_1$ and $S_2$ are semantically equivalent, i.e., $\sem{S_1} = \sem{S_2}$, if and only if,
    \begin{align*}
    %\label{eqn:equal program sym1}
        \vdash (Y_1,Y_2,n)\in\mathcal{Y} : \{nI + P_{sym}^\bot\}
        S_1 \sim S_2 \{ \tr_2(Y_1)\otimes I + I \otimes(nI - \tr_2(Y_2)) \}.
    \end{align*}
where $\mathcal{Y} = \{(Y_1, Y_2\in \Pos(\HH\otimes\HH_2), n\in\mathbb{N}) \mid 0\sqsubseteq Y_1, 0\sqsubseteq 2Y_2\sqsubseteq nI, P_{sym}^\bot[\HH\otimes\HH_2] \ge 2(Y_1\otimes I - I\otimes Y_2)\}$.
%\label{thm: equivalence sym1}
\end{theorem}

\begin{proof}
%[Proof of \Cref{thm:equal rule}]
The if part is relatively easy, while, the only if part requires \Cref{lem: channel equivalence} that conclude from stable QOP \cite{mullerhermes2022monotonicity}.

\noindent(\textbf{if} part).
Suppose Eqn. (\ref{eqn:equal program sym1}) holds. For any $\rho\in\DD(\HH)$, select the input coupling as the witness of $\rho(=_{sym})^\#\rho$ whose existence is ensured by Prop 3.2 in \cite{barthe_rqpd}, i.e., the coupling $\rho_{in} : \<\rho, \rho\>$ such that $\tr(P_{sym}^\bot\rho_{in}) = 0$.
By Eqn. (\ref{eqn:equal program sym1}), we know for any $(Y_1,Y_2,n)\in\mathcal{Y}$, there exists coupling $\sigma : \<\sem{S_1}(\rho), \sem{S_2}(\rho)\>$ such that:
\begin{align*}
    &\tr((\tr_2(Y_1)\otimes I - I \otimes\tr_2(Y_2))\sigma)\\
    \le\ &\tr(P_{sym}^\bot[\HH]\rho) = 0,
\end{align*}
since $\tr(\rho) = \tr(\sigma)$, or equivalently, 
\begin{align*}
    0 \ge &\tr((\tr_2(Y_1)\otimes I - I \otimes\tr_2(Y_2))\sigma)\\
= \ &\tr(2Y_1(\sem{S_1}(\rho)\otimes\halfI)) - \tr(2Y_2(\sem{S_2}(\rho)\otimes\halfI))
\end{align*}
Since $Y_1, Y_2\in \Pos(\HH\otimes\HH_2)$ are arbitrary (since we can always select sufficient large $n\in\mathbb{N}$) such that $P_{sym}^\bot[\HH\otimes\HH_2] \ge 2(Y_1\otimes I - I\otimes Y_2)\}$,
according to Theorem \ref{thm:quantum strassen defect alter}, we have:
$$(\sem{S_1}(\rho)\otimes\halfI) (=_{sym})^\# (\sem{S_2}(\rho)\otimes\halfI),$$
which, again by Prop 3.2 in \cite{barthe_rqpd}, leads to $\sem{S_1}(\rho)\otimes\halfI = \sem{S_2}(\rho)\otimes\halfI$, or equivalently, 
$\sem{S_1}(\rho) = \sem{S_2}(\rho)$. Since $\rho$ is arbitrary, we must have $\sem{S_1} = \sem{S_2}$.

\noindent(\textbf{only if} part).
Since two programs are equivalent, by \Cref{lem: channel equivalence}, we know that for any $\rho_1,\rho_2\in\DD^1$, $T_s(\sem{S_1}(\rho_1),\sem{S_2}(\rho_2))\le T_s(\rho_1,\rho_2)\le T(\rho_1,\rho_2)$. Next, by \Cref{thm:strassen QOT}, we know that for all $Y_1, Y_2\in \Pos(\HH\otimes\HH_2)$ such that $P_{sym}^\bot[\HH\otimes\HH_2] \ge 2(Y_1\otimes I - I\otimes Y_2)$, it holds that:
    $$\tr(\tr_2(Y_1)\sem{S_1}(\rho_1)) \le \tr(\tr_2(Y_2))\sem{S_2}(\rho_2)) + T(\rho_1,\rho_2).$$
If $n\in \mathbb{N}$ such that $0\sqsubseteq 2Y_2\sqsubseteq nI$, then we have:
\begin{align*}
    T_{\tr_2(Y_1)\otimes I + I\otimes (nI - \tr_2(Y_2))}(\sem{S_1}(\rho_1), \sem{S_2}(\rho_2)) &=
    n + \tr(\tr_2(Y_1)\sem{S_1}(\rho_1)) - \tr(\tr_2(Y_2))\sem{S_2}(\rho_2))\\
    &\le n + T(\rho_1,\rho_2) = T_{nI + P_{sym}^\bot}(\rho_1,\rho_2)
\end{align*}
where we use the fact that $S_1,S_2$ are AST programs. The rest is straightforward \Cref{lem:validity quantum operation alter} and \Cref{lem:validity and monotone}.
\end{proof}

\begin{proposition}[Encoding of Trace Distance]
\label{prop:encodingoftracedistance-appendix}
  The following are equivalent for all AST programs $S_1, S_2$ such that\footnote{We could also just ask all programs to be interpreted over $\HH = \HH_{\text{all variables}}$, or over $\HH = \HH_{\var(S_1) \cup \var(S_2)}$.} $\HH_{S_1} = \HH_{S_2}$:
  \begin{enumerate}
    \item $\TD(\sem{S_1}(\rho_1), \sem{S_2}(\rho_2)) \leq \tr(\Phi_1 \rho_1) + \tr(\Phi_2 \rho_2)$ for all $z \in Z$ and $\rho_1 X^{\#} \rho_2$;
    \item $\vDash  0 \leqlow P \leqlow I: \rtriple{X\mid (I+\Phi_1\otimes I + I\otimes \Phi_2)}{S_1}{S_2}{P \otimes I + I \otimes (I-P)}$.
  \end{enumerate}
\end{proposition}

\begin{proof}
%[Proof of \cref{lemma:encodingoftracedistance}]
  Firstly, (2) is equivalent to saying that for all $\rho:\langle\rho_1, \rho_2\rangle, P$ there exists a coupling $\sigma$ such that
  \begin{align*}
    \tr((X\mid (I+\Phi_1\otimes I + I\otimes \Phi_2)) \rho) &\geq \tr((P \otimes I)\sigma) + \tr((I \otimes (I-P))\sigma)\\
    &= \tr(P\sigma_1) + \tr(\sigma) - \tr(P\sigma_2)
  \end{align*}
  where $\sigma_1 = \sem{S_1}(\rho_1)$ and $\sigma_2 = \sem{S_2}(\rho_2)$. Because $S_1, S_2$ are AST, this is in turn equivalent to saying that for all $\rho, z$ and $P$, 
  \[
   \tr((X\mid (I+\Phi_1\otimes I + I\otimes \Phi_2)) \rho) - \tr(\rho)  \geq \tr(P(\rho_1 - \rho_2)),
  \]
  which is equivalent to saying that 
  \[
  \TD(\rho_1, \rho_2) = \max_{0 \leqlow P \leqlow I} \tr(P(\rho_1-\rho_2)) \leq  \tr((X\mid (I+\Phi_1\otimes I + I\otimes \Phi_2)) \rho) - \tr(\rho). 
  \]

  Now, if $\rho_1X^{\#}\rho_2$ does not hold,
  then 
  $\tr(X^{\bot} \rho) > 0$
  for any $\rho$,
  In this case, 
  $\tr((X\mid (I+\Phi_1\otimes I + I\otimes \Phi_2)) \rho) =+\infty$
  and the inequality trivially holds.
  Thus, 
  we only need to consider
  the case when $\rho_1X^{\#}\rho_2$.
  In this case,
  we only need to 
  consider any coupling $\rho$
  with $\tr(\rho X^{\bot}) = 0$
  (by our assumption, 
  such $\rho$ must exist).
  This gives
\[
  \TD(\rho_1, \rho_2) = \max_{0 \leqlow P \leqlow I} \tr(P(\rho_1-\rho_2)) \leq  \tr(\Phi_1 \rho_1) + \tr(\Phi_2 \rho_2)
\]
    as we desired.

\end{proof}

\begin{proposition}[\Cref{prop:encoding-of-wasserstein-distances-lipschitz}]%
%\label{prop:encoding-of-wasserstein-distances-lipschitz}
  Let $\lambda>0$. 
  The following are equivalent for all AST programs $S_1, S_2$ such that 
  $\HH_{S_1} = \HH_{S_2}$:
  \begin{enumerate}
    \item $W(\sem{S_1}(\tr_2(\rho)), \sem{S_2}(\tr_1(\rho))) \leq \lambda \cdot W(\tr_2(\rho), \tr_1(\rho))$ for all $\rho \in \DD(\HH_{S_1} \otimes \HH_{S_2})$;
    \item $ \vDash  \rtriple{\lambda^2 P_{sym}^{\bot}}{S_1}{S_2}{P_{sym}^{\bot}}$.
  \end{enumerate}
\end{proposition}

\begin{proof}
%[Proof of \cref{lemma:encoding-of-wasserstein-distances-lipschitz}]
    By definition, we know that 
    the second condition is equivalent to: 
    for every $\rho \in \DD(\HH_{S_1} \otimes \HH_{S_2})$,
    there is a coupling 
    $\sigma :\langle \sem{S_1}(\tr_2(\rho)), \sem{S_2}(\tr_1(\rho))\rangle$
    such that
    \[
    \lambda^2 \tr (\rho P_{sym}^{\bot})
    \ge \tr (\sigma P_{sym}^{\bot}).
    \]
    Note that 
    $\sigma$ only depends on
    $\tr_2(\rho), \tr_1(\rho)$.
    Thus, 
    for fixed $\rho_1$ and $\rho_2$
    we can write the second 
    condition equivalently as
    \[
    \min_{\rho:\langle\rho_1, \rho_2\rangle} \max_{\sigma:\langle \sem{S_1}(\rho_1), \sem{S_2}(\rho_2)\rangle} \lambda^2  \tr (\rho P_{sym}^{\bot})
    -  \tr (\sigma P_{sym}^{\bot}) \ge 0.
    \]
    This can be simplified to 
    \[
     \lambda^2\min_{\rho:\langle\rho_1, \rho_2\rangle}  \tr (\rho P_{sym}^{\bot})
    \ge \min_{\sigma:\langle \sem{S_1}(\rho_1), \sem{S_2}(\rho_2)\rangle} \tr (\sigma P_{sym}^{\bot}),
    \]
    which is equivalent to the first condition by definition.
\end{proof}

\begin{proposition}[\Cref{prop:encoding-of-non-interference}]
%\label{prop:encoding-of-non-interference}
    For a quantum system $\mathbb{S} = \left\langle\mathcal{H},\rho_0,A,C,do,measure \right\rangle$ with $\rho_0 = \ket{0}\bra{0}$,
    let $G_1, G_2\subseteq A$ be 
    two groups of agents, and $D\subseteq C$
    be a set of commands.
    The following are equivalent:
    \begin{itemize}
        \item $G_1, D : |G_2$.
        \item $\forall \alpha\in (A\times C)^{*}$,  
        \[
        \begin{aligned}
                    &\vDash  a\in G_2, E = \{E_{\lambda}|\lambda\in \Lambda_E\}\in \mathbb{M}_a,  T\subseteq \Lambda_E: \\
        &\rtriple{I}{q:=\ket{0};S_{\alpha}}{q:=\ket{0};S_{\mathsf{purge}_{G_1, D}(\alpha)}}{M}{},
        \end{aligned}
        \]
    \end{itemize}
    where $M= M_T\otimes I+I\otimes (I-M_T)$ with $M_T = \sum_{\lambda\in T} E_{\lambda}$.
\end{proposition}

\begin{proof}
   By the definition of validity, the second condition
   can be equivalently written as the following.
   
   $\forall \alpha\in (A\times C)^{*}, a\in G_2, E = \{E_{\lambda}|\lambda\in \Lambda_E\}\in \mathbb{M}_a$ and  $T\subseteq \Lambda_E$,
   we have
   \[
   \tr (\rho I )\ge \tr(\sigma (M_T\otimes I+I\otimes (I-M_T))),
   \]
   which could be simplified to
   \[
   \tr(M_T(\sigma_1-\sigma_2))\le 0,
   \]
   with $\sigma_1 = \tr_2(\sigma) = \mathcal{E}_{\alpha} (\ket{0}\bra{0})$, and 
   $\sigma_2 = \tr_1(\sigma) =\mathcal{E}_{\mathsf{purge}_{G_1, D}(\alpha)} (\ket{0}\bra{0})$.
    Now, notice that
    $\forall T\subseteq \Lambda_E$
    $\tr(M_T(\sigma_1-\sigma_2))\le 0$
    is equivalent to
    \[
    \max_{T} \left(p_{E, \sigma_1}(T) - p_{E, \sigma_2}(T) \right)\le 0,
    \]
    we can rewrite the second condition as 
    $\forall \alpha\in (A\times C)^{*}, a\in G_2$,
     \[
     d_a \left(\mathcal{E}_{\alpha} (\ket{0}\bra{0}) , \mathcal{E}_{\mathsf{purge}_{G_1, D}(\alpha)} (\ket{0}\bra{0})\right) \le 0.
     \]
     Thus, it is equivalent to 
     $G_1, D : |G_2$ by definition.
\end{proof}

\begin{proposition}[\Cref{prop:encoding-of-differential-privacy}]%
%\label{prop:encoding-of-differential-privacy}
  The following are equivalent for all AST programs $S_1$ on an $n$-qubit system with:
  \begin{enumerate}
    \item $S_1$ is  $(\varepsilon, \delta)$-differentially private;
    \item $ \vDash i\in [n], 0\sqsubseteq M \sqsubseteq I: \rtriple{P_{i,sym}\vert (\exp(\varepsilon)+ \delta) I}{S_1}{S_{1}}{M\otimes I  + \exp(\varepsilon) I\otimes (I-M)}$.
  \end{enumerate}
  Here $P_{i, sym} = P_{sym}[\mathcal{H}_{[n] - i}] \otimes (I_{i\varsinone}\otimes I_{i\varsintwo})$ for $i\in [n]$.
  % Here $P_{i, sym} = P_{sym} \otimes I_{i}$, where $P_{sym}$ acts on the space $\mathcal{H}_{[n] - i}$ for $i\in [n]$.
\end{proposition}

\begin{proof}
  We first notice that,
  1) states that $\forall M, S$ and $\forall \rho,\sigma$,
  if $\exists i\in [n]$ such that
  $\tr_i(\rho) = \tr_i(\sigma)$,
  then 
  \[
  \Pr[\EE(\rho) \in_M S] \leq \exp(\varepsilon) \cdot \Pr(\EE(\sigma) \in_M S) + \delta.
  \]
  This is equivalent to say
  $\forall M, S, \rho, \sigma$ and $\forall i\in [n]$,
  if 
  $\tr_i(\rho) = \tr_i(\sigma)$,
  then the above inequality holds.
  
  Now consider some arbitrary fixed $i\in [n]$,
  and any $\rho$ and $\sigma$ with a coupling
  $\rho_{0}:\langle \rho, \sigma\rangle$.
  If $\tr_{i}(\rho) \ne \tr_{i}(\sigma)$,
  then $\tr(P_{i,sym}^{\bot}\rho_{0}) > 0 $,
  and thus
  $\tr(\rho_0 (P_{i,sym}\vert (\exp (\varepsilon)+\delta)I)) = +\infty$,
  meaning that it is always valid in this case.
  Therefore, we only need to consider the case
  where $\tr_{i}(\rho) = \tr_{i}(\sigma)$.
  In this case,
  we consider
  the case $\supp(\rho_0)\sqsubseteq P_i$
  without loss of generality
  (otherwise the validity condition
  holds directly).
  The condition is
  for any $M$ satisfying $0\sqsubseteq M \sqsubseteq I$,
  \[
    \exp(\varepsilon) + \delta \ge \tr(\sem{S_{1}}(\rho) M) + \exp(\varepsilon)(1- \tr(\sem{S_{1}}(\sigma) M)),
  \]
  which could be simplified to
  \[
        \tr(\sem{S_{1}}(\rho) M) \le  \exp(\varepsilon) \tr(\sem{S_{1}}(\sigma) M) + \delta.
  \]
  Note that when $M$ goes over all $0\sqsubseteq M \sqsubseteq I$,
  it exactly goes over all POVMs $\sum_{m\in S} M_{m}$,
  we know that the second condition is then equivalent to the first
  as we want.
\end{proof}

We need the following proposition
about the projector onto the symmetric subspace.

\begin{proposition}[Adapted from Proposition 3.2 in \cite{Gilles19Relational}]
    $\rho_1 = \rho_2$ if and only if $\rho_1 (=_{sym})^{\#} \rho_2$, where $=_{sym}$ stands for the space of $\supp (P_{sym})$.
\end{proposition}

\begin{proposition}[\Cref{prop:encodingofdiamondnorm}]
  Let $c\in\mathbb{R}^+$. The following are equivalent for all AST programs $S_1, S_2$ such that $\HH = \HH_{S_1} = \HH_{S_2}$:
  \begin{enumerate}
    \item $\|\sem{S_1} - \sem{S_2}\|_\diamond \leq 2c$;
    \item $\vDash  0 \leqlow P \leqlow I_{\HH\otimes\HH}: \rtriple{P_{sym}[\HH\otimes\HH]\mid (1+c)I}{S_1}{S_2}{P \otimes I + I \otimes (I-P)}$.
  \end{enumerate}
\end{proposition}

\begin{proof}
    This is direct by applying 
    \Cref{prop:encodingoftracedistance}
    with $X = P_{sym} [\HH\otimes \HH]$,
    $\Phi_1 = \Phi_2 = cI/2$,
    and the definition of diamond distance
    for completely positive and trace non-increasing
    linear maps.
\end{proof}

\section{Stabilized Quantum Optimal Transport Cost}
\label{sec:stablized-quantum-ot-cost}

In this section, we briefly review the 
proof of
\[
 T_s(\rho, \sigma) = T(\rho\otimes \frac{I}{2}, \sigma \otimes \frac{I}{2})
\]
in~\cite{mullerhermes2022monotonicity} for completeness. 

We begin with a decomposition 
of the projector onto the asymmetric 
subspace.

\begin{lemma}[Identity (1) in \cite{mullerhermes2022monotonicity}]%
\label{lemma:asym-projector-tensor-decompose}
    We have
    \[
    \mathrm{P}_{\mathrm{asym}}(d_1\otimes d_2)=\mathrm{P}_{\mathrm{asym}}(d_1)\otimes\mathrm{P}_{\mathrm{sym}}(d_2)+\mathrm{P}_{\mathrm{sym}}(d_1)\otimes\mathrm{P}_{\mathrm{asym}}(d_2).
    \]
\end{lemma}

The following lemma is a special case of 
the Schur-Weyl duality.
Suppose $X\in \mathcal{D}(\mathbb{C}^d\otimes \mathbb{C}^d)$,
and let $\Sigma_d$ be the following $UU$-twirling channel
\[
    \Sigma_d (X) = \int_{\mathcal{U}_d} (U\otimes U) X (U^{\dagger}\otimes U^{\dagger}) \text{d}U,
\]
where the integral is with respect to the Haar measure on group $\mathcal{U}_d$ of $d\times d$ unitary matrices. Then, we have:

\begin{lemma}[Theorem 10 in~\cite{Mel24}]\label{thm:schur-weyl-2-dim}
    \[
    \Sigma_d(X)=\mathrm{Tr}[XP_{\mathrm{sym}}(d)]\frac{P_{\mathrm{sym}}(d)}{\mathrm{Tr}[P_{\mathrm{sym}}(d)]}+\mathrm{Tr}[XP_{\mathrm{asym}}(d)]\frac{P_{\mathrm{asym}}(d)}{\mathrm{Tr}[P_{\mathrm{asym}}(d)]},
    \]
    where $P_{\mathrm{sym}}(d) = (I-F_d)/2$ is the projector onto the symmetric subspace, and
    $P_{\mathrm{asym}}(d) = I - P_{\mathrm{sym}}(d)$.
\end{lemma}

\begin{proposition}\label{fact:sigma-d-self-dual}
        The channel $\Sigma_d$ is self-dual with respect 
    to the Hilbert-Schmidt inner product.
\end{proposition}

\begin{proof}
    For any $A$ and $B$, we have
    \[
    \begin{aligned}
        \tr \sbra{A^{\dagger}\Sigma_d(B)} &= \tr \sbra{A^{\dagger} \int_{\mathcal{U}_d} (U\otimes U) B (U^{\dagger}\otimes U^{\dagger}) \text{d}U} \\
        &= \int_{\mathcal{U}_d} \tr \sbra{A^{\dagger}  (U\otimes U) B (U^{\dagger}\otimes U^{\dagger}) } \text{d}U \\
        &= \int_{\mathcal{U}_d} \tr \sbra{(U^{\dagger}\otimes U^{\dagger}) A^{\dagger}  (U\otimes U) B } \text{d}U \\
        &= \int_{\mathcal{U}_d} \tr \sbra{(U\otimes U) A^{\dagger}  (U^{\dagger}\otimes U^{\dagger}) B } \text{d}U \\
        &= \tr \sbra{\Sigma_d(A^{\dagger})B}.
    \end{aligned}
    \]
    % Here we uses a property of Haar measure (Proposition 3 in~\cite{Mel24}) that allows 
    % free to change the variable $U$ with $U^{\dagger}$.
\end{proof}

\begin{proposition}\label{prop:fixed-point-asym}
    Using the above notations, we have
    \[
        (id_{d_1}\otimes id_{d_1} \otimes \Sigma_{d_2})
        (\mathrm{P}_{\mathrm{asym}}(d_1\otimes d_2))
        = \mathrm{P}_{\mathrm{asym}}(d_1\otimes d_2).
    \]
\end{proposition}

\begin{proof}
    This can be verified directly by using \cref{lemma:asym-projector-tensor-decompose}
    and \cref{thm:schur-weyl-2-dim}.
\end{proof}

\begin{theorem}[Theorem 3.3 in~\cite{mullerhermes2022monotonicity}]
\label{thm:min-of-tensor-ot-cost}
    For $d_1$-dimensional quantum density matrices $\rho_1$
    and $\sigma_1$, and $d_2$-dimensional quantum matrices
    $\rho_2$ and $\sigma_2$,
    we have
    \[
        T(\rho_1\otimes\rho_2,\sigma_1\otimes\sigma_2)\geq T\left(\rho_1\otimes\frac{I_2}{2},\sigma_1\otimes\frac{I_2}{2}\right).
    \]
\end{theorem}

\begin{proof}
    Let $\tau_{A_1A_2B_1B_2}$ be an optimal coupling state
    that gives the value $T(\rho_1 \otimes \rho_2, \sigma_1\otimes \sigma_2)$.
    Equivalently speaking, we have
    \[
        T\left(\rho_{1}\otimes\rho_{2},\sigma_{1}\otimes\sigma_{2}\right)=\mathrm{Tr}\left[\tau_{A_{1}B_{1}A_{2}B_{2}}\operatorname{P}_{\mathrm{asym}}(d_{1}\otimes d_{2})\right],
    \]
    with $\tau_{A_1A_2} = \rho_1\otimes \rho_2$,
    and $\tau_{B_1B_2} = \sigma_1\otimes \sigma_2$.
    By \cref{prop:fixed-point-asym} and \cref{fact:sigma-d-self-dual}, we have
    \[
    T\left(\rho_{1}\otimes\rho_{2},\sigma_{1}\otimes\sigma_{2}\right)=\mathrm{Tr}\left[\tau_{A_{1}B_{1}A_{2}B_{2}}\Phi(\operatorname{P}_{\mathrm{asym}}(d_{1}\otimes d_{2}))\right] = \mathrm{Tr}\left[\Phi(\tau_{A_{1}B_{1}A_{2}B_{2}})\operatorname{P}_{\mathrm{asym}}(d_{1}\otimes d_{2})\right]
    \]
    where $\Phi$ is the $id_{d_1}\otimes id_{d_1} \otimes \Sigma_{d_2}$ channel.
    Using \cref{lemma:asym-projector-tensor-decompose},
    we have
    \[
    \begin{aligned}
        \Phi(\tau_{A_{1}B_{1}A_{2}B_{2}})
    &= \tr_{A_2 B_2} [\tau_{A_{1}B_{1}A_{2}B_{2}} (I\otimes \operatorname{P}_{\mathrm{sym}}(d_{2}))] \otimes \frac{P_{\mathrm{sym}}(d_2)}{\mathrm{Tr}[P_{\mathrm{sym}}(d_2)]} \\ 
    &+ 
    \tr_{A_2 B_2} [\tau_{A_{1}B_{1}A_{2}B_{2}} (I\otimes \operatorname{P}_{\mathrm{asym}}(d_{2}))] \otimes \frac{P_{\mathrm{asym}}(d_2)}{\mathrm{Tr}[P_{\mathrm{asym}}(d_2)]}.
    \end{aligned}
    \]
    Therefore, we could write 
    \[
    X_{A_1B_1} =  \tr_{A_2 B_2} [\tau_{A_{1}B_{1}A_{2}B_{2}} (I\otimes \operatorname{P}_{\mathrm{asym}}(d_{2}))]
    \]
    and 
    \[
    Y_{A_1B_1} =  \tr_{A_2 B_2} [\tau_{A_{1}B_{1}A_{2}B_{2}} (I\otimes \operatorname{P}_{\mathrm{sym}}(d_{2}))],
    \]
    with
    \[
    X_{A_1B_1} + Y_{A_1B_1} = \tr_{A_2 B_2} [\tau_{A_{1}B_{1}A_{2}B_{2}} ].
    \]
    Then, we have
    \[
    \begin{aligned}
        \mathrm{Tr}\left[\Phi(\tau_{A_{1}B_{1}A_{2}B_{2}})\operatorname{P}_{\mathrm{asym}}(d_{1}\otimes d_{2})\right]
    &= \mathrm{Tr}\left[ \rbra*{X_{A_1B_1}  \otimes \frac{P_{\mathrm{sym}}(d_2)}{\mathrm{Tr}[P_{\mathrm{sym}}(d_2)]} }\operatorname{P}_{\mathrm{asym}}(d_{1}\otimes d_{2})\right] \\
    &+ \mathrm{Tr}\left[ \rbra*{Y_{A_1B_1}  \otimes \frac{P_{\mathrm{asym}}(d_2)}{\mathrm{Tr}[P_{\mathrm{asym}}(d_2)]} }\operatorname{P}_{\mathrm{asym}}(d_{1}\otimes d_{2})\right] \\
    &= \mathrm{Tr}\left[ X_{A_1B_1}  \operatorname{P}_{\mathrm{asym}}(d_{1})\right]
    +  \mathrm{Tr}\left[ Y_{A_1B_1}  \operatorname{P}_{\mathrm{sym}}(d_{1})\right].
    \\
    \end{aligned}
    \]
    The last step is by using \cref{lemma:asym-projector-tensor-decompose} and the orthogonality of 
    $\operatorname{P}_{\mathrm{asym}}$ and
    $\operatorname{P}_{\mathrm{sym}}$.

    Now, define the state
     \[
        \widetilde{\tau}_{A_{1}B_{1}A_{2}B_{2}}
    = X_{A_1B_1}\otimes \frac{P_{\mathrm{sym}}(2)}{\mathrm{Tr}[P_{\mathrm{sym}}(d_2)]} 
    + 
    Y_{A_1B_1}\otimes \frac{P_{\mathrm{asym}}(2)}{\mathrm{Tr}[P_{\mathrm{asym}}(2)]}.
    \]
    Note that $\widetilde{\tau}_{A_{1}B_{1}A_{2}B_{2}}$
    is a density operator of dimension $d_1\times d_2 \times 2\times 2$. We claim it is a coupling state
    of $\rho_1 \otimes I/2$ and $\sigma_1 \otimes I/2$,
    this is because
    \[
    \tau_{A_1A_2} = X_{A_1}\otimes I/2 + Y_{A_1}\otimes I/2 = \rho_1 \otimes I/2,
    \]
    and 
    \[
    \tau_{B_1B_2} = X_{B_1}\otimes I/2 + Y_{B_1}\otimes I/2 = \sigma_1 \otimes I/2.
    \]
    From a similar argument as above, 
    we know
    \[
    \mathrm{Tr}\left[\widetilde{\tau}_{A_{1}B_{1}A_{2}B_{2}}\operatorname{P}_{\mathrm{asym}}(d_{1}\otimes 2)\right] = \mathrm{Tr}\left[ X_{A_1B_1}  \operatorname{P}_{\mathrm{asym}}(d_{1})\right]
    +  \mathrm{Tr}\left[ Y_{A_1B_1}  \operatorname{P}_{\mathrm{sym}}(d_{1})\right]
    \]
    Therefore,
    we know 
    \[
    T\left(\rho_{1}\otimes\rho_{2},\sigma_{1}\otimes\sigma_{2}\right) = \mathrm{Tr}\left[\widetilde{\tau}_{A_{1}B_{1}A_{2}B_{2}}\operatorname{P}_{\mathrm{asym}}(d_{1}\otimes 2)\right] \ge T\left(\rho_{1}\otimes I/2 ,\sigma_{1}\otimes I/2 \right),
    \]
    as $T\left(\rho_{1}\otimes I/2 ,\sigma_{1}\otimes I/2 \right)$ is the minimum value
    of $\mathrm{Tr}\left[\tau'_{A_{1}B_{1}A_{2}B_{2}}\operatorname{P}_{\mathrm{asym}}(d_{1}\otimes 2)\right] $
    over all coupling states $\tau'_{A_{1}B_{1}A_{2}B_{2}}$.
\end{proof}

From \cref{thm:min-of-tensor-ot-cost},
we could immediately get
$T_s(\rho, \sigma) = T(\rho\otimes \frac{I_2}{2}, \sigma \otimes \frac{I_2}{2})$.
This is because
\[
    T(\rho\otimes \frac{I_2}{2}, \sigma \otimes \frac{I_2}{2}) \ge \inf_{\gamma} T(\rho\otimes \gamma, \sigma\otimes \gamma) 
    \ge  T(\rho\otimes \frac{I_2}{2}, \sigma \otimes \frac{I_2}{2}). 
\]
The first inequality is by the definition of $\inf$,
and the second is by applying \cref{thm:min-of-tensor-ot-cost}.

\section{Kantorovich-Rubinstein Duality}

We first review the Kantorovich duality theorem in the theory
of optimal transport.

\begin{theorem}[Kantorovich Duality, Theorem 5.10 in~\cite{Villani2008Optimal}]
\label{thm:Kantorovich-duality}
    Suppose$(\mathcal{X}, \mu)$ and $(\mathcal{Y}, \nu)$
    are two Polish spaces.
    Let $c:\mathcal{X}\times \mathcal{Y}\to \mathbb{R}^{+\infty}$
    be a lower semi-continuous function such that there
    exists some real-valued upper semi-continuous functions
    $a\in L^1(\mu)$ and $b\in L^1(\nu)$, with
    \[
    c(x,y) \ge a(x) + b(y), \forall (x,y)\in \mathcal{X}\times \mathcal{Y}.
    \]
    Then, we have
    \[
    \min_{\pi\in \Gamma(\mu, \nu)} 
    \int_{\mathcal{X} \times \mathcal{Y}} c(x,y) \textup{d}\pi(x,y)
    =
    \sup_{
    \substack{\psi\in\mathcal{C}_b(\mathcal{X}),\\
    \phi\in\mathcal{C}_b(\mathcal{X});\\
    \psi-\phi \le c}}
    \rbra{\int_{\mathcal{Y}} \phi(y) \textup{d}\nu(y)
    - \int_{\mathcal{X}} \psi(x) \textup{d}\mu(x)},
    \]
    where $\mathcal{C}_b$ denotes the set of continuous bounded 
    functions, and $L^1$ is the Lesbegue space of exponent $1$.
\end{theorem}

Here, 
a Polish space refers to a topological space that is separable and completely metrizable.
With the above theorem, we have:

\begin{theorem}[Kantorovich-Rubinstein Duality]
\label{prop:Kantorovich-Rubinstein-duality-appendix}
Let $\mu,\nu$ be discrete probability distributions over $X$ and $Y$
respectively, and let $c:X\times Y \rightarrow [0,+\infty)$ be a
  bounded function. Then
$$\inf_{\theta \in \Gamma (\mu,\nu)} \E_\theta [c] =
  \sup_{(n,c_1,c_2)\in\mathcal{W}} (\E_\mu[c_1]+\E_\nu[c_2] -n)
  $$ where $\Gamma (\mu,\nu)$ denotes the set of probabilistic
  couplings of $\mu$ and $\nu$ and $(n, c_1,c_2) \in \mathcal{W}$ iff
  for every $x\in X$ and $y\in Y$, we have $0\leq c_1(x),c_2(y)$
  and $c_1(x)+c_2(y)\leq c(x,y)+n$.
\end{theorem}

% \begin{proposition}[Kantorovich-Rubinstein Duality]
% \label{prop:Kantorovich-Rubinstein-duality-appendix}
% Let $\mu,\nu$ be probability distributions over countable discrete spaces $X$ and $Y$
% respectively, and let $c:X\times Y \rightarrow [0,+\infty)$ be a
%   bounded function. 
%   Then
% \[
% \inf_{\theta \in \Gamma (\mu,\nu)} \E_\theta [c] =
%   \sup_{\{c_1,c_2\mid c_1+c_2\leq c \}} \E_\mu[c_1]+\E_\nu[c_2]
% \] 
%   where $c_1$ and $c_2$ range over bounded real-valued functions
%   from $X$ and $Y$ respectively.
% \end{proposition}
\begin{proof}
    We first note that a countable set with the discrete topology is always a Polish space,
    and bounded functions on the discrete space
    is always continuous.
    Now, 
    by applying \cref{thm:Kantorovich-duality}
    with $\psi = -b_1$ and $\phi = b_2$,
    we have
    \[
    \inf_{\theta \in \Gamma (\mu,\nu)} \E_\theta [c] =
  \sup_{(b_1,b_2)\in\mathcal{B}} (\E_\mu[b_1]+\E_\nu[b_2] ),
    \]
    where 
    $(b_1, b_2)\in \mathcal{B}$
    iff $b_1$ and $b_2$
    are bounded functions,
    and for every $x\in X$
    and $y\in Y$,
    we have 
    $b_1(x) + b_2(y)\le c(x,y)$.
   
    We now show 
    \[
     \sup_{(b_1,b_2)\in\mathcal{B}} (\E_\mu[b_1]+\E_\nu[b_2] )
     =
     \sup_{(n,c_1,c_2)\in\mathcal{W}} (\E_\mu[c_1]+\E_\nu[c_2] -n).
    \]
    Since 
    $b_1$ and $b_2$ are bounded,
    there exists some 
    integer
    $m$
    satisfying
    $\abs{b_1(x)}\le m$
    and 
    $\abs{b_2(y)}\le m$.
    We can then 
    write 
    $c_1(x) = b_1(x)+m$
    and 
    $c_2(y) = b_2(y)+m$
    with $n = 2m$,
    getting
    $0\le c_1(x)$
    and 
    $0\le c_2(y)$.
    This gives
    \[
     \sup_{(b_1,b_2)\in\mathcal{B}} (\E_\mu[b_1]+\E_\nu[b_2] )
     \le
     \sup_{(n,c_1,c_2)\in\mathcal{W}} (\E_\mu[c_1]+\E_\nu[c_2] -n).
    \]
    
    For the other direction,
    observe that for fixed $n$, 
    $c$ is bounded and
    $c_1(x) + c_2(y) \le c(x,y)+n$,
    we can
    know
    $c_1$ and $c_2$
    is bounded.
    Therefore,
    write 
    $b_1 = c_1$
    and 
    $b_2 = c_2-n$
    gives the desired result.

\end{proof}

\section{Postponed Technical Proofs}
\label{sec:postponed-technical-proofs}

In this section, we give proofs of the lemmas
that are omitted in the previous part of the
appendix.

\begin{lemma}[\Cref{{lemma:basic-properties-of-ops-of-ivp}}]
We have the following properties for $A, A_1, A_2\in \PosI(\HH)$:
\begin{itemize}
  \item Scalar product $cA$ for $c\in \mathbb{R}^{+\infty}$ is defined such that for all $\ket{\psi}$, $\bra{\psi}cA\ket{\psi}= c\bra{\psi}A\ket{\psi}$.
  \item Addition $A_1+A_2$ such that for all $|\psi\>\in \HH$, $\<\psi|(A_1+A_2)|\psi\> = \<\psi|A_1|\psi\>+\<\psi|A_2|\psi\>$.
  \item Tensor product $A_1\otimes A_2$ such that for all $|\psi_1\>, |\psi_2\>$, $(\<\psi_1|\otimes\<\psi_2|)(A_1\otimes A_2)(|\psi_1\>\otimes|\psi_2\>) = (\<\psi_1|A_1|\psi_1\>)\cdot(\<\psi_2|A_2|\psi_2\>)$.
  \item Let $M$ be a linear operator with $\HH$ as its domain, $M^\dagger AM$ can be defined such that for all $|\psi\>$, $\<\psi|(M^\dagger AM)|\psi\> = \<\phi|A|\phi\>$ where $|\phi\> = M|\psi\>$.
  \item For $P\in \Pos$ 
  with decomposition $P = \sum_ia_i|\psi_i\>\<\psi_i|$ ($0\le a_i$), 
  the trace is $\tr(AP) = \sum_i a_i \<\psi_i| A |\psi_i\>$. 
  Note that the value is unique for any decomposition.
        %\lz{Minbo: help prove it.}
  \item For $\EE\in \QO$ (more generally, CP maps) with Kraus operators $\{E_i\}$, $\EE^\dagger(A) = \sum_i E_i^\dagger A E_i$. Note that it is unique for arbitrary Kraus operators.
        %\lz{Minbo: help prove it.}
  \item $A_1 = A_2$ if for all $|\psi\>$, $\bra{\psi}A_1\ket{\psi} = \bra{\psi}A_2\ket{\psi}$.
  %\item $A_{1}\otimes A_{2} = B_{1}\otimes B_{2}$ if for all $\ket{\psi_{1}}\otimes \ket{\psi_{2}}$,
  \item $A_1\sqsubseteq A_2$ if for all $|\psi\>$, $\<\psi|A_1|\psi\> \le \<\psi|A_2|\psi\>$.
\end{itemize}
\end{lemma}

\begin{proof}
  In the following, let $A = \sum_{i} \lambda_{i} X_{i}$, where $X_{i}$ is the
  projection onto the corresponding eigenspace.
  We will also write $A$ as $A = P_{A} + \infty X_{A}$, where
  $P_{A} = \sum_{\lambda_{i}< +\infty} \lambda_{i} X_{i}$ is the ``finite''
  component of $A$, and $X_{A} = \sum_{\lambda_{i} = +\infty} X_{i}$ is the
  ``infinite'' space of $A$.
  Similarly, we write
  $A_{1} = \sum_{i} \lambda^{(1)}_{i} X_{i}^{(1)} = P_{A_{1}} + \infty\cdot X_{A_{1}}$
  and
  $A_{2} = \sum_{i} \lambda_{i}^{(2)} X_{i}^{(2)} = P_{A_{2}} + \infty\cdot X_{A_{2}}$.
  \begin{itemize}
    \item For the scalar product $cA$, we define it as
          $cA = \sum_{i} c\lambda_{i} X_{i}$.
          By the definition of inner product, it is direct that
          $\bra{\psi}cA\ket{\psi}= c\bra{\psi}A\ket{\psi}$.
    \item For the addition operation, we define it as
          $A_{1}+ A_{2} = X^{\bot}(P_{A_{1}}+P_{A_{2}})X^{\bot} + \infty\cdot X$,
          where $X = X_{A_{1}}\vee X_{A_{2}}$.
          We now verify that
          $\bra{\psi} A_{1}+A_{2}\ket{\psi} = \bra{\psi} A_{1}\ket{\psi} + \bra{\psi} A_{2}\ket{\psi}$.
          If $\bra{\psi} A_{1}+A_{2}\ket{\psi} < +\infty$, then we know
          $\ket{\psi}\in X^{\bot}$, meaning that
          $X^{\bot}\ket{\psi} = \ket{\psi}$.
          Then, by the definition of $X$, we know
          $\bra{\psi} A_{1}\ket{\psi} < +\infty$ and
          $\bra{\psi} A_{2}\ket{\psi} < +\infty$.
          Thus, in this case the condition holds.
          Now, consider the case
          when $\bra{\psi} A_{1}+A_{2}\ket{\psi} = +\infty$.
          Then, we know
          $\bra{\psi} X\ket{\psi} > 0 $.
          We claim either
          $\bra{\psi} X_{1}\ket{\psi} > 0 $
          or
          $\bra{\psi} X_{2}\ket{\psi} > 0 $,
          because otherwise
          $\ket{\psi} \in X_{1}^{\bot} \cap X_{2}^{\bot}$,
          contradicting
          $\bra{\psi} X\ket{\psi} > 0 $.
    \item For the tensor product,
          we define it as
          $P_{A_1} \otimes P_{A_2}+\infty\cdot X $,
          where $X = \left( \supp \left( P_{A_{1}} \right)\otimes X_{A_{2}}\right)\vee \left( X_{A_{1}} \otimes \supp \left( P_{A_{2}} \right) \right)\vee \left(X_{A_{1}}\otimes X_{A_{2}}\right)$.
          We now verify that it satisfies the property.
          For any $\ket{\psi_{1}}$ and $\ket{\psi_{2}}$,
          first suppose that
          $\bra{\psi_{1}}\bra{\psi_{2}}A_{1}\otimes A_{2}\ket{\psi_{1}}\ket{\psi_{2}} =0$.
          We know that
          $X \ket{\psi_{1}}\ket{\psi_{2}} = 0$
          and $\bra{\psi_{1}}\bra{\psi_{2}}P_{A_{1}}\otimes P_{A_{2}}\ket{\psi_{1}}\ket{\psi_{2}} = 0$.
          Therefore,
          without loss of generality
          we can assume
          $\bra{\psi_{1}}P_{A_{1}}\ket{\psi_{1}} = 0$.
          From
          $X \ket{\psi_{1}}\ket{\psi_{2}} = 0$,
          we know
          $X_{A_{1}} \ket{\psi_{1}} = 0$,
          meaning that
          $\bra{\psi_{1}}A_{1}\ket{\psi_{1}} = 0$
          as we want.

          Now, consider the case
          $0< \bra{\psi_{1}}\bra{\psi_{2}}A_{1}\otimes A_{2}\ket{\psi_{1}}\ket{\psi_{2}} < +\infty$.
          By definition, we know
          $X \ket{\psi_{1}}\ket{\psi_{2}} = 0$,
          $\bra{\psi_{1}}P_{A_{1}}\ket{\psi_{1}} \neq 0$,
          $\bra{\psi_{2}}P_{A_{2}}\ket{\psi_{2}} \neq 0$,
          We first claim
          $\bra{\psi_{1}}A_{1}\ket{\psi_{1}} \neq 0$,
          and
          $\bra{\psi_{2}}A_{2}\ket{\psi_{2}} \neq 0$.
          If not, without loss of generality,
          we assume
          $\bra{\psi_{1}}A_{1}\ket{\psi_{1}} = 0$.
          It gives
          $\ket{\psi_{1}}\in \left( \supp(P_{A_{1}}) \vee X_{A_{1}} \right)$,
          meaning that $X\ket{\psi_{1}}\ket{\psi_{2}} = 0$.
          We then know $\bra{\psi_{1}}\bra{\psi_{2}}A_{1}\otimes A_{2}\ket{\psi_{1}}\ket{\psi_{2}} = 0 $,
          a contradiction.
          We then claim
          $\bra{\psi_{1}}A_{1}\ket{\psi_{1}} < +\infty$,
          and
          $\bra{\psi_{2}}A_{2}\ket{\psi_{2}} < +\infty$.
          Suppose for simplicity that
          $\bra{\psi_{1}}A_{1}\ket{\psi_{1}} = +\infty$,
          we know
          $X_{A_{1}}\ket{\psi_{1}} \ne 0$.
          Combined with
          $\bra{\psi_{2}}A_{2}\ket{\psi_{2}} \neq 0$.
          we conclude that
          projecting $\ket{\psi_{1}}\ket{\psi_{2}}$
          onto the space
          $X_{A_{1}}\otimes \supp \left( P_{A_{2}} \right) \vee X_{A_{1}}\otimes X_{A_{2}}$
          is non-zero,
          meaning
          $\bra{\psi_{1}}\bra{\psi_{2}}A_{1}\otimes A_{2}\ket{\psi_{1}}\ket{\psi_{2}} = +\infty$,
          a contradiction.
          Since both $\bra{\psi_{1}}A_{1}\ket{\psi_{1}}$ and $\bra{\psi_{2}}A_{2}\ket{\psi_{2}}$
          are non-zero and finite,
          a direct computation will give the equation we want.

          For the case when
          $\bra{\psi_{1}}\bra{\psi_{2}}A_{1}\otimes A_{2}\ket{\psi_{1}}\ket{\psi_{2}} = +\infty$,
          we know
          $X\ket{\psi_{1}}\ket{\psi_{2}}\ne 0$.
          Suppose
          $\bra{\psi_{1}}A_{1}\ket{\psi_{1}} < +\infty$,
          and
          $\bra{\psi_{2}}A_{2}\ket{\psi_{2}} < +\infty$.
          we could conclude
          $X_{A_{1}}\ket{\psi_{1}} = X_{A_{2}}\ket{\psi_{2}}=  0$,
          giving $X\ket{\psi_{1}}\ket{\psi_{2}}=  0$,
          a contradiction.
          Thus, without loss of generality,
          we can assume
          $\bra{\psi_{1}}A_{1}\ket{\psi_{1}} = +\infty$.
          We then claim
          $\bra{\psi_{2}}A_{2}\ket{\psi_{2}} \ne 0$.
          Otherwise we will get
          $X\ket{\psi_{1}}\ket{\psi_{2}}=  0$.
          We then conclude
          $\bra{\psi_{1}}A_{1}\ket{\psi_{1}} \bra{\psi_{2}}A_{2}\ket{\psi_{2}}= +\infty$
          as we want.

    \item $MAM^{\dagger}$ can be defined as $MP_{A} M^{\dagger} + \infty\cdot \supp(MXM^{\dagger})$.
          For any $\ket{\phi} = M\ket{\psi}$,
          if $\bra{\phi}A\ket{\phi} < +\infty$, we know that $X\ket{\phi} = 0$,
          meaning that
          $\bra{\psi}MXM^{\dagger}\ket{\psi} = 0$
          Thus,
          $\bra{\psi} MAM^{\dagger}\ket{\psi} = \bra{\psi}MP_{A}M^{\dagger} \ket{\psi}< +\infty$,
          and by definition
          $\bra{\psi} MAM^{\dagger}\ket{\psi} = \bra{\psi}MP_{A}M^{\dagger} \ket{\psi} = \bra{\phi}A\ket{\phi} $ as we want.
          If $\bra{\phi}A\ket{\phi} = +\infty$,
          we know that
          $X\ket{\phi} \ne 0$,
          or equivalently,
          $XM^{\dagger}\ket{\psi} \ne 0$,
          meaning $\bra{\psi} MAM^{\dagger}\ket{\psi} = +\infty$ as we want.

    \item For $P\in \Pos$,
          define $\tr(AP) = \tr(P_{A} P)$ if $X_{A} P = 0$, and $+\infty$ otherwise.
          If $P = \sum_{i} a_{i} \ket{\psi_{i}}\bra{\psi_{i}}$ with $a_{i}\ge 0$.
          We first notice that
          it suffices to consider the set of $j$ with $a_{j} >0$
          as $0\cdot +\infty = 0$.
          We then notice that
          $\supp (P) =\text{span}_{i} \left\{ \ket{\psi_{i}} \right\}$.
          Thus,
          if $\tr(AP) < +\infty$,
          then $X_{A}\ket{\psi_{i}} = 0$ for any $i$,
          with $\tr(AP) = \tr(P_{A}P) = \sum_{i} a_{i}\bra{\psi_{i}} P_{A} \ket{\psi_{i}}$ as we want.
          If $\tr(AP) = +\infty$,
          then $X_{A}\cap \spanv_{i} \left\{ \ket{\psi_{i}} \right\} \ne 0$,
          meaning that
          there must be some $i$ with $\bra{\psi_{i}}A \ket{\psi_{i}} = +\infty$.
          Then, we know $\sum_{i} a_{i} \bra{\psi_{i}}A \ket{\psi_{i}} = +\infty $ as we want.
    \item For a CP map $\mathcal{E}^{\dagger}$ and $A = P_{A}+ \infty X_{A}$,
          define $\mathcal{E}^{\dagger}(A) = X^{\bot}\mathcal{E}(P_{A})X^{\bot}+ \infty X$,
          where $X = \supp(\mathcal{E}(X_{A}))$.
          For $\mathcal{E}$ with Kraus operators $E_{i}$,
          we know $X$ can be written as $X = \supp(\sum_{i} E_{i}^{\dagger} X_{A} E_{i})  = \vee_{i} \supp (E_{i}^{\dagger} X_{A} E_{i})$.
          By definition, we know in this case
          $\mathcal{E}^{\dagger}(A)$ can be written as $\sum_{i} E_{i}^{\dagger}A E_{i}$ as we desired.

    \item We first prove that $X_{A_{1}} =  X_{A_{2}}$.
          If not, let $\ket{\psi}$ be a normalized state in $X_{A_{1}}\cap X_{A_{2}}^{\bot}$.
          We know $X_{A_{2}}\ket{\psi} = 0$ but $X_{A_{1}}\ket{\psi} = \ket{\psi}$.
          This means $\bra{\psi}A_{1}\ket{\psi} = +\infty$ but  $\bra{\psi}A_{2}\ket{\psi} < +\infty$,
          a contradiction.
          Now, given that $X_{A_{1}} =  X_{A_{2}}$,
          we know that for any $\ket{\psi} \in  X_{A_{1}}^{\bot} \supseteq (\supp(P_{A_{1}})\vee \supp(P_{A_{2}}))$,
          $\bra{\psi}P_{A_{1}}\ket{\psi} = \bra{\psi}P_{A_{2}}\ket{\psi}$,
          meaning that $P_{A_{1}} = P_{A_{2}}$ as we desired.
    \item We extend the L\"{o}wner order of positive semi-definite operators to infinite-valued positive semi-definite operators as follows:
          for $A, B\in \PosI(\mathcal{H})$,
          we say $A\sqsubseteq B$ if for any $\ket{\psi}$, $\bra{\psi} A\ket{\psi} \le \bra{\psi} B\ket{\psi}$.
          By the definition of trace,
          it is clear that $A\sqsubseteq B$ if for any $\rho$, $\tr(A\rho)\le \tr(B\rho)$.
          It is clear that it satisfies reflexivity and transitivity.
          For antisymmetry, it follows from the previous property.
          Thus it is a partial order.
  \end{itemize}
\end{proof}

\begin{lemma}[\Cref{lem: IVP algebraic}]%
  In the following, 
  let $a,b,c\in \mathbb{R}^{+\infty}$, $A,A_1,A_2 \in\PosI$, $M,M_1,M_2,\cdots\in \LL$, and $P,P_1,P_2\cdots\in\Pos$.
  We have the following properties:
  \begin{itemize}
    \item $0A = 0$, $1A = A$, $a(bA) = (ab)A$;
    \item $0+A = A+0 = A$, $A_1 + A_2 = A_2 + A_1$, $A_1 + (A_2 + A_3) = (A_1 + A_2) + A_3$;
    \item $0\otimes A = A\otimes 0 = 0$, $A_1\otimes (A_2 \otimes A_3) = (A_1 \otimes A_2) \otimes A_3$;
    \item $A\otimes (cA_1 + A_2) = c(A\otimes A_1) + (A\otimes A_2)$;
          $(cA_1 + A_2)\otimes A = c(A_1\otimes A) + (A_2\otimes A)$;
    \item $0^\dagger A0 = 0$, $M_2^\dagger (M_1^\dagger AM_1) M_2 = (M_1M_2)^\dagger A (M_1M_2)$; 
          $M^\dagger(cA_1+A_2)M = c(M^\dagger A_1M) + M^\dagger A_2M$;
    \item $(M_1\otimes M_2)^\dagger (A_1\otimes A_2) (M_1\otimes M_2) = 
           (M_1^\dagger A_1M_1)\otimes (M_2^\dagger A_2M_2)$;
    \item $\tr(A(cP_1+P_2)) = c\tr(AP_1) + \tr(AP_2)$;
          $\tr((cA_1 + A_2)P) = c\tr(A_1P) + \tr(A_2P)$;
    \item $\tr((A_1\otimes A_2)(P_1\otimes P_2)) = \tr(A_1P_1)\tr(A_2P_2)$;
          $\tr((M^\dagger AM)P) = \tr(A (MPM^\dagger))$.
    \item $\tr((A\otimes I)P) = \tr(A\tr_2(P))$;
          $\tr((I\otimes A)P) = \tr(A\tr_1(P))$;
    \item $\tr(A\ket{\phi}\bra{\phi}) = \bra{\phi}A\ket{\phi}$.
    \item $A_1=A_2$ iff for all $P\in\Pos$ (or $P\in\DD$) such that $\tr(A_1 P) = \tr(A_2 P)$;
    \item $A_1\sqsubseteq A_2$ iff for all $P\in\Pos$ (or $P\in\DD$) such that $\tr(A_1 P) \le \tr(A_2 P)$;
    \item $A_1\sqsubseteq A_2$ implies $M^\dagger A_1M\sqsubseteq M^\dagger A_2M$; $A_1\sqsubseteq A_2$ and $A_3\sqsubseteq A_4$ implies $cA_1+A_3\sqsubseteq cA_2+A_4$.
  \end{itemize}
  As direct corollaries, for CP map $\EE, \EE_1,\EE_2$,
  \begin{itemize}
    \item $\tr(A\EE(P)) = \tr(\EE^\dagger(A) P)$;
          $A_1\sqsubseteq A_2$ implies $\EE (A_1)\sqsubseteq \EE (A_2)$;
    \item $(c\EE_1 + \EE_2)(A) = c\EE_1(A) + \EE_2(A)$;
          $\EE(cA_1+A_2) = c\EE(A_1) + \EE(A_2)$;
    \item $\EE_2(\EE_1(A)) = (\EE_2\circ\EE_1)(A)$;
          $(\EE_1\otimes\EE_2)(A_1\otimes A_2) = \EE_1(A_1)\otimes \EE_2(A_2)$.
  \end{itemize}
\end{lemma}
\begin{proof}
  We prove the above properties as follows:
  \begin{itemize}
    \item For $0A = 0$, $1A = A$, $a(bA) = (ab)A$, it follows directly from the definition.
    \item For $0+A = A+0 = A$, $A_1 + A_2 = A_2 + A_1$, $A_1 + (A_2 + A_3) = (A_1 + A_2) + A_3$,
          it follows from the definition of addition of infinite-valued predicates.
          For instance, to prove $A+0 = A$,
          we have, for any $\ket{\psi}$, $\bra{\psi}A+0\ket{\psi}  = \bra{\psi}A\ket{\psi} + \bra{\psi}0\ket{\psi} =  \bra{\psi}A\ket{\psi}$.
          Then the claim follows directly by noting that
          $A_{1} = A_{2}$ if for any $\ket{\psi}$, $\bra{\psi}A_{1}\ket{\psi} = \bra{\psi}A_{2}\ket{\psi}$.

    \item For $0\otimes A = A\otimes 0 = 0$, $A_1\otimes (A_2 \otimes A_3) = (A_1 \otimes A_2) \otimes A_3$,
          it follows from the definition.

    \item For $A\otimes (cA_1 + A_2) = c(A\otimes A_1) + (A\otimes A_2)$;
          $(cA_1 + A_2)\otimes A = c(A_1\otimes A) + (A_2\otimes A)$, it follows from the definition.

    \item To prove $0^\dagger A0 = 0$, we notice that for any $\ket{\psi}$,
          $\bra{\psi} 0^{\dagger}A 0 \ket{\psi} = \bra{\phi} A \ket{\phi} = 0 = \bra{\psi} 0 \ket{\psi}$ for $\ket{\phi} = 0\ket{\psi} = 0$.
          Then we know $0^{\dagger}A 0 = 0$ as we want.
          The propositions
          $M_2^\dagger (M_1^\dagger AM_1) M_2 = (M_1M_2)^\dagger A (M_1M_2)$ and
          $M^\dagger(cA_1+A_2)M = c(M^\dagger A_1M) + M^\dagger A_2M$
          can be proved similarly.

    \item To prove $(M_1\otimes M_2)^\dagger (A_1\otimes A_2) (M_1\otimes M_2) =
          (M_1^\dagger A_1M_1)\otimes (M_2^\dagger A_2M_2)$,
          it follows from the definition and the fact that
          $X_{M_{1}^{\dagger} A_{1} M_{1}} = \supp (M_{1}^{\dagger} X_{A_{1}} M_{1})$.
          Actually,
          $(M_1\otimes M_2)^\dagger (A_1\otimes A_2) (M_1\otimes M_2) =
          (M_{1}^{\dagger} P_{A_{1}} M_{1}) \otimes (M_{2}^{\dagger} P_{A_{2}} M_{2})
          + \supp((M_1\otimes M_2)^\dagger)X(M_{1}\otimes M_{2}))$,
          where $X = \left( \supp \left( P_{A_{1}} \right)\otimes X_{A_{2}}\right)\vee \left( X_{A_{1}} \otimes \supp \left( P_{A_{2}} \right) \right)\vee \left(X_{A_{1}}\otimes X_{A_{2}}\right)$.
          $(M_1^\dagger A_1M_1)\otimes (M_2^\dagger A_2M_2) =
          (M_{1}^{\dagger} P_{A_{1}} M_{1}) \otimes (M_{2}^{\dagger} P_{A_{2}} M_{2})
          +\infty Y$,
          where
          $Y = \left( \supp \left(M_{1}^{\dagger} P_{A_{1}} M_{1} \right)\otimes \supp(M_{2}^{\dagger}X_{A_{2}} M_{2})\right)\vee \left( \supp(M_{1}^{\dagger} X_{A_{1}}M_{1}) \otimes \supp \left( M_{2}^{\dagger} P_{A_{2}}M_{2} \right) \right)\vee \left(\supp(M_{1}^{\dagger}X_{A_{1}} M_{1})\otimes \supp(M_{2}^{\dagger}X_{A_{2}} M_{2})\right)$.
          It can be shown that
          $X = Y$
          by using the properties
          $\mathcal{E}(X_{1}\vee X_{2}) = \mathcal{E}(X_{1})\vee \mathcal{E}(X_{2})$
          and $\mathcal{E}(\supp (\rho)) = \supp (\mathcal{E}(\rho))$.

    \item For $\tr(A(cP_1+P_2)) = c\tr(AP_1) + \tr(AP_2)$,
          if $c = 0$ then it is direct.
          In the following, we assume $c > 0$.
          if $X_{A}P_{1} = X_{A} P_{2} = 0$, then the equation is direct by definition.
          Suppose $X_{A}P_{1} \ne 0$,
          which means $\tr(AP_{1}) = +\infty$.
          In this case, we have
          $X_{A}(cP_{1} + P_{2})\ne 0$, meaning the left hand side is also $+\infty$.
          The case $X_{A}P_{2} \ne 0$ and $\tr((cA_1 + A_2)P) = c\tr(A_1P) + \tr(A_2P)$
          can be proved similarly.

    \item For $\tr((A_1\otimes A_2)(P_1\otimes P_2)) = \tr(A_1P_1)\tr(A_2P_2)$,
          if $X_{A_{1}} P_{1} = X_{A_{2}} P_{2} = 0$, then it is direct by computation.
          Now,
          by symmetry
          consider the case $X_{A_{1}} P_{1}\ne 0$,
          which means $\tr(A_{1} P_{1}) = +\infty$
          (the case $X_{A_{2}}P_{2}\ne 0$ can be proved similarly).
          If $\tr(A_{2} P_{2}) = 0$, then $P_{2}\in (\supp(P_{A_{2}})\vee X_{A_{2}})^{\bot}$.
          In this case,
          $X_{A_{1}\otimes A_{2}}P_{1}\otimes P_{2} = 0$,
          and $\tr((A_{1}\otimes A_{2})(P_{1}\otimes P_{2})) = \tr((P_{A_{1}}\otimes P_{A_{2}})(P_{1}\otimes P_{2})) = 0$.
          If $\tr(A_{2}P_{2})\ne 0$, then
          $X_{A_{1}\otimes A_{2}}P_{1}\otimes P_{2} \ne 0$,
          and both left and right hand sides takes $+\infty$ as we desired.

          For $\tr((M^\dagger AM)P) = \tr(A (MPM^\dagger))$,
          we note that $\supp(M^{\dagger}X_{A} M) P = 0$ is equivalent to $M^{\dagger}\supp(X_{A})M P = 0$,
          and the latter can be written as $\supp(X_{A}) MPM^{\dagger} = 0$.
          Then, the property follows directly from definition.

    \item For $\tr((A\otimes I)P) = \tr(A\tr_2(P))$,
          we note that
          $A\otimes I = P_{A}\otimes I + X_{A}\otimes I$.
          Thus,
          $X_{A} \tr_{2}(P) = 0$
          if and only if $X_{A\otimes I} P = 0$.
          Then, the equation follows directly from the definition.
          $\tr((I\otimes A)P) = \tr(A\tr_1(P))$
          can be proved in a similar way.

    \item For $\tr(A\ket{\phi}\bra{\phi}) = \bra{\phi}A\ket{\phi}$,
          it is direct by definition.

    \item For $A_1=A_2$ iff for all $P\in\Pos$ (or $P\in\DD$) such that $\tr(A_1 P) = \tr(A_2 P)$,
          the ``if'' part can be proved using the previous property
          and the proposition that $A_1 = A_2$ if for all $\ket{\psi}$, $\bra{\psi}A_1\ket{\psi} = \bra{\psi}A_2\ket{\psi}$.
          The ``only if'' part is direct by definition.

    \item For $A_1\sqsubseteq A_2$ iff for all $P\in\Pos$ (or $P\in\DD$) such that $\tr(A_1 P) \le \tr(A_2 P)$,
          the ``if'' part can be proved by limiting $P$ to be rank-$1$ projectors $\ket{\psi}\bra{\psi}$.
          For the ``only if'' part,
          consider the spectral decomposition of $P = \sum_{i} a_{i} \ket{\psi_{i}}\bra{\psi_{i}}$.
          We have
          $\tr(A_{1}P) = \tr(A_{1}\sum_{i} a_{i}\ket{\psi_{i}}\bra{\psi_{i}}) = \sum_{i} a_{i}\bra{\psi_{i}}A_{1}\ket{\psi_{i}}$,
          and
          $\tr(A_{2}P) = \tr(A_{2}\sum_{i} a_{i}\ket{\psi_{i}}\bra{\psi_{i}}) = \sum_{i} a_{i}\bra{\psi_{i}}A_{2}\ket{\psi_{i}}$.
          Since $A_{1}\sqsubseteq  A_{2}$,
          we have $\bra{\psi_{i}}A_{1}\ket{\psi_{i}} \le \bra{\psi_{i}}A_{2}\ket{\psi_{i}} $ for every $i$.
          The result then follows directly.

    \item For $A_1\sqsubseteq A_2$ implies $M^\dagger A_1M\sqsubseteq M^\dagger A_2M$,
          let $\ket{\psi}$ be any state.
          Then,
          a direct computation gives
          $\bra{\psi}M^{\dagger} A_{1} M\ket{\psi} = \bra{\phi} A_{1} \ket{\phi} \le \bra{\phi} A_{2} \ket{\phi}  = \bra{\psi}M^{\dagger} A_{2} M\ket{\psi}$,
          where $\ket{\phi} = M\ket{\psi}$.
          Thus,
          $M^\dagger A_1M\sqsubseteq M^\dagger A_2M$ follows by definition.
          For $A_1\sqsubseteq A_2$ and $A_3\sqsubseteq A_4$ implies $cA_1+A_3\sqsubseteq cA_2+A_4$,
          consider any $\ket{\psi}$,
          $A_{1}\sqsubseteq A_{2}$ implies
          $\bra{\psi}A_{1}\ket{\psi} \le \bra{\psi}A_{2}\ket{\psi}$.
          Similarly we have
          $\bra{\psi}A_{3}\ket{\psi} \le \bra{\psi}A_{4}\ket{\psi}$.
          Therefore
          we know
          $\bra{\psi}cA_{1} + A_{3}\ket{\psi} = c\bra{\psi}A_{1}\ket{\psi} + \bra{\psi}A_{3}\ket{\psi}\le c\bra{\psi}A_{2}\ket{\psi}+ \bra{\psi}A_{4}\ket{\psi}
          = \bra{\psi}cA_{2} + A_{4}\ket{\psi}$,
          and the result follows by definition.

    \item For $\tr(A\EE(P)) = \tr(\EE^\dagger(A) P)$,
          we write $\mathcal{E}(P) = \sum_{j} E_{j} P E_{j}^{\dagger}$.
          Then, we have
          $\tr(A \mathcal{E}(P)) = \tr(A \sum_{j} E_{j} P E_{j}^{\dagger})
          = \sum_{j}\tr(A E_{j} P E_{j}^{\dagger}) = \sum_{j} \tr(E_{j}^{\dagger} A E_{j} P)  = \tr(\mathcal{E}^{\dagger}(A)P)$.
          For
          $A_1\sqsubseteq A_2$ implies $\EE (A_1)\sqsubseteq \EE (A_2)$,
          write $\mathcal{E}(A) = \sum_{j} E_{j}^{\dagger} A E_{j}$.
          We now for any $j$, $E_{j}^{\dagger} A_{1} E_{j} \sqsubseteq E_{j}^{\dagger} A_{2} E_{j}$,
          thus $\sum_{j}E_{j}^{\dagger} A_{1} E_{j} \sqsubseteq \sum_{j}E_{j}^{\dagger} A_{2} E_{j}$ as we want.

    \item For $(c\EE_1 + \EE_2)(A) = c\EE_1(A) + \EE_2(A)$
          take any $P\in \Pos$,
          we have
          $\tr(P(c\EE_1 + \EE_2)(A) ) = \tr((c \mathcal{E}_{1} + \mathcal{E}_{2})^{\dagger}(P) (A))
          c\tr(\mathcal{E}_{1}^{\dagger}(P)A) + \tr(\mathcal{E}_{2}^{\dagger}(P)A)
          = \tr(P c\mathcal{E}_{1}(A) ) + \tr(P \mathcal{E}_{2}(A))= \tr(P(c \mathcal{E}_{1} + \mathcal{E}_{2})(A))$.
          Then the result follows.
          For $\EE(cA_1+A_2) = c\EE(A_1) + \EE(A_2)$,
          write $\mathcal{E}(A) = \sum_{j} E_{j} A E_{j}^{\dagger}$.
          We then have $\mathcal{E}(cA_{1} + A_{2}) = \sum_{j} E_{j} (cA_{1}+A_{2}) E_{j}^{\dagger} = c\sum_{j} E_{j} A_{1} E_{j}^{\dagger} + \sum_{j} E_{j} A_{2} E_{j}^{\dagger}
          = c\mathcal{E}(A_{1}) + \mathcal{E} (A_{2})$ as we want.
    \item $\EE_2(\EE_1(A)) = (\EE_2\circ\EE_1)(A)$
          is by definition.
          For $(\EE_1\otimes\EE_2)(A_1\otimes A_2) = \EE_1(A_1)\otimes \EE_2(A_2)$,
          let $A_{1} = P_{A_{1}} + \infty X_{A_{1}}$,
          $A_{2} = P_{A_{2}} + \infty X_{A_{2}}$,
          and $A_{1} \otimes A_{2} = P_{A_{1}}\otimes P_{A_{2}} + \infty X $,
          where
          $X = \left( \supp \left( P_{A_{1}} \right)\otimes X_{A_{2}}\right)\vee \left( X_{A_{1}} \otimes \supp \left( P_{A_{2}} \right) \right)\vee \left(X_{A_{1}}\otimes X_{A_{2}}\right)$.
          Then,
          $(\EE_1\otimes\EE_2)(A_1\otimes A_2) =  (\EE_1\otimes\EE_2) (P_{A_{1}}\otimes P_{A_{2}}) + \infty Y$,
          where
          $Y = \supp((\EE_1\otimes\EE_2)(X)) $
          and
          $\EE_1(A_1)\otimes \EE_2(A_2) =
          (\EE_1\otimes\EE_2) (P_{A_{1}}\otimes P_{A_{2}}) + \infty Z$,
          where
          $Z = \left( \supp \left( \EE_{1}(P_{A_{1}}) \right)\otimes \EE_{2}(X_{A_{2}})\right)\vee \left( \EE_{1}(X_{A_{1}}) \otimes \supp \left( \EE_{2}(P_{A_{2}}) \right) \right)\vee \left(\EE_{1}(X_{A_{1}})\otimes \EE_{2}(X_{A_{2}})\right)$.
          It is clear that $Y = Z$
          by using the properties
          $\mathcal{E}(X_{1}\vee X_{2}) = \mathcal{E}(X_{1})\vee \mathcal{E}(X_{2})$
          and $\mathcal{E}(\supp (\rho)) = \supp (\mathcal{E}(\rho))$.
  \end{itemize}
\end{proof}
%\input{appendix}

% {\centering\Large \textbf{Supplementary material and deferred proofs}}

% \tw{Try to show this with Proof of Prop 7.1 in the original paper.}

% we could imagine programs $S_1, S_2$ such that if the initial state $\rho$ is in $X$, then the final state coupling $\sigma$ is in $Y$, but if some components of $\rho$ are not in $X$, then $\tr(P_{X^\perp} \rho) \geq \tr(P_{Y^\perp}\sigma)$ doesn't hold. 

% \tw{Continue here.}

%\input{pos_inf}

\end{appendices}

\end{document}